\documentclass[12pt]{article} 

\usepackage{etoolbox}

\newbool{arxiv}

\booltrue{arxiv}

\newbool{blind}

\boolfalse{blind}

\newbool{jrssb}

\boolfalse{jrssb} 

\ifbool{arxiv}{
    \usepackage[style=authoryear,natbib=true]{biblatex}
    \addbibresource{references.bib}
    \usepackage[hidelinks=True]{hyperref}

    \addtolength{\oddsidemargin}{-.5in}%
    \addtolength{\evensidemargin}{-1in}%
    \addtolength{\textwidth}{1in}%
    \addtolength{\textheight}{1.7in}%
    \addtolength{\topmargin}{-1in}%

} {
    \ifbool{jrssb}{
        \onecolumn 
    } {

        \usepackage[style=authoryear,natbib=true]{biblatex}
        \addbibresource{references.bib}
        \usepackage[hidelinks=True]{hyperref}

        \def\spacingset#1{\renewcommand{\baselinestretch}%
        {#1}\small\normalsize} \spacingset{1}        
    }
}

\usepackage{microtype}
\usepackage{graphicx}
\usepackage{subfigure}
\usepackage{booktabs} 
\usepackage{xcolor}
\usepackage{xargs}[2008/03/08]

\usepackage{refstyle}
\usepackage{varioref} 

\usepackage{amsmath}
\usepackage{amssymb}
\usepackage{amsfonts}
\usepackage{amsthm}
\usepackage{mathrsfs} 
\usepackage{mathtools}
\usepackage{bm} 
\usepackage{thmtools} 
\usepackage{accents} 

\usepackage{algorithm}
\usepackage{algpseudocode}
\usepackage{pdfpages}

\usepackage{setspace}

\usepackage[disable]{todonotes}

\newbool{showlabels}

\boolfalse{showlabels}

\ifbool{showlabels} {
    \usepackage[inline]{showlabels}
        
}

\ifbool{arxiv}{

    \newbool{largespacing}

    \boolfalse{largespacing}
    \ifbool{largespacing}{%
        \doublespacing
        \usepackage{showkeys}
    }
    {

    }

}

\usepackage[]{graphicx}
\usepackage[]{color}
\makeatletter
\def\maxwidth{ %
  \ifdim\Gin@nat@width>\linewidth
    \linewidth
  \else
    \Gin@nat@width
  \fi
}
\makeatother

\definecolor{fgcolor}{rgb}{0.345, 0.345, 0.345}
\newcommand{\hlkwb}[1]{\textcolor[rgb]{0.69,0.353,0.396}{#1}}%

\usepackage{framed}
\makeatletter
\newenvironment{kframe}{%
 \def\at@end@of@kframe{}%
 \ifinner\ifhmode%
  \def\at@end@of@kframe{\end{minipage}}%
  \begin{minipage}{\columnwidth}%
 \fi\fi%
 \def\FrameCommand##1{\hskip\@totalleftmargin \hskip-\fboxsep
 \colorbox{shadecolor}{##1}\hskip-\fboxsep
     \hskip-\linewidth \hskip-\@totalleftmargin \hskip\columnwidth}%
 \MakeFramed {\advance\hsize-\width
   \@totalleftmargin\z@ \linewidth\hsize
   \@setminipage}}%
 {\par\unskip\endMakeFramed%
 \at@end@of@kframe}
\makeatother

\definecolor{shadecolor}{rgb}{.97, .97, .97}
\definecolor{messagecolor}{rgb}{0, 0, 0}
\definecolor{warningcolor}{rgb}{1, 0, 1}
\definecolor{errorcolor}{rgb}{1, 0, 0}
\newenvironment{knitrout}{}{} 

\usepackage{alltt}

\def\ind#1{\mathbb{I}\left(#1\right)}
\def\abs#1{\left\vert #1\right\vert }
\def\norm#1{\left\Vert #1\right\Vert }

\def\fracat#1#2#3{\left.\frac{#1}{#2}\right\vert_{#3}}
\def\iid{\overset{iid}{\sim}}

\DeclareMathOperator*{\argmax}{\mathrm{argmax}}

\def\trans{\intercal}   
\def\diag#1{\mathrm{diag}\left(#1\right)}
\def\trace#1{\mathrm{tr}\left(#1\right)}

\def\plim{\xrightarrow[N\rightarrow\infty]{prob}\,}
\def\dlim{\xrightarrow[N\rightarrow\infty]{dist}\,}

\def\expect#1#2{\underset{#1}{\mathbb{E}}\left[#2\right]}
\def\cov#1#2{\underset{#1}{\mathrm{Cov}}\left(#2\right)}
\def\var#1#2{\underset{#1}{\mathrm{Var}}\left(#2\right)}
\def\covhat#1#2{\underset{#1}{\widehat{\mathrm{Cov}}}\left(#2\right)}

\def\p{\mathbb{P}}
\def\post{\p(\theta | \xvec)}
\def\postf{\p(\theta | \fdist, N)}
\def\postfn{\p(\theta | \fndist, N)}
\def\postft{\p(\theta | \ftdist, N)}
\def\postw{\p(\theta | \xvec, \w)}
\def\postwtil{\p(\theta | \xvec, \wtil)}
\def\postpert#1{\p(\theta | #1)}

\def\gl{\gamma, \lambda}

\def\postglf{\p(\gamma, \lambda | \fdist, N)}

\def\postl{\p(\lambda | \xvec, \gamma)}

\def\postg{\p(\gamma | \xvec)}
\def\postglt{\p(\gamma, \lambda | \ftdist, N)}
\def\postlt{\p(\lambda | \ftdist, N, \gamma)}
\def\postlf{\p(\lambda | \fdist, N, \gamma)}
\def\postgt{\p(\gamma | \ftdist, N)}
\def\postgf{\p(\gamma | \fdist, N)}

\def\fdist{\mathbb{F}}
\def\xfdist{\xvec\iid\fdist}
\def\fdistnm{\fdist(\x_n) \fdist(\x_m)}

\def\prior{\pi}
\def\logprior{\log\pi}

\def\normdist{\mathcal{N}}
\def\normhatarg#1{\hat{\mathcal{G}}(#1)}
\def\normresid#1#2{\mathcal{G}_{#1}^{#2}}
\def\normhat{\normhatarg{\tau}}
\def\normalizer{\hat{Z}}  

\def\gdist{\mathbb{G}}
\def\fndist{\mathbb{F}_N}
\def\ftdist{\mathbb{F}_N^{\t}}

\def\sumn{\sum_{n=1}^N}
\def\summ{\sum_{m=1}^N}
\def\meann{\frac{1}{N}\sum_{n=1}^N}
\def\meanm{\frac{1}{N}\sum_{m=1}^N}
\def\sumnm{\sum_{n=1}^N \sum_{m=1}^N}
\def\meannm{\frac{1}{N^2} \sumnm}
\def\indexset{\mathcal{S}}
\def\means{\frac{1}{|\indexset|}\sum_{s \in \indexset}}
\def\sumg{\sum_{g=1}^G}

\def\g{g}   
\def\r{r}   
\def\y{y}   
\def\x{x}   
\def\zbar{\underline{z}}   
\def\a{a}   
\def\xn{x_0}  
\def\xvec{X}    
\def\z{z}   
\def\w{w}   
\def\wtil{\tilde{w}}   
\def\T{T}   
\def\t{t}   
\def\ttil{\tilde{\t}}   
\def\onevec{1_N}    
\newcommand\htil[1][\t]{r^{#1}}    
\def\v{v}   
\def\m{m}   

\def\thetadom{\Omega_\theta}
\def\gammadom{\Omega_\gamma}
\def\lambdadom{\Omega_\lambda}
\def\rdom#1{\mathbb{R}^{#1}}
\def\thetaball#1{B_{#1}}
\def\gammadim{D_{\gamma}}   
\def\ydim{D_{y}}   
\def\lambdadim{D_{\lambda}} 
\def\gdim{D_g}  
\def\thetadim{D_\theta}  
\def\ng{\mathscr{N}_g}  

\def\deltaulln{\delta_{LLN}}

\def\info{\mathcal{I}}
\def\infoev{\lambda_\mathcal{I}}
\def\infoevhat{\hat{\lambda}_\mathcal{I}}
\def\infohat{\hat{\mathcal{I}}}

\def\scorecov{\Sigma}
\def\scorecovhat{\hat\Sigma}
\def\thetatrue{\accentset{\infty}{\theta}}

\def\thetahat{\hat\theta}

\def\thetatil{\tilde{\theta}}   

\def\etabar{\overline{\eta}}   
\def\mubar{\overline{\mu}}   
\def\L{\mathbb{L}} 
\def\M{\mathbb{M}} 
\def\S{\mathbb{S}} 
\def\J{\mathbb{J}} 

\def\ghat{\hat{\g}} 

\def\resid#1{\mathscr{R}^{#1}}

\def\tresid#1#2{\mathscr{T}_{(#1)}\left(#2\right)}

\def\regevent{\mathfrak{R}}

\def\gmeantrue{\mu^g}   
\def\gcovtrue{V^g}
\def\gcovmaphat{\hat{V}^{\mathrm{MAP}}}
\def\gcovhat{\hat{V}^{g}}
\def\gcovij{V^{\mathrm{IJ}}}
\def\gcovijhat{\hat{V}^{\mathrm{IJ}}}
\def\gcovboothat{\hat{V}^{\mathrm{Boot}}}
\def\gcovboot{V^{\mathrm{Boot}}}
\def\gcovbayeshat{\hat{V}^{\mathrm{Bayes}}}
\def\gcovbayes{V^{\mathrm{Bayes}}}

\def\infl{\psi}
\def\infltil{\psi^{0}}

\def\ggrad#1{g_{(#1)}}
\def\psigrad#1{\psi_{(#1)}}
\def\ellgrad#1{\ell_{(#1)}}
\def\phigrad#1{\phi_{(#1)}}
\def\psigrad#1{\psi_{(#1)}}

\def\ord#1{O\left(#1\right)}
\def\ordlog#1{\tilde{O}\left(#1\right)}
\def\ordlogp#1{\tilde{O}_p\left(#1\right)}
\def\ordp#1{O_p\left(#1\right)}
\def\littleop#1{o_p\left(#1\right)}

\def\lik{\mathscr{L}}

\def\likhat{\hat{\mathscr{L}}}
\def\likk#1{\mathscr{L}_{(#1)}}
\def\likhatk#1{\hat{\mathscr{L}}_{(#1)}}

\def\gbar{\bar{\g}}

\def\ghatgrad#1{\g_{(#1)}}

\def\ellbar{\bar{\ell}}
\def\ellunderbar{\underline{\ell}}
\def\ellbarbar{\underline{\ellbar}}

\def\ellunderbargrad#1{\ellunderbar_{(#1)}}

\def\ellhatunderbar{\hat{\underline{\ell}}} 
\def\ellhatunderbargrad#1{\ellhatunderbar_{(#1)}}

\def\phibar{\underline{\phi}}

\def\seboot{\Xi^{\textrm{boot}}}

\def\seij{\Xi^{\textrm{IJ}}}
\def\sebayes{\Xi^{\textrm{bayes}}}
\def\zdiff{Z}
\def\normdiffij{\Delta^{\mathrm{IJ}}}
\def\normdiffbayes{\Delta^{\mathrm{Bayes}}}

\declaretheoremstyle[
  qed=$\blacktriangleleft$
]{mythmstyle}
\declaretheorem[style=mythmstyle,name=Example]{ex}
\declaretheorem[style=mythmstyle,name=Definition]{defn}

\theoremstyle{plain}
\newtheorem{lem}{Lemma}
\newtheorem{thm}{Theorem}
\newtheorem{prop}{Proposition}

\newtheorem{assu}{Assumption}
\newtheorem{cor}{Corollary}

\theoremstyle{definition}

\newref{event}{
    name=Event~, %
    names=Events~, %
    Name=Event~,
    Names=Events~
    }

\newref{item}{
    name=Item~, %
    names=Items~, %
    Name=Item~
    Names=Items~
    }

\newref{fig}{
    name=Figure~, %
    Name=Figure~
    }

\newref{tab}{
    name=Table~, %
    Name=Table~
    }

\newref{sec}{
    name=Section~, %
    Name=Section~
    }

\newref{app}{
    name=Appendix~, %
    Name=Appendix~
    }

\newref{eq}{
    name=Eq.~, %
    Name=Eq.~,
    names=Eqs.~, %
    Names=Eqs.~
    }

\newref{fig}{
    name=Figure~, %
    Name=Figure~
    }

\newref{def}{
    name=Definition~, %
    Name=Definition~
    }

\newref{assu}{
    name=Assumption~, %
    Name=Assumption~,
    names=Assumptions~, %
    Names=Assumptions~,
    }

\newref{cond}{
    name=Condition~, %
    Name=Condition~,
    names=Conditions~, %
    Names=Conditions~
    }

\newref{prop}{
    name=Proposition~, %
    Name=Proposition~,
    names=Propositions~, %
    Names=Propositions~
    }

\newref{lem}{
    name=Lemma~, %
    Name=Lemma~,
    names=Lemmas~, %
    Names=Lemmas~
    }

\newref{ex}{
    name=Example~, %
    Name=Example~}

\newref{cor}{
    name=Corollary~, %
    Name=Corollary~
    }

\newref{thm}{
    name=Theorem~, %
    Name=Theorem~,
    names=Theorems~, %
    Names=Theorems~
    }

\newref{proof}{
    name=Proof~, %
    Name=Proof~
    }

\newref{conj}{
    name=Conjecture~, %
    Name=Conjecture~
    }

\newref{algr}{
    name=Algorithm~, %
    Name=Algorithm~,
    names=Algorithms~, %
    Names=Algorithms~,
    }

\newcommand{\armNumModels}{56}
\newcommand{\armNumDatasets}{14}
\newcommand{\armNumCovsEstimated}{768}
\newcommand{\armNumMCMCSamples}{2,000}
\newcommand{\armNumBootstraps}{200}
\newcommand{\armMedianNumObs}{240}
\newcommand{\armMinNumObs}{40}
\newcommand{\armMaxNumObs}{3,020}
\newcommand{\armMinMCMCTimeSecs}{1.4}
\newcommand{\armMaxMCMCTimeMins}{42}
\newcommand{\armTotalMCMCTimeMins}{58}
\newcommand{\armTotalBootTimeHours}{267}
\newcommand{\armPilotMCMCTimeSecs}{6.2}
\newcommand{\armPilotBootTimeMins}{47}
\newcommand{\armPilotNumObs}{40}

\newcommand{\armPilotNumGroups}{5}
\newcommand{\armPilotNumScenarios}{8}
\newcommand{\armPilotNumBoots}{200}

\newcommand{\batsNumObs}{181}
\newcommand{\batsNumMCMCDraws}{4,000}
\newcommand{\batsNumTimes}{19}
\newcommand{\batsMCMCTime}{8}
\newcommand{\batsBootTime}{1,403}
\newcommand{\batsBootOverMCMCTime}{186}
\newcommand{\batsNumSEBlocks}{100}
\newcommand{\batsNumSEDraws}{200}
\newcommand{\batsNumBoots}{200}
\newcommand{\batsMeanMeanP}{0.747}
\newcommand{\batsMeanMeanPhi}{0.741}
\newcommand{\batsMeanLogSigma}{-0.495}
\newcommand{\batsOptMeanP}{0.0000191}
\newcommand{\batsOptMeanPhi}{1}

\newcommand{\electionNumObs}{361}
\newcommand{\electionNumPollsters}{44}
\newcommand{\electionNumMCMCDraws}{3,000}
\newcommand{\electionNumBoots}{100}

\newcommand{\electionMCMCHours}{27}
\newcommand{\electionBootHours}{1,231}
\newcommand{\electionBootOverMCMC}{45}
\newcommand{\electionWorstState}{DC}
\newcommand{\electionBestState}{OH}
\newcommand{\reNumSims}{300}
\newcommand{\reNumBoots}{300}

\newcommand{\reZPriorMean}{10}
\newcommand{\reZPriorSD}{2}

\def\ARMTable{
\begin{table}[h]
\begin{center}
\begin{tabular}{|l|c|c|}
  \hline
 & Fixed effects only & Random and fixed effects \\ 
  \hline
Linear regression & 398 & 168 \\ 
   \hline
Logistic regression & 192 &  10 \\ 
   \hline
\end{tabular}

Models

\vspace{0.1in}
\begin{tabular}{|l|c|c|}
  \hline
 & Cross-covariance estimate & Variance estimate \\ 
  \hline
Log scale parameter & 191 &  62 \\ 
   \hline
Regression parameter & 325 & 190 \\ 
   \hline
\end{tabular}

Covariances to estimate

\end{center}

\caption{Models from \cite{gelman:2006:arm}\tablabel{arm_models}}
\end{table}
}

\def\ARMGraphZ{

\begin{knitrout}
\definecolor{shadecolor}{rgb}{0.969, 0.969, 0.969}\color{fgcolor}\begin{figure}[!h]

{\centering \includegraphics[width=0.98\linewidth,height=0.784\linewidth]{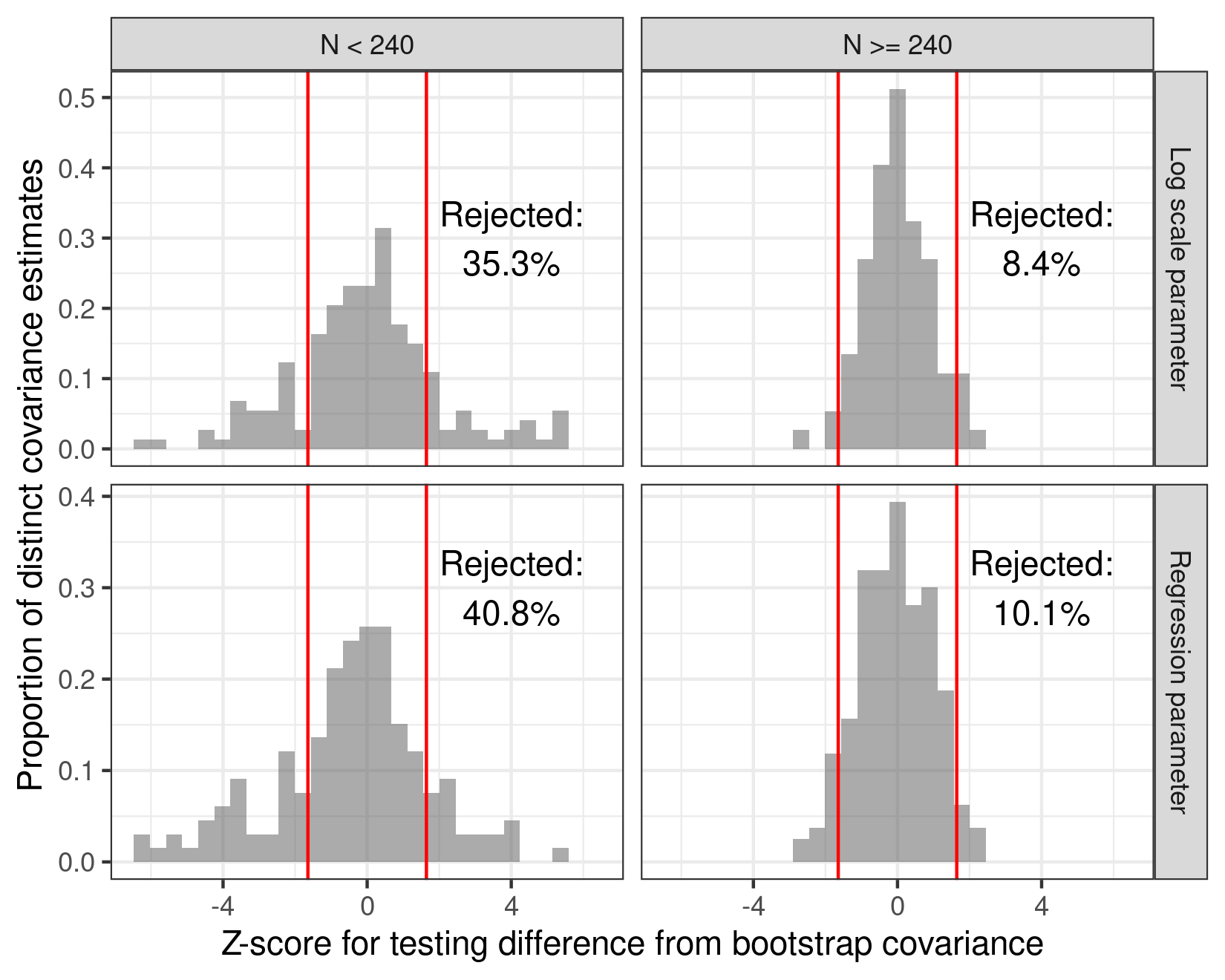} 

}

\caption[The distribution of the z-statistics $\zdiff$]{The distribution of the z-statistics $\zdiff$.  Red lines indicate the boundaries of a normal test for significance with level $0.1$, and ``Rejected'' counts the number of covariances in the rejection region.  Log scale parameters include all variances or covariances that involve at least one log scale parameters.}\label{fig:relerr_graph}
\end{figure}

\end{knitrout}
}

\def\ARMGraphDiff{

\begin{knitrout}
\definecolor{shadecolor}{rgb}{0.969, 0.969, 0.969}\color{fgcolor}\begin{figure}[!h]

{\centering \includegraphics[width=0.98\linewidth,height=0.784\linewidth]{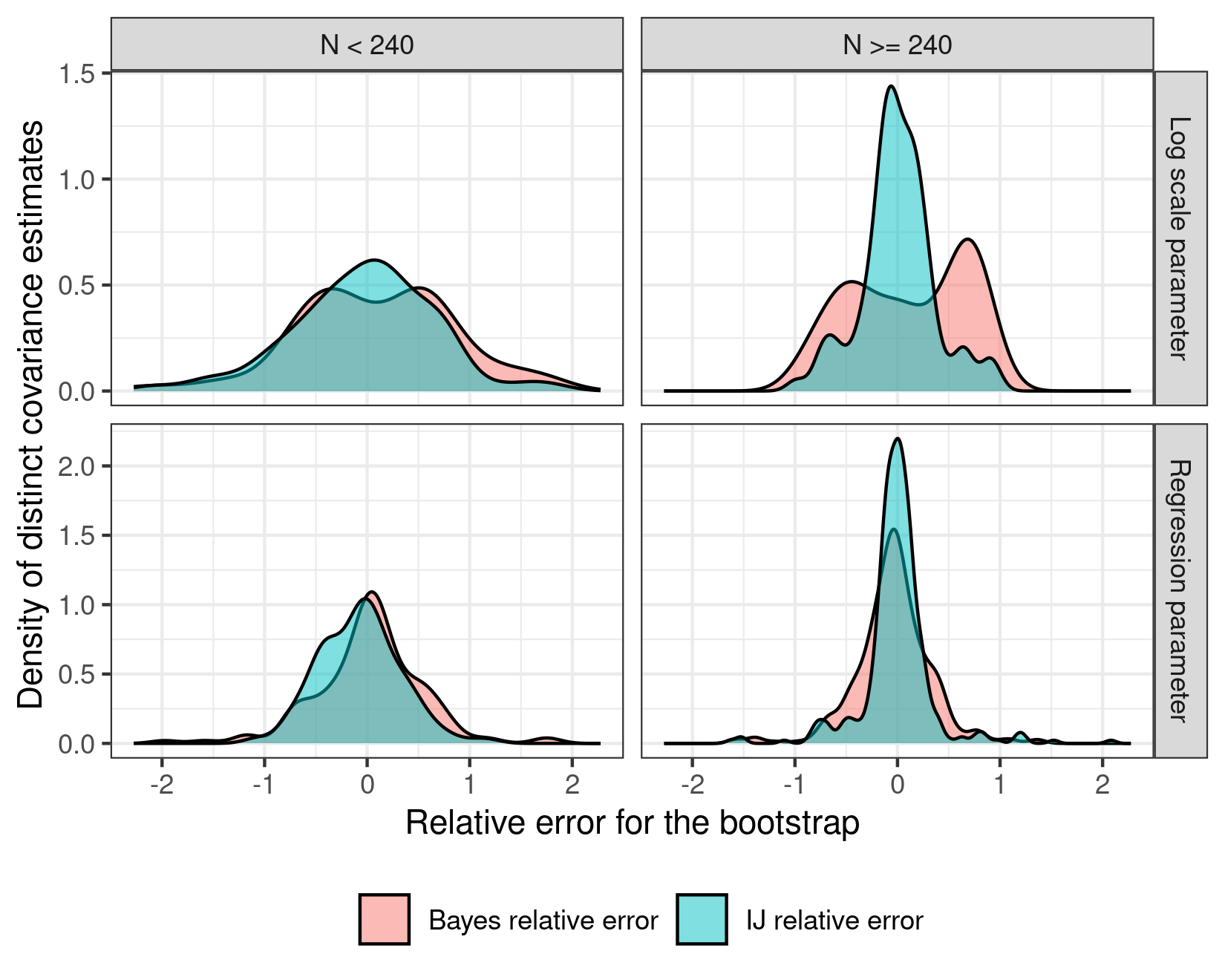} 

}

\caption[The distribution of the relative errors $\normdiffij$ and$\normdiffbayes$.Log scale parameters include all variances or covariances that involve at least one log scale parameters]{The distribution of the relative errors $\normdiffij$ and$\normdiffbayes$.Log scale parameters include all variances or covariances that involve at least one log scale parameters.}\label{fig:normerr_graph}
\end{figure}

\end{knitrout}
}

\def\PilotsIntervalsGraph{

\begin{knitrout}
\definecolor{shadecolor}{rgb}{0.969, 0.969, 0.969}\color{fgcolor}\begin{figure}[!h]

{\centering \includegraphics[width=0.98\linewidth,height=0.784\linewidth]{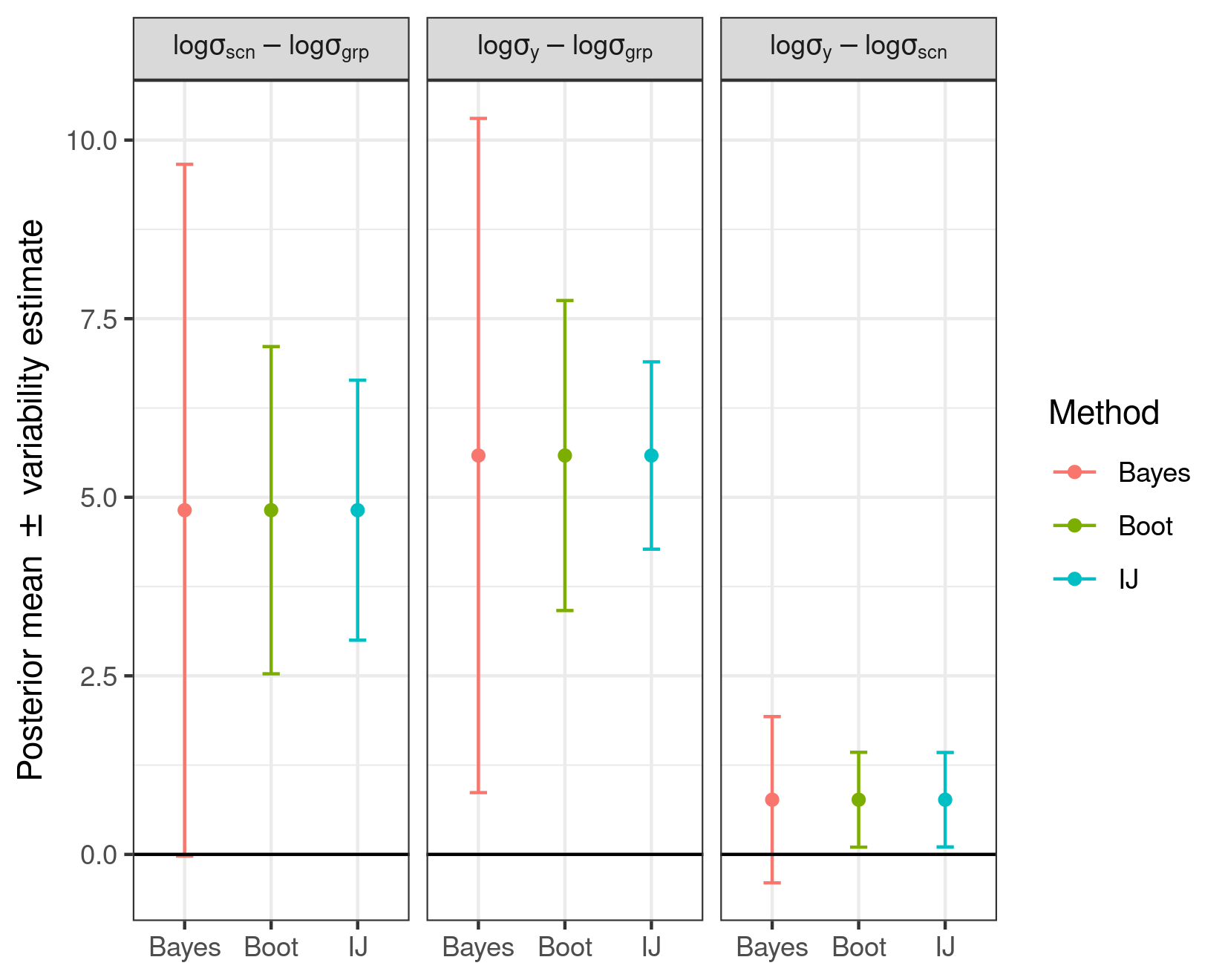} 

}

\caption[Means and standard errors for the pilots data]{Means and standard errors for the pilots data. Note that one may make different decisions about the evidence when using the frequentist versus the Bayesian credible intervals, particularly in the first panel.}\label{fig:interval_graph}
\end{figure}

\end{knitrout}
}

\def\PilotsSEGraph{

\begin{knitrout}
\definecolor{shadecolor}{rgb}{0.969, 0.969, 0.969}\color{fgcolor}\begin{figure}[!h]

{\centering \includegraphics[width=0.98\linewidth,height=0.784\linewidth]{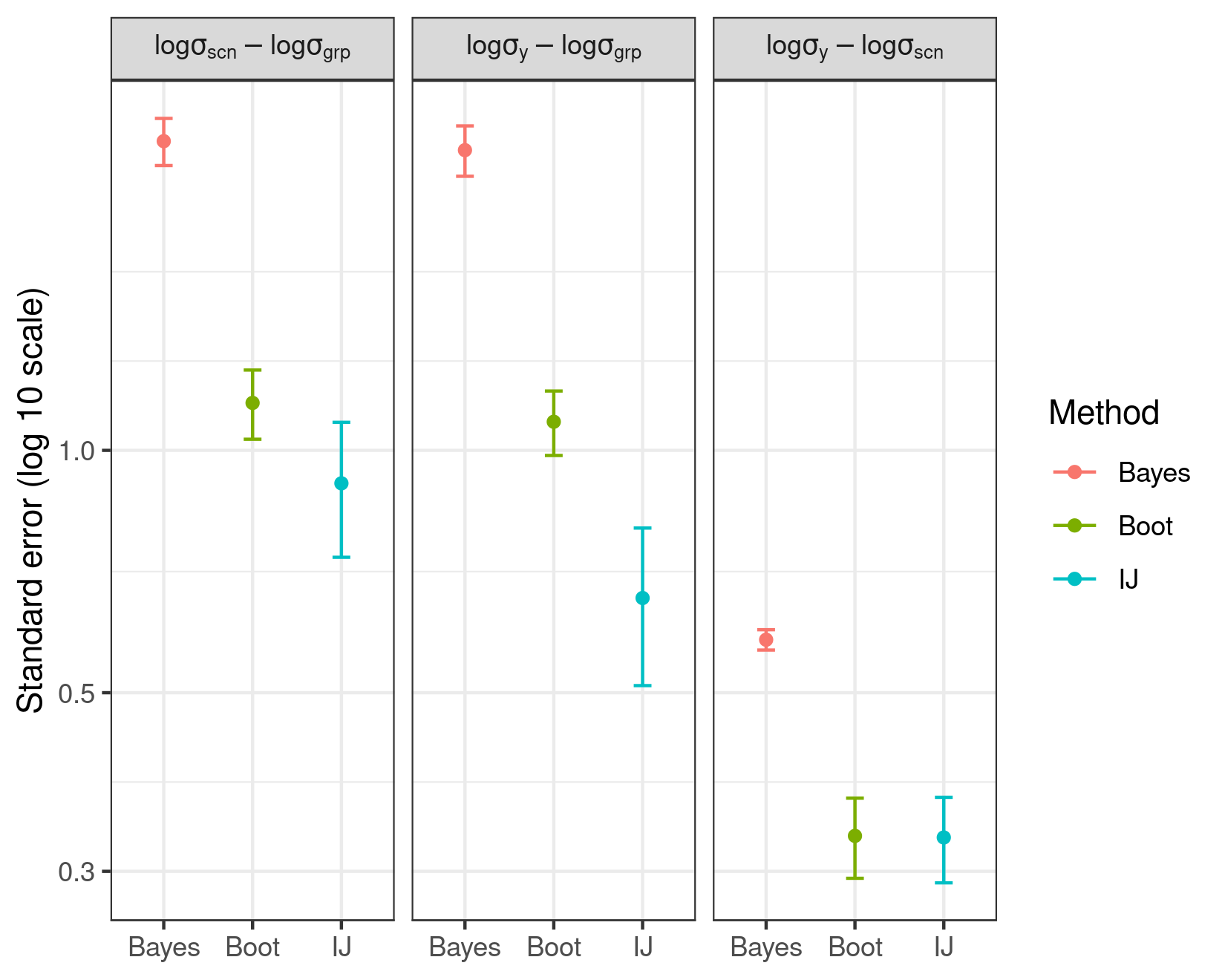} 

}

\caption[Standard errors (square root of variance estimates / $N$) for the pilots data]{Standard errors (square root of variance estimates / $N$) for the pilots data.}\label{fig:width_graph}
\end{figure}

\end{knitrout}
}

\def\BatsData{

\begin{knitrout}
\definecolor{shadecolor}{rgb}{0.969, 0.969, 0.969}\color{fgcolor}\begin{figure}[!h]

{\centering \includegraphics[width=0.98\linewidth,height=0.549\linewidth]{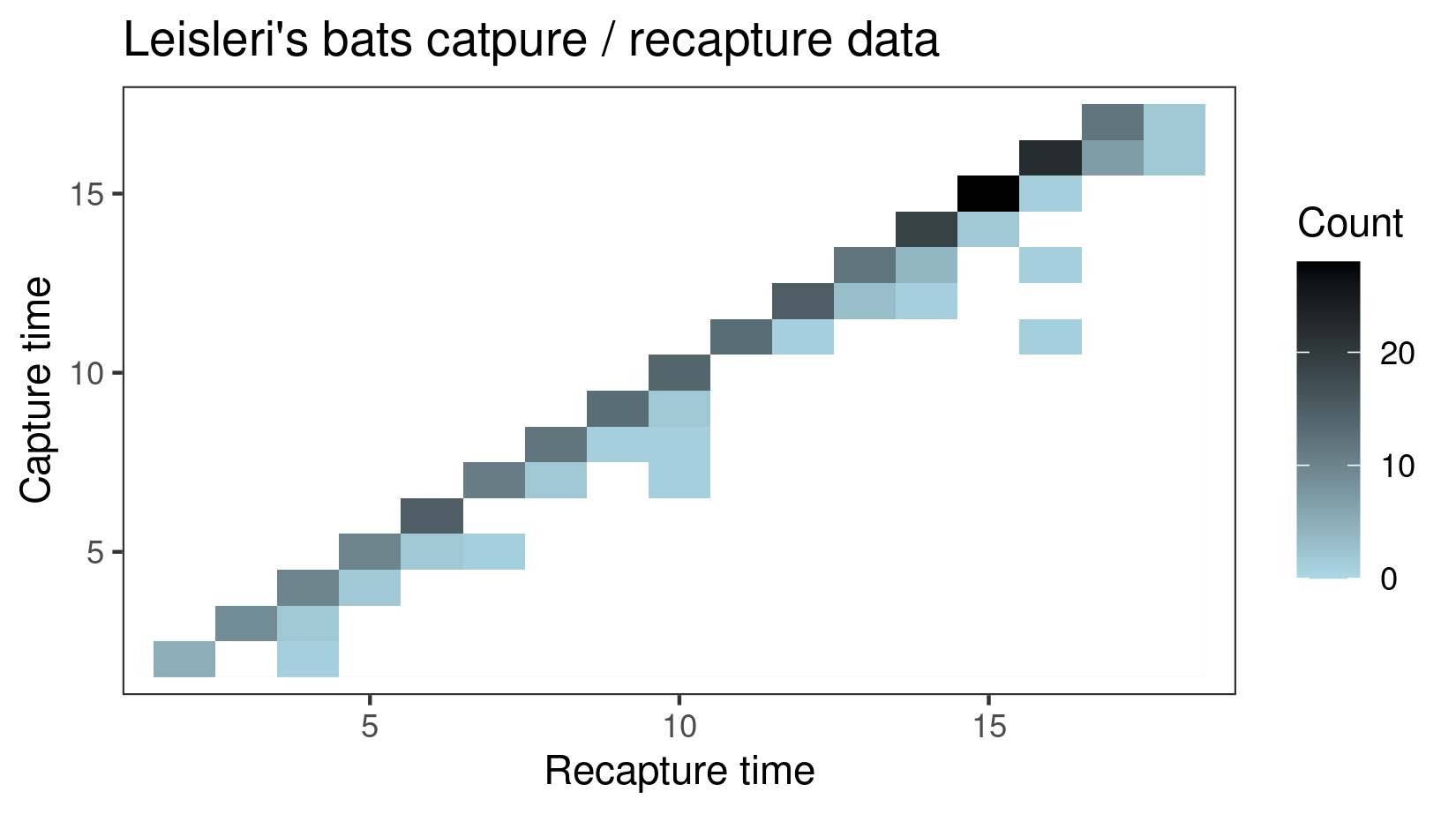} 

}

\caption[The ragged array counting the number of bats at re-observed a certain number of time periods after each time period]{The ragged array counting the number of bats at re-observed a certain number of time periods after each time period.}\label{fig:bats_data}
\end{figure}

\end{knitrout}
}

\def\BatsResults{

\begin{knitrout}
\definecolor{shadecolor}{rgb}{0.969, 0.969, 0.969}\color{fgcolor}\begin{figure}[!h]

{\centering \includegraphics[width=0.98\linewidth,height=0.784\linewidth]{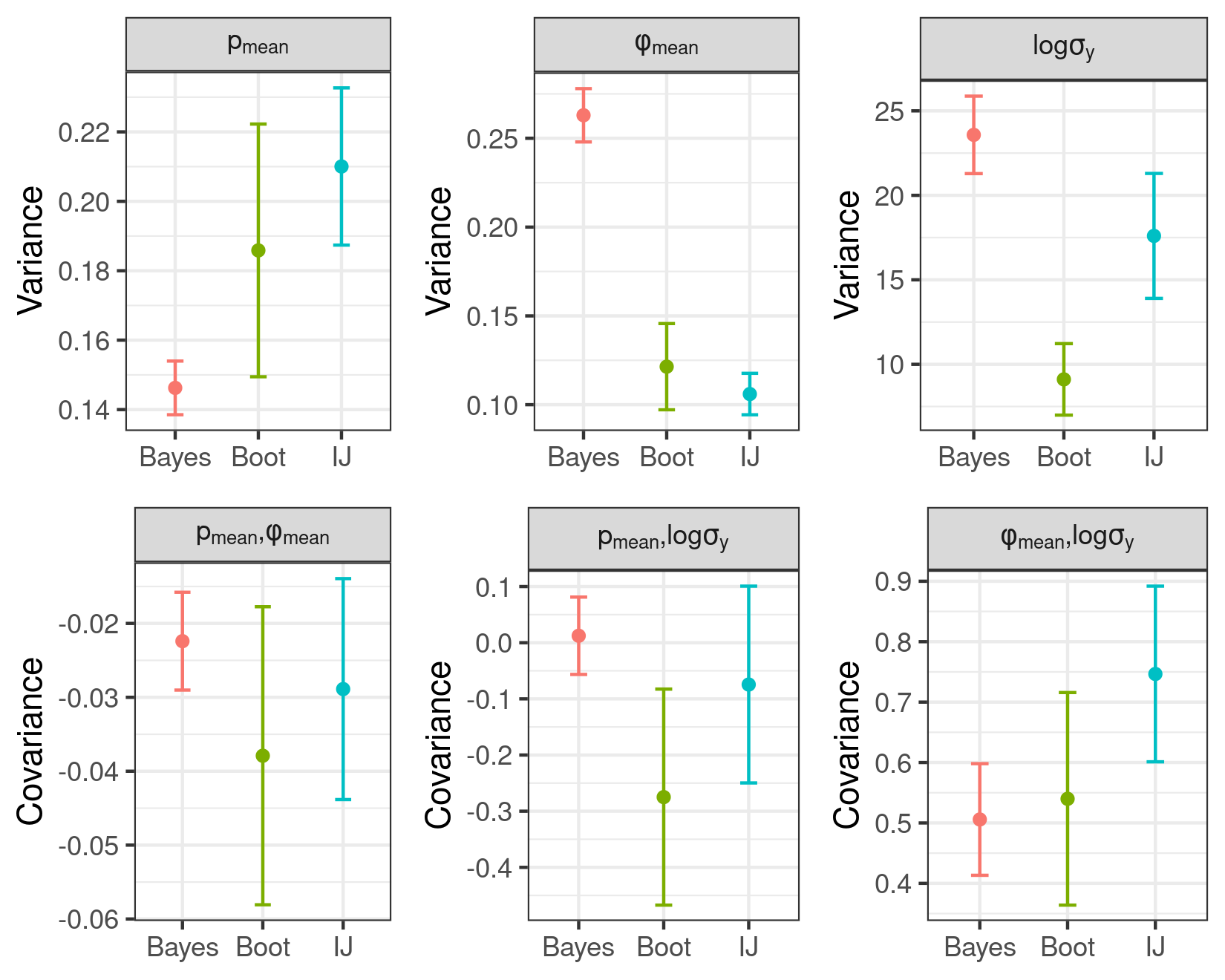} 

}

\caption[Comparison between $\gcovbayeshat$, $\gcovijhat$, and $\gcovboothat$ for the Leisleri Bats dataset]{Comparison between $\gcovbayeshat$, $\gcovijhat$, and $\gcovboothat$ for the Leisleri Bats dataset.  The first row shows the diagonal entries (variances), and the second row shows the off-diagonal entries (covariances).  Error bars are $2 \sebayes$, $2 \seij$, and $2 \seboot$, respectively.  Note that the scales of the y axes are all different to allow for joint plotting.}\label{fig:bats_result}
\end{figure}

\end{knitrout}
}

\def\POTUSStatesGraph{

\begin{knitrout}
\definecolor{shadecolor}{rgb}{0.969, 0.969, 0.969}\color{fgcolor}\begin{figure}[!h]

{\centering \includegraphics[width=0.98\linewidth,height=0.549\linewidth]{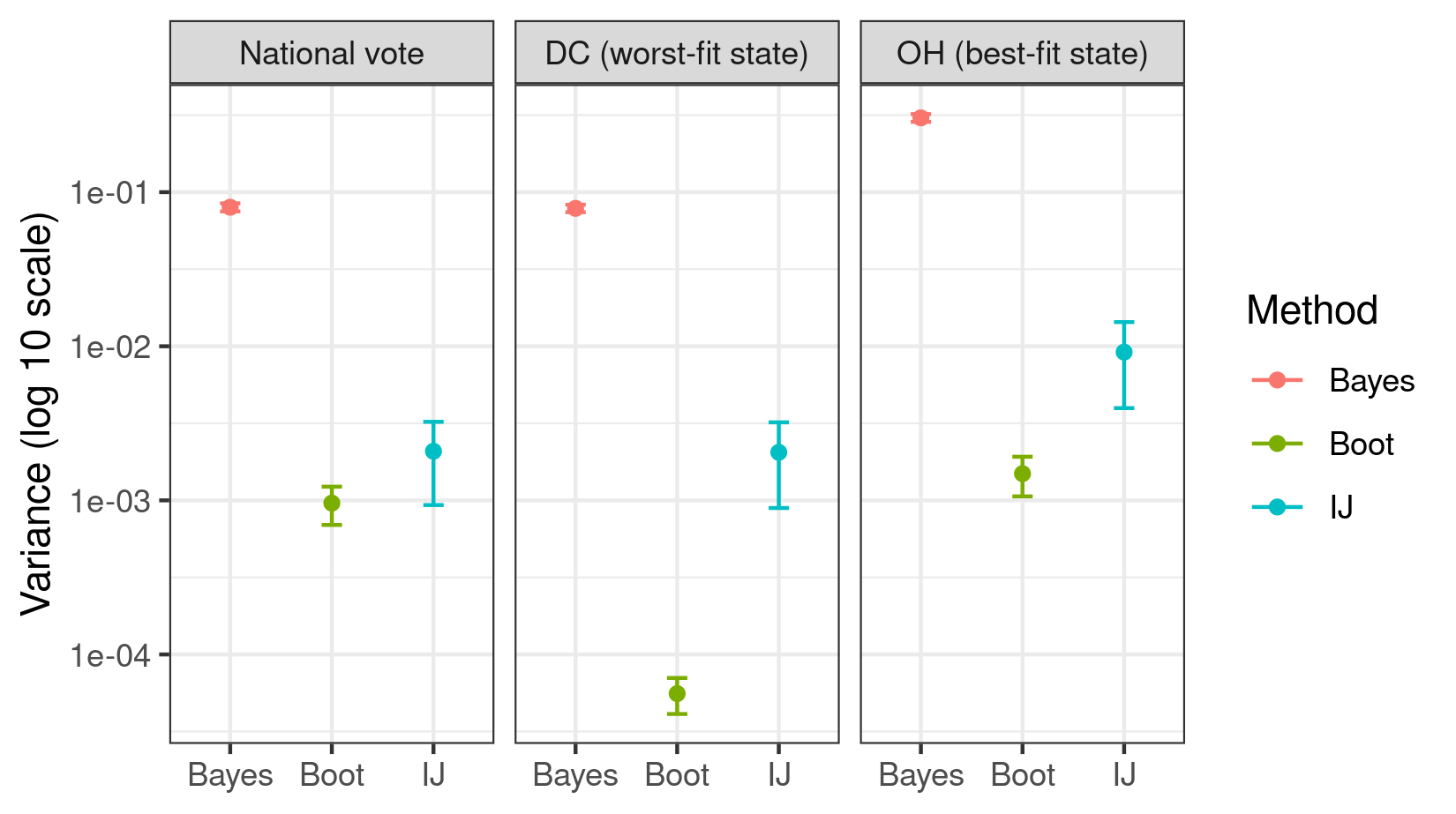} 

}

\caption[Selected for the national vote and the two states in which the IJ performed worst and best]{Selected for the national vote and the two states in which the IJ performed worst and best.  Note the log10 scale of the y-axis.  The upper and lower extent of the error bars were computed using $2 \seij$, $2 \sebayes$, and $2 \seboot$ before transforming by log10 for plotting.}\label{fig:election_result}
\end{figure}

\end{knitrout}
}

\def\POTUSScatterGraph{

\begin{knitrout}
\definecolor{shadecolor}{rgb}{0.969, 0.969, 0.969}\color{fgcolor}\begin{figure}[!h]

{\centering \includegraphics[width=0.98\linewidth,height=0.549\linewidth]{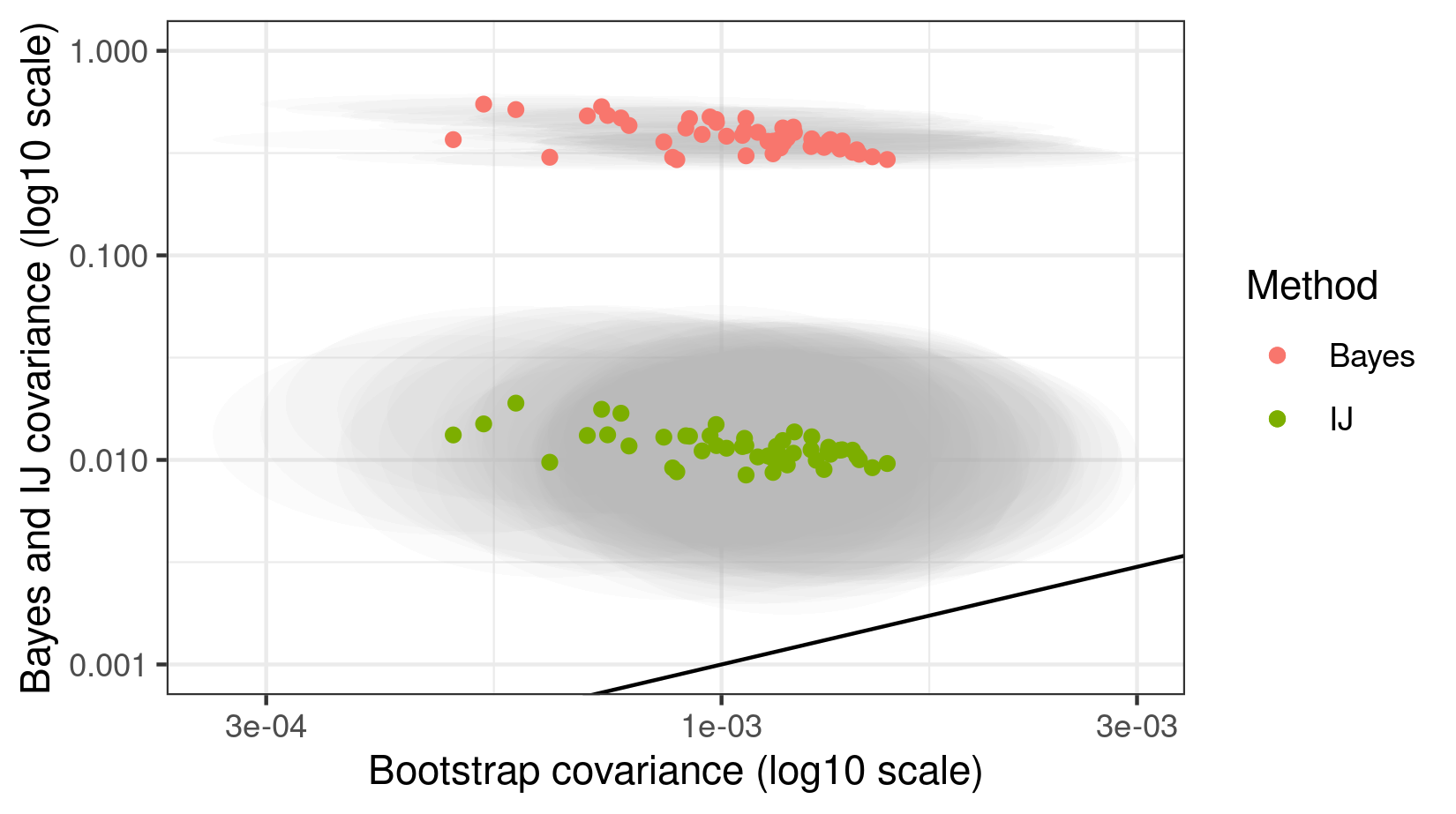} 

}

\caption[Scatterplot of the diagonals of $\gcovbayeshat$, $\gcovijhat$, and $\gcovboothat$ for all the states except DC]{Scatterplot of the diagonals of $\gcovbayeshat$, $\gcovijhat$, and $\gcovboothat$ for all the states except DC.  Note the log10 scale of the y-axis.  The extent of the confidence ellipses (shown in gray) were computed using $2 \seij$, $2 \sebayes$, and $2 \seboot$ before transforming by log10 for plotting.}\label{fig:election_state_result}
\end{figure}

\end{knitrout}
}

\newcommand{\PoissonREGraph}{

\begin{knitrout}
\definecolor{shadecolor}{rgb}{0.969, 0.969, 0.969}\color{fgcolor}\begin{figure}[!h]

{\centering \includegraphics[width=0.98\linewidth,height=0.980\linewidth]{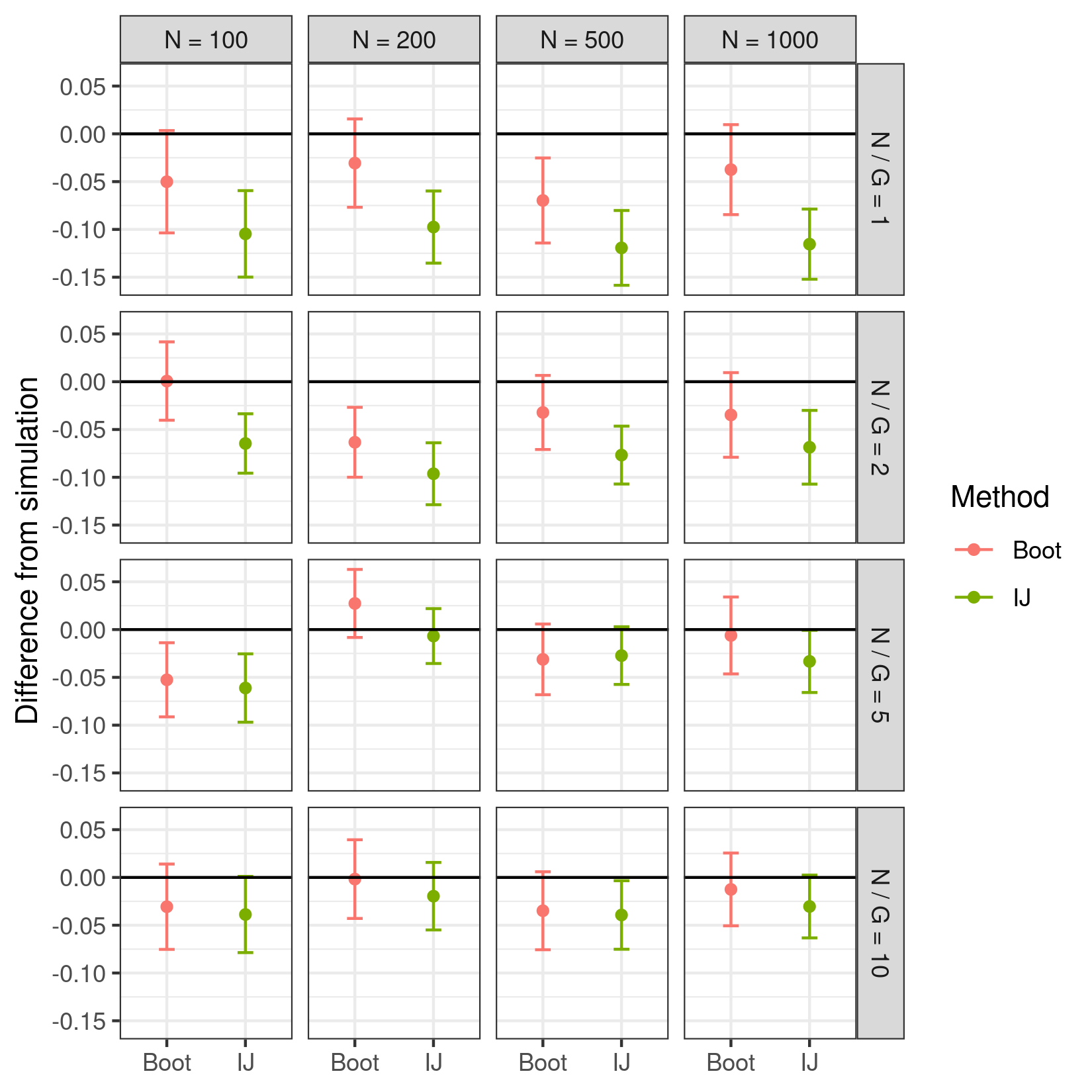} 

}

\caption[The error of the IJ and bootstrap covariances for different values of $N$ and $G$]{The error of the IJ and bootstrap covariances for different values of $N$ and $G$.  The y-axis shows the difference between $N (V - \hat{V}_{\mathrm{sim}})$, where $V$ is either $\gcovijhat$ or $\gcovboothat$.}\label{fig:poisson_re_graph}
\end{figure}

\end{knitrout}
}

\begin{document}

\title{The Bayesian Infinitesimal Jackknife for Variance}

\ifbool{arxiv}{
    \author{
        Ryan Giordano \\ \texttt{rgiordano@berkeley.edu}
        \and
        Tamara Broderick \\ \texttt{tbroderick@mit.edu}
      }
} {
    \ifbool{jrssb} {

    \journaltitle{JRSS-B}
    \DOI{DOI HERE}
    \copyrightyear{2023}
    \pubyear{2024}
    \access{Advance Access Publication Date: Day Month Year}
    \appnotes{Paper}
    
    \firstpage{1}

    \author[1]{Ryan Giordano \ORCID{0000-0002-5686-9210}}
    \author[2]{Tamara Broderick \ORCID{0000-0003-4704-5196}}

    \authormark{Giordano \& Broderick}
     
    \address[1]{\orgdiv{Statistics}, \orgname{University of California Berkeley}, 
                \orgaddress{\state{CA}, \country{USA}}}
    \address[2]{\orgdiv{EECS}, \orgname{Massachusetts Institute of Technology}, 
                \orgaddress{\state{MA}, \country{USA}}}

    \abstract{The frequentist variability of Bayesian posterior expectations can provide
meaningful measures of uncertainty even when models are misspecified. Classical
methods to asymptotically approximate the frequentist covariance of Bayesian
estimators such as the Laplace approximation and the nonparametric bootstrap can
be practically inconvenient, since the Laplace approximation may require an
intractable integral to compute the marginal log posterior, and the bootstrap
requires computing the posterior for many different bootstrap datasets. We
develop and explore the infinitesimal jackknife (IJ), an alternative method for
computing asymptotic frequentist covariance of smooth functionals of
exchangeable data, which is based on the ``influence function'' of robust
statistics.  We show that the influence function for posterior expectations has
the form of a simple posterior covariance, and that the IJ covariance estimate
is, in turn, easily computed from a single set of posterior samples. Under
conditions similar to those required for a Bayesian central limit theorem to
apply, we prove that the corresponding IJ covariance estimate is asymptotically
equivalent to the Laplace approximation and the bootstrap.  In the presence of
nuisance parameters that may not obey a central limit theorem, we argue using a
von Mises expansion that the IJ covariance is inconsistent, but can remain a
good approximation to the limiting frequentist variance.  We demonstrate the
accuracy and computational benefits of the IJ covariance estimates with
simulated and real-world experiments.
}


    } {
        \title{The Bayesian Infinitesimal Jackknife for Variance}
        \ifbool{blind} { 
            \author{(ANONYMIZED)}
        } {
            \author{
                Ryan Giordano \\ \texttt{rgiordan@mit.edu}
                \and
                Tamara Broderick \\ \texttt{tbroderick@mit.edu}
              }        
        }
    }
}

\maketitle

\ifbool{jrssb} {
} {
    \begin{abstract}
    \end{abstract}        
} 

\notbool{arxiv} {
    \notbool{jrssb} {

        \noindent%
        \vfill    

        \newpage
        \spacingset{1.9} 
    }
}

\section{Introduction}\seclabel{introduction}

Suppose we want to use Bayesian methods to infer some parameter or quantity of
interest from exchangeable data. The Bayesian posterior is often summarized by a
point estimate and measure of subjective uncertainty, which we will respectively
take to be the posterior mean and posterior variance.  The computed posterior
mean is a function of the observed data, and so has a well-defined frequentist
variance under re-sampling of the exchangeable data.  The Bayesian posterior
variance of a quantity and the frequentist variance of its mean are conceptually
distinct: the posterior variance can be thought of as describing the
uncertainty of subjective beliefs under the assumption of correct model
specification, and the frequentist variance captures the sensitivity of the
Bayesian point estimate to re-sampling of the observed data.  In general,
both may be of interest.

In general, the frequentist and Bayesian variances are not only conceptually
distinct but can be numerically very different.  It is true that, with
independent and identically distributed (IID) data, some reasonable regularity
conditions, a correctly specified model, and a parameter vector of fixed
dimension, the frequentist and Bayesian variances coincide asymptotically as a
consequence of the Bernstein-von Mises theorem, a.k.a the Bayesian central limit
theorem.  However, under \emph{model misspecification}, the equality between the
frequentist and Bayesian variances no longer holds in general, even
asymptotically \cite{kleijn:2012:bvm} --- and most models are misspecified.

Another practically relevant setting in which frequentist and Bayesian variances
may differ, even asymptotically, is when one wishes to compute the frequentist
variance under the assumption that the observed data is exchangeable, an
assumption which may make sense only \emph{conditional on unobserved quantities}
that are integrated out of the posterior.  For example, consider the use of
polling data to predict the outcome of the 2016 U.S. presidential election, an
example from \cite{economist:2020:election} which we discuss at length in
\secref{election}. The data consists of polling responses from randomly selected
individuals and the quantity of interest is an expected vote proportion on
election day for one of the two candidates.  Since the individuals were selected
randomly, and the pollsters might have plausibly selected different individuals
to survey, it is reasonable to ask how variable the final posterior election
prediction is under the random selection of individuals in a frequentist sense.
Under this sampling distribution, the individuals chosen for polls on different
days are independent of one another.  However, the Bayesian model contains
random effects for the biases of different polling agencies as well as
persistent time random effects, which correlate the polling results at different
time points, marginally over the random effects and data generating process. The
frequentist variability we are interested thus implicitly conditions on these
random effects, but the Bayesian posterior variability is marginal over these
random effects.  In general, the two variances will differ, even for very
large data sets.

Though the meanings are different, both the frequentist and Bayesian notions of
variance can be useful.  Even a strict Bayesian can meaningfully interpret the
frequentist variance of the posterior mean as a \emph{sensitivity analysis},
quantifying how much the Bayesian point estimate would vary had we received
different but realistic data.  Such a sensitivity analysis is particularly
germane when the data is known to have come from a random sample from a real
population.  Consider again the polling example.  Since the poll respondents
were randomly chosen in the first place, we hope that our observed responses
were representative, and that our model did not over-fit the particular dataset
we observed.  That is, we would hope that our election-day prediction would not
be too different if we had gone back and randomly chosen a different set of
respondents.  If, under re-sampling of the polling data, the expected vote
proportion varies much more than the width of a posterior credible interval,
then we would say that the model is \emph{not robust} to re-sampling the polling
data.  In other words, excess frequentist variability can be an indicator of
important non-robustness.

When posterior means are estimated using Markov Chain Monte Carlo (MCMC), the
two standard tools for estimating the frequentist variability of Bayesian
posterior expectations --- the nonparametric bootstrap  (henceforth simply
``bootstrap'') and the Laplace approximation --- can be problematic.
%
To compute a non-parametric bootstrap estimate of
the frequentist variance of a posterior mean, one repeatedly draws ``bootstrap
datasets,'' with replacement from the original set of datapoints. For each
bootstrap dataset, one then computes the posterior mean for each set of
re-sampled data \cite{huggins:2019:bayesbag}.  Though straightforward to
implement and naively parallelizable, the non-parametric bootstrap remains
computationally expensive, since a separate posterior estimate (e.g., a
run of MCMC) must be computed for each bootstrap dataset.
%
Alternatively, the Laplace approximation can be used to compute the frequentist
sandwich covariance using the maximum a-posteriori (MAP) estimate, in
computational mimicry of the asymptotic analysis of
\cite{kleijn:2012:bvm}.\footnote{The Laplace approximation is typically thought
of as the quadratic approximation to the posterior using the observed Fisher
information. Here, we hope we do no injustice to Laplace by interpreting the
procedure more broadly, using the quadratic approximation to compute an
asymptotic frequentist variance.}  However, in the presence of high-dimensional
random effects, the Laplace approximation must either fix the random effects at
their MAP estimates, which often results in a bad or even singular posterior
approximation, or perform an intractable integral to marginalize over the random
effects.\footnote{When using the Laplace approximation, one wants to integrate
out random effects out of the likelihood even when computing frequentist
variance conditional on said random effects, since the posterior mean whose
variance interests us is computed marginally over the random effects.  The
assumption that the data points are independent then enters in the computation
of the sampling variance of the score.  Indeed, it is often precisely the
intractability of the marginalizing integral that motivates the use of MCMC
rather than marginal MLE methods, rendering the Laplace approximation
intractable.}

In the present paper, we provide an estimator of the frequentist variability of
posterior expectations which enjoys the model-agnostic properties of the
non-parametric bootstrap, but can be computed from a single run of MCMC. Our
method is based on the \emph{infinitesimal jackknife} (IJ) variance estimator, a
classical frequentist variance estimator derived from the \emph{empirical
influence function} of a statistical functional.  We show that the Bayesian
empirical influence function is easily computable as a particular set of
posterior covariances, the variance of which is a consistent estimator of the
frequentist covariance of a posterior expectation when a Bayesian central limit
theorem (BCLT) holds for all model parameters.  We also argue a negative result:
in the presence of high-dimensional random effects, the IJ variance is, in
general, an \emph{inconsistent estimator} of the frequentist variance
conditional on those random effects.  Nevertheless, we argue theoretically and
through experiments that the IJ variance remains useful in finite sample,
especially given its computational advantages.

Finally, though our present concern is frequentist variance estimates, the
empirical influence function is a powerful statistical tool, with applications
in robustness, cross validation, and bias estimation. Consequently, our results
for the influence function of Bayesian posteriors and a BCLT for expectations of
data-dependent functions are of broader interest.  We leave off exploring these
topics for future work.

\subsection{Problem statement}\seclabel{basics}
Suppose we observe $N$ data points, $\xvec = (\x_1, \ldots \x_N)$.  We will
assume that the entries of $\xvec$ are independently and identically distributed
(IID) according to a distribution $\fdist$.  In a slight abuse of notation, we
will write $\fdist(\xvec)$ for $\prod_{n=1}^N \fdist(\x_n)$.  Throughout, the
distribution $\fdist$ will be taken as fixed but unknown.  We assume that we
model the IID distribution of $\x_n$, perhaps incorrectly, by a parametric model
$\p(\xvec, \theta) = \exp\left(\sumn \ell(\x_n \vert \theta)\right)
\prior(\theta)$, where $\theta \in \thetadom \subseteq \mathbb{R}^D$ and
$\prior(\cdot)$ is a density with respect to the Lebesgue measure on
$\mathbb{R}^D$.  The model may be misspecified, i.e., there may be no $\theta
\in \thetadom$ such that $\p(\x_n \vert \theta)$ is the same distribution as
$\fdist$.  The posterior expectation of some vector-valued quantity $g(\theta)
\in \mathbb{R}^{\gdim}$ is given by given by Bayes' rule:
\begin{align*}
\expect{\post}{\g(\theta)} =
    \frac{
    \int_{\thetadom} g(\theta)
        \exp\left(N \meann \ell(\x_n \vert \theta) \right) \prior(\theta)
                        d\theta}
         {\int_{\thetadom}
         \exp\left(N \meann \ell(\x_n \vert \theta) \right) \prior(\theta)
         d\theta}.
\end{align*}

The quantity $\expect{\post}{\g(\theta)}$ is a function of the data $\xvec$,
through its dependence on the random function $\theta \mapsto \meann \ell(\x_n
\vert \theta)$.  We will be primarily concerned with the problem of how to
estimate the variability of $\expect{\post}{\g(\theta)}$ under sampling from
$\fdist(\xvec)$, when $N$ is large, using Markov Chain Monte Carlo (MCMC)
samples from the posterior $\post$, in the presence of possible
misspecification.
Using an estimate of its variability, one can use classical frequentist tests to
estimate how much the posterior mean might be expected to have varied were we to
have drawn a different sample $\xvec$.

We will study conditions under which, under sampling from $\fdist(\xvec)$,
%
\begin{align}\eqlabel{gtrue}
\gmeantrue := \expect{\fdist(\xvec)}{\expect{\post}{\g(\theta)}} < \infty
\quad\textrm{and}\quad
\sqrt{N}\left(\expect{\post}{\g(\theta)} - \gmeantrue \right)
\dlim
\normdist\left(\cdot \vert 0, \gcovtrue\right).
%
\end{align}
%

We will consider two different regimes for the dimension $D$ of $\theta$ as $N$
grows: a regime in which $D$ remains fixed as $N$ grows, as in the classical
parametric Bayesian central limit theorem (BCLT), and a high-dimensional
``global--local'' regime in which $D$ grows with $N$, but a finite-dimensional
subspace of $\theta$ is expected to obey a BCLT marginally.  The simplest case
is when $D$ remains fixed as $N$ grows; we will call this a ``fixed-dimension''
model.  A more complex model (which is more typical of real-world MCMC
applications) is one in which $\theta$ can be decomposed into two parts: a
``global'' parameter $\gamma \in \gammadom \subseteq \rdom{\gammadim}$ whose
dimension $\gammadim$ stays fixed as $N$ grows, and a ``local'' parameter
$\lambda \in \lambdadom \subseteq \rdom{\lambdadim}$ where $\lambdadim$ may grow
with $N$.  We will refer to this case as a ``global--local'' model.   A classical
motivating example for global--local models is a generalized linear mixed model,
in which $\gamma$ represents ``fixed effects'' which influence the distribution
of each observation, and $\lambda = (\lambda_1, \ldots, \lambda_N)$ is a vector
of ``random effects,'' one for each datapoint $\x_n$
(\cite{mcculloch:2003:glmms, gelman:2006:arm}).  For a global--local model, we
write $\theta = (\gamma, \lambda)$ and $D = \gammadim + \lambdadim$, and assume
that $N$ is large relative to $\gammadim$ but perhaps not large relative to
$\lambdadim$. We assume that we are only interested in the expectations of
functions of $\gamma$, so that in a slight abuse of notation we can write
$\g(\theta) = \g(\gamma, \lambda) = \g(\gamma)$.  Thus, we assume that we are
only interested in the posterior expectations of fixed-dimensional quantities,
though we may be using MCMC to integrate out high-dimensional parameters in a
global--local model.


For the remainder of the paper we will investigate the accuracy of $\gcovij$.
First, in \secref{finite_dim}, we prove the consistency of $\gcovij$ in the
finite-dimensional setting.  In \secref{global_local}, we will introduce our
high-dimensional setting, and discuss why the analysis of \secref{finite_dim}
will not suffice.  We then introduce a von Mises expansion approach, providing
another proof of the consistency of $\gcovij$ in \secref{von_mises} under
stronger assumptions, but with a sufficiently general proof technique to
understand the high-dimensional setting.  In \secref{von_mises_high_dim}, we
show that $\gcovij$ is inconsistent in general in high dimensions, but provide
intuition for why we may still observe good results in practice. In
\secref{experiments}, we examine the performance of the IJ covariance relative
to the bootstrap on a range of real and simulated examples.

\subsection{Notation}\seclabel{notation}
We will begin by defining some notation. Given a distribution $\p(t)$ and an
integrable function $\phi(t)$, we will use the expectation operator as a formal
shorthand for the integral
%
$\expect{\p(t)}{\phi(t)} = \int \phi(t) \p(dt)$.
%
Any quantities not explicitly part of the distribution beneath the expectation
operator should be thought of as fixed. Covariances and variances are treated
similarly, i.e,
%
%
$\cov{\p(t)}{\phi(t), \psi(t)} =
    \expect{\p(t)}{\phi(t) \psi(t)^T} -
        \expect{\p(t)}{\phi(t)} \expect{\p(t)}{\psi(t)}^T$.
%
%
We will denote sample covariances over an index $n \in \{1,
\ldots, N \} = [N]$ as $\covhat{n \in [N]}{\cdot, \cdot}$.

We denote partial derivatives with respect to $\theta$ with bracketed
subscripts.  For example, for a function $\psi(\theta)$ taking values in
$R^{D_\psi}$, let $\psigrad{k}(\theta)$ denote the array of $k-$th order
partial derivatives with respect to $\theta$, i.e.,
\begin{align*}
\psigrad{k}(\theta) :=
    \fracat{\partial^k \psi(\theta)}{\partial\theta^k}{\theta}.
\end{align*}
If $\psi$ depends on any variables other than $\theta$, those variables will be
taken as fixed in the partial differentiation.
In general, $\psigrad{k}(\theta)$ is an array with $D_\psi D^k$ entries, and we
will introduce special notation as needed when summing over the array's entries.
But we will mostly be able to avoid such special notation by interpreting a few
special cases in standard matrix notation:
\begin{itemize}
\item When $D_\psi = 1$ and $k = 1$, $\psigrad{1}(\theta)$ is taken to be a
$D$-column vector.
\item When $D_\psi > 1$ and $k = 1$, $\psigrad{1}(\theta)$ is taken to be a
$D_\psi \times D$ matrix.
\item When $D_\psi = 1$ and $k=2$, $\psigrad{2}(\theta)$ is taken to be a $D
\times D$ matrix.
\item When $k = 0$, $\psigrad{0}(\theta)$ is just $\psi(\theta)$.
\end{itemize}
%

Our proposal and existing methods for estimating $\gcovtrue$ will depend on some key
quantities to which we give symbols in \defref{loglik_def}.  In the left hand
column of \defref{loglik_def} we have finite-sample quantities, and on the right
hand side we have the corresponding population quantities. In \secref{bayes_clt}
below we will state precise conditions under which all the quantities of
\defref{loglik_def} are well-defined.  We will use $\xn$ to denote a generic
single datapoint, in contrast to specific observed datapoints, $\x_n$.
\begin{defn}\deflabel{loglik_def}
\begin{align*}
&\textrm{Finite-sample quantity} && \textrm{Population quantity}\\
\likhat(\theta) :={}&
   \frac{1}{N} \sumn \ell(\x_n \vert \theta) +
   \frac{1}{N} \logprior(\theta)
&
\quad \lik(\theta) :={}& \expect{\fdist(\xn)}{\ell(\xn \vert \theta)}
\\
\thetahat :={}& \argmax_{\theta \in \thetadom} \likhat(\theta)
&
\thetatrue :={}& \argmax_{\theta \in \thetadom} \lik(\theta)
\\
\scorecovhat :={}&
   \frac{1}{N} \sumn \ellgrad{1}(\x_n \vert \thetahat)
                     \ellgrad{1}(\x_n \vert \thetahat)^T
&
\scorecov :={}&
    \cov{\fdist(\xn)}{\ellgrad{1}(\x_n \vert \thetatrue)}
\\
\infohat :={}& -\likhatk{2}(\thetahat)
&
\info :={}& - \expect{\fdist(\xn)}{\ellgrad{2}(\x_n \vert \thetatrue)}
\end{align*}
\end{defn}
%

We will take the population quantities in the right column of
\defref{loglik_def} to be fixed (though, in general, unknown), and study the
sampling behavior of quantities in the left column. That is, we will study the
asymptotic behavior of formal Bayesian posterior quantities as a function of a
fixed data generating process, following much of the literature on the
frequentist behavior of Bayesian estimators (e.g., \cite{van:2000:asymptotic}).
%
Our notation for stochastic convergence reflects these assumptions. The notation
$\dlim$ will mean convergence in distribution, and the notation $\plim$ will
mean convergence in probability, both under IID sampling from $\xfdist$ as $N
\rightarrow \infty$.  Similarly, order notation $O(\cdot)$ applies as $N
\rightarrow \infty$, and probabilistic order notation $O_p(\cdot)$ denotes order
in probability over $\fdist(\xn)$.  Tilde order notation $\ordlogp{\cdot}$
neglects terms that are sub-polynomial in $N$.

Let $\norm{\cdot}_1$, $\norm{\cdot}_2$, and $\norm{\cdot}_\infty$ denote the
usual vector norms.  When applied to an array of derivatives, the norms mean the
the vector norm applied to the vectorized version of the array.


\subsection{Three methods for estimating $\gcovtrue$}


There are two well-known classical methods for estimating $\gcovtrue$: the
sandwich covariance and the nonparametric bootstrap.  The sandwich covariance
requires computation of the {\em maximum a posteriori} (MAP) estimate, and is
not available from MCMC samples.  The bootstrap requires running many MCMC
chains with slightly different data sets, which can be computationally
expensive. We describe a third method, the infinitesimal jackknife (IJ)
covariance estimate, which is also classical, but which has not been previously
applied to MCMC estimates of posterior means.  All three
are succinctly summarized in \figref{algorithm}.

The first well-known approach to estimating $\gcovtrue$, the ``sandwich
approximation,'' is given in \algrref{sandwich}. The sandwich covariance first
forms the {\em maximum a posteriori} (MAP) estimate, $\thetahat$, approximate
the posterior expectation $\expect{\post}{\g(\theta)}$ by $\g(\thetahat)$, and
use the standard frequentist ``sandwich covariance'' matrix for the covariance
of optima, together with the delta method (\cite{stefanski:2002:mestimation},
\cite[Chapter 7]{van:2000:asymptotic}).
When a sufficiently strong Bayesian CLT applies, it can be shown that
$g(\thetahat)$ consistently estimates $\expect{\post}{g(\theta)}$ and the delta
method sandwich covariance matrix consistently estimates $\gcovtrue$ (more
details will follow in \secref{finite_dim_clts} below).  However, when $D$ is
not small relative to $N$ (as in the global--local regime), or when the posterior
is not approximately normal, the MAP can be a poor approximation to the
posterior mean, and the sandwich covariance matrix estimate can even be
singular.

When the MAP and sandwich covariance are unreliable, one typically turns to MCMC
estimates.  To estimate $\gcovtrue$ from MCMC samples, a classical method is the
nonparametric bootstrap, which we will simply refer to as the bootstrap
(\cite{efron:1994:bootstrap,huggins:2019:bayesbag}).  The bootstrap is shown in
\algrref{boot}.  The primary difficulty with the bootstrap is that it requires
$B$ different MCMC chains, which can be computationally expensive; often, in
practice, even a single run of MCMC requires considerable computing resources.
When the MCMC noise can be made arbitrarily small, the bootstrap is a consistent
estimator of $\gcovtrue$ in the fixed-dimension case under conditions given in
\cite{huggins:2019:bayesbag}.
%
    %
    %
    \begin{align}\eqlabel{ij_terms}
    \infl_n := N \cov{\post}{g(\theta), \ell(\x_n \vert \theta)}
    \textrm{ (Influence score)}
    \quad\quad
    \gcovij := \covhat{n=1,\ldots,N}{\infl_n}.
    \end{align}
    %

We propose an estimator of $\gcovtrue$ that can be
computed from a single MCMC run using only the tractable $\ell(\x_n \vert \z,
\theta)$.  We call our estimator the {\em infinitesimal jackknife (IJ)
covariance estimator}, and denote it as $\gcovij$.  The IJ covariance estimator
is given by the sample variance of a set of simple posterior covariances, as
given in in \eqref{ij_terms}.

In practice, the posterior covariances of \eqref{ij_terms} typically need to be
estimated with MCMC.  We denote  by $\gcovijhat$ the MCMC estimate of $\gcovij$,
and as described the computation of $\gcovijhat$ in \algrref{ijalg}.  For any
fixed $N$, the output $\gcovijhat$ of \algrref{ijalg} approaches $\gcovij$ under
a well-mixing MCMC chain as the number of MCMC samples $M \rightarrow \infty$.
We will theoretically analyze $\gcovij$, though our experiments will all employ
$\gcovijhat$. Importantly, \algrref{ijalg} requires only single run of MCMC
without relying on accuracy of the MAP, as required by the sandwich estimator,
or running $B$ different MCMC samplers, as required by the bootstrap estimator.

\begin{figure}
    
    
    \begin{algorithm}[H]
        \caption{Sandwich covariance}\algrlabel{sandwich}
        \textbf{Requires: } Differentiable $\theta \mapsto \likhat(\theta)$
        \begin{algorithmic}
        \State $\thetahat \gets \argmax_{\theta \in \thetadom} \likhat(\theta)$
        \State $\hat{\Sigma}_\theta \gets \infohat^{-1} \scorecovhat \infohat^{-1}$
        \State \Return $\gcovmaphat \gets \ggrad{1}(\thetahat)
            \hat{\Sigma}_\theta \ggrad{1}(\thetahat)^\trans$
        %
        \end{algorithmic}
    \end{algorithm}
    %
    %
    %
    
    \begin{algorithm}[H]
        \caption{Bootstrap covariance}\algrlabel{boot}
        \textbf{Requires: } Posterior sampler $\xvec^* \mapsto \theta^1,\ldots,\theta^M \sim \p(\theta | \xvec^*)$
        (e.g. MCMC)
        \begin{algorithmic}
        \For{$b=1,\ldots,B$}
            \State $\xvec^* \gets$ Sample with replacement from the rows of $\xvec$
            \State Sample
                $\theta^1,\ldots,\theta^M \sim \p(\theta | \xvec^*)$
            \State $\mathscr{G}^b \gets
                \frac{1}{M} \sum_{m=1}^M g(\theta^m)$
        \EndFor
        \State \Return $\gcovboothat \gets \covhat{b \in [B]}{\sqrt{N} \mathscr{G}^b}$
        %
        \end{algorithmic}
    \end{algorithm}
    %
    
    \begin{algorithm}[H]
        \caption{IJ covariance (our proposal)}\algrlabel{ijalg}
        \textbf{Requires: } Log likelihood $\theta, \x_n \mapsto \ell(\x_n \vert \theta)$,
        posterior samples $\theta^1,\ldots,\theta^M \sim \p(\theta | \xvec)$
        
        \begin{algorithmic}
        %
        \For{$n=1,\ldots,N$}
            \State $\hat{\infl}_n \gets N
                \covhat{m \in [M]}{\ell(\x_n \vert \theta^m), \g(\theta^m)}$
        \EndFor
        \State \Return $\gcovijhat \gets \covhat{n\in [N]}{\hat{\infl}_n}$
        %
        \end{algorithmic}
    \end{algorithm}
    %
        
    
    \caption{Three algorithms for estimating $\gcovtrue$.}\figlabel{algorithm}
    \end{figure}

%
\section{The finite-dimensional case}
\seclabel{finite_dim}
In this section, we formally state theoretical conditions under which the IJ
covariance estimate is consistent in the finite-dimensional case.  For the
remainder of this section only, we assume that $\thetadim$ remains fixed as $N
\rightarrow \infty$.  Our assumptions will be slightly stronger than those
required for a classical finite-dimensional Bayesian central limit theorem
(BCLT) in order to control the data-dependence in the influence scores.

\subsection{Intuition}
\seclabel{bayes_clt}
Before we state our formal results in \secref{bayes_clt}, we will present some
high-level intuition to motivate the IJ covariance estimate, as well as
our data-dependent Bayesian central limit theorem of \secref{finite_dim_clts}.

For simplicity let us for the moment take the quantity of interest,
$g(\theta)$, to simply be $\g(\theta) = \theta$.
%
According to \cite{kleijn:2012:bvm}, we might informally expect the following
central-limit theorem results to hold:
\begin{align*}
N \cov{\post}{\theta} \plim& \info^{-1}
    &\textrm{Bayesian CLT (BCLT)}\\
\sqrt{N}\left(\thetahat - \thetatrue\right) \dlim&
    \normdist\left(\cdot | 0, \info^{-1} \scorecov \info^{-1} \right)
    &\textrm{Frequentist CLT (FCLT)}\\
\norm{\expect{\post}{\theta} - \thetahat}_2 ={}& O_p(N^{-1}).
    &\textrm{Bayesian CLT (BCLT)}
\end{align*}
Here and throughout the paper, we somewhat loosely use
the term ``Bayesian CLT'' (BCLT) to refer to results concerning the asymptotic
properties of the posterior measure, and ``Frequentist CLT'' (FCLT) for results
concerning properties under data sampling.

The limiting variance we are trying to estimate is thus the sandwich covariance
matrix, $\gcovtrue = \info^{-1} \scorecov \info^{-1}$. Since $\info^{-1} \ne
\info^{-1} \scorecov \info^{-1}$ in general unless the model is correctly
specified, $\gcovtrue$ cannot simply be approximated by the Bayesian posterior
covariance.

If one could compute $\thetahat$, and so $\scorecovhat$, and $\infohat$, then
one could form the sample estimator $\gcovhat := \infohat^{-1} \scorecov
\infohat^{-1}$, which would consistently estimate $\gcovhat$.  The IJ covariance
estimator is motivated by circumstances when it is not easy or desirable to
compute $\thetahat$, but when one can compute or approximate the posterior
expectations required to compute the IJ covariance estimator $\gcovij$.

The reason that $\gcovij$ might be expected to approximate $\gcovtrue$
asymptotically can be seen by Taylor expanding the constituent terms of the
influence score $\infl_n$ and applying the CLT results given above.  Using
\eqref{ij_terms}, for any particular $n$, we expect that
\begin{align}\eqlabel{infl_heuristic}
    N^{-1} \infl_n = \cov{\post}{\theta, \ell(\x_n \vert \theta)} =
    \infohat^{-1} \ellgrad{1}(\x_n \vert \thetahat) + O_p(N^{-1}).
\end{align}
%
If \eqref{infl_heuristic} holds sufficiently uniformly in $n$
(in a sense which we analyze precisely below), then \eqref{infl_heuristic}
implies the asymptotic equivalence of $\gcovij$ and $\gcovtrue$:
\begin{align}
%
\infl(\x_n) \approx{}& \infohat^{-1} \ellgrad{1}(\x_n \vert \thetahat)
\quad\Rightarrow\quad
\gcovij ={} \meann \infl_n \infl_n^T - \bar{\infl} \, \bar{\infl}^T
\approx{} \infohat^{-1} \scorecovhat \infohat^{-1}, \eqlabel{ij_heuristic}
%
\end{align}
and we know that $\infohat^{-1} \scorecovhat \infohat^{-1} = \gcovhat \plim
\gcovtrue$.

However, in general, the $\ordp{N^{-1}}$ term depends in \eqref{infl_heuristic}
on $\x_n$, and the IJ covariance depends on {\em every} $\infl(\x_n): n \in
[N]$, some of which may be pathologically behaved, especially under
misspecification.
The question, then, is whether the approximation \eqref{infl_heuristic} holds
sufficiently strongly for us to apply the reasoning of \eqref{ij_heuristic} ---
in particular, whether we can safely sum the squares of the $\ordp{N^{-1}}$
terms of \eqref{infl_heuristic}.  It is to this task we now turn in the next
section.


\subsection{Consistency theorem}
\seclabel{finite_dim_clts}

To prove the consistency of the IJ covariance,
we assume the following conditions on the model.

\begin{assu}\assulabel{bayes_clt}

Let the following conditions hold for the model and data generating process.
\begin{enumerate}
\item \itemlabel{finite_dim} The dimension of $\theta$ stays fixed as $N$ changes.
\item \itemlabel{prior_smooth}
The prior $\prior(\theta)$ is four times continuously differentiable in a
neighborhood of $\thetatrue$, and is non-zero for all $\theta \in \thetadom$.
\item \itemlabel{prior_proper}
The prior is proper.%
\footnote{Since our analysis is asymptotic, this assumption on the prior can
be replaced by the assumption that they hold for the posterior after some
finite number of observations, which would then be incorporated into a ``prior''
for the remaining, unboundedly large number of observations.}
\item \itemlabel{loglik_smooth} The functions $\lik(\theta)$ and $\theta \mapsto
\ell(\xn \vert \theta)$ and are four times continuously differentiable in a
neighborhood of $\thetatrue$, and the interchange of partial differentiation
with respect to $\theta$ and expectations with respect to $\xfdist$ is
justified.
\item \itemlabel{loglik_ulln} There exists a function $M(\xn)$ and a
$\deltaulln > 0$ such that
\begin{align*}
\norm{
    \sup_{\theta \in \thetaball{\deltaulln}}
        \ellgrad{4}(\xn \vert \theta)
}_2 \le M(\xn)
\quad\textrm{and}\quad
\expect{\fdist(\xn)}{M(\xn)^2} < \infty.
\end{align*}
\item \itemlabel{strict_opt}
The log likelihood has a strict optimum in the following sense:
    \begin{itemize}
    \item
    A unique maximum
    $\thetatrue = \mathrm{argmax}_{\theta \in \thetadom} \lik(\theta)$
    exists.
    \item The matrix $\info$ is positive definite.
    \item For any $\delta > 0$, with probability approaching one as $N
    \rightarrow \infty$,
    \begin{align*}
    \sup_{\theta: \thetadom \setminus \thetaball{\delta}}
    \left( \meann \ell(\x_n \vert \thetatrue) -
           \meann \ell(\x_n \vert \theta)
    \right) > 0.
    \end{align*}
    \end{itemize}
\end{enumerate}

\end{assu}


Under \assuref{bayes_clt}, a BCLT will apply in total variation distance by
classical results (\cite{van:2000:asymptotic} Chapter 10), but we require an
expansion of posterior moments.  Series expansions of posterior moments are
studied by \cite{kass:1990:posteriorexpansions} (under stricter
differentiability conditions than ours), but with no control over the data
dependence in the error when the expectation depends on the data.

We now state conditions on functions $\phi(\theta)$ and $\phi(\theta, \x_n)$
sufficient to form a controllable series expansion of their posterior
expectations.


\begin{defn}\deflabel{no_data_bclt_okay}
A function of $\theta$, $\phi(\theta)$, is ``$K$-th order
BCLT-okay'' if
\begin{enumerate}
\item \itemlabel{eprior_finite} The prior expectation $\expect{\prior(\theta)}{\norm{\phi(\theta)}_2}$
is finite.%
\footnote{As in \assuref{bayes_clt}, these assumptions may alternatively be taken
to hold almost surely after some finite number of observations without
affecting the asymptotic results.}
\item \itemlabel{cont_diffable} The map $\theta \mapsto \phi(\theta)$ is $K$ times continuously
differentiable.
\end{enumerate}
\end{defn}

When the function depends on the data, we make the following additional requirements.


\begin{defn}\deflabel{bclt_okay}
A function $\phi(\theta, \x_n)$ depending on a single datapoint is $K$-th order
BCLT-okay if, for each $k = 1, \ldots, K$,
\begin{enumerate}
\item\itemlabel{bclt_okay_as}
    $\theta \mapsto \phigrad{k}(\theta, \x_n)$ is BCLT-okay
    $\fdist$-almost surely
\item\itemlabel{prior_op1}
    $\meann \expect{\prior(\theta)}{\norm{\phi(\theta, \x_n)}^2_2} = O_p(1)$.
\item\itemlabel{grad_op1}
For some $\delta > 0$,
$\sup_{\theta \in \thetaball{\delta}}
    \meann \norm{\phigrad{k}(\theta, \x_n)}^2_2 = O_p(1)$.
\end{enumerate}
\end{defn}

\Assuref{bayes_clt} and \defref{no_data_bclt_okay} are sufficient (indeed,
stronger than necessary) to prove a classical FCLT, which we now state in
\propref{freq_clt}. \propref{freq_clt} implies that our target covariance is the
delta method sandwich covariance $\gcovtrue = \ggrad{1}(\thetatrue)\info^{-1}
\scorecov \info^{-1} \ggrad{1}(\thetatrue)^T$.

\begin{prop}\proplabel{freq_clt}(\cite[Chapter 7]{van:2000:asymptotic},
    see also \appref{freq_clt_proof})
Let \assuref{bayes_clt} hold, and assume that $\g(\theta)$ is third-order
BCLT-okay. Then $\thetahat \plim \thetatrue$,and a misspecified FCLT applies in
the following sense:
\begin{align*}
\sqrt{N}\left(g(\thetahat) - g(\thetatrue)\right)
\dlim
\normdist\left(0, \gcovtrue \right)
\quad\textrm{where}\quad
\gcovtrue =
    \ggrad{1}(\thetatrue)
    \info^{-1} \scorecov \info^{-1} \ggrad{1}(\thetatrue)^\trans.
\end{align*}
\end{prop}
%

Using these assumptions, we can state a BCLT with control over a data-dependent
residual.  The following theorem is our primary technical contribution; from it,
all our consistency results will follow directly.

%
\begin{thm}\thmlabel{bayes_clt_main}(See proof in \appref{bayes_clt_proof}.)
Let \assuref{bayes_clt} hold, and assume that $\phi(\theta, \x_n)$ is
third-order BCLT-okay.  Let $\hat{M}$ denote the
$\thetadim\times\thetadim\times\thetadim\times\thetadim$ array of fourth moments
of the Gaussian distribution $\normdist(0, \infohat^{-1})$. With
$\fdist$-probability approaching one, for each $n$,
\begin{align}
\expect{\post}{\phi(\theta, \x_n)} - \phi(\thetahat, \x_n)
={}&
N^{-1} \left(
    \frac{1}{2} \phigrad{2}(\thetahat, \x_n) \infohat^{-1} +
    \frac{1}{6} \phigrad{1}(\thetahat, \x_n) \likhatk{3}(\thetahat) \hat{M}
\right) +
\nonumber\\&
\ordlog{N^{-2}} \resid{\phi}(\x_n),
\textrm{ where }\meann \resid{\phi}(\x_n)^2 = O_p(1).
\eqlabel{bclt_expansion}
\end{align}
%

In \eqref{bclt_expansion},
$\resid{\phi}(\x_n)$ is a scalar to be repeated componentwise, and
derivatives with respect to $\theta$ are summed against the indices of the
Gaussian arrays.\footnote{Although a slight abuse of notation, these simple
conventions avoid a tedious thicket of indicial sums. Concretely, we mean that
\begin{align*}
\phigrad{2}(\thetahat, \x_n) \infohat^{-1} ={}&
    \sum_{a,b=1}^{\thetadim}
    \fracat{\partial^2 \phi(\theta, \x_n)}
           {\partial\theta_a \partial\theta_b}{\thetahat}
   (\infohat^{-1})_{ab}
\\
\phigrad{1}(\thetahat, \x_n) \likhatk{3}(\thetahat) \hat{M}
={}&
\sum_{a,b,c,d=1}^{\thetadim}
\fracat{\partial \phi(\theta, \x_n)}{\partial\theta_a }{\thetahat}
\fracat{\partial^3 \lik(\theta)}
       {\partial\theta_b \partial\theta_c \partial \theta_c }{\thetahat}
\hat{M}_{abcd}
\end{align*}
each of which is a vector expression of the same length as $\phi(\theta,
\x_n)$.  The term $\resid{\phi}(\x_n)$ \eqref{bclt_expansion} should
be thought of as a vector of the same length as $\phi(\theta, \x_n)$
with the scalar $\resid{\phi}(\x_n)$ in each entry.
}
\end{thm}
%



The limiting frequentist distribution of $\expect{\post}{g(\theta)}$ under
misspecification, proven under weaker conditions than ours in \cite[Theorem
10.8]{van:2000:asymptotic}, follows from \thmref{bayes_clt_main}.  We state and
prove the result here, since it establishes the estimand which the IJ covariance
is targeting.

\begin{cor}\corlabel{bayes_expectation_dist} \cite[Theorem 10.8]{van:2000:asymptotic}
Let \assuref{bayes_clt} hold, and assume that $\g(\theta)$ is
third-order BCLT-okay.  Then
\begin{align*}
\sqrt{N} \left( \expect{\post}{g(\theta)} - g(\thetatrue) \right)
    \dlim \normdist\left(0, \gcovtrue\right).
\end{align*}
%
%
\begin{proof}
Under \assuref{bayes_clt} and \defref{no_data_bclt_okay},
$\ggrad{1}(\thetahat)$, $\infohat$, and $\likhatk{3}(\thetahat)$ are all
$O_p(1)$.  (See \lemref{regular} in \appref{bayes_clt_proof}.)
By \thmref{bayes_clt_main} we thus have,
\begin{align*}
\sqrt{N} \left( \expect{\post}{g(\theta)} - g(\thetatrue) \right) =&
\sqrt{N} \left( g(\thetahat) - g(\thetatrue) \right) + O_p(N^{-1/2}),
\end{align*}
The distributional limit them follows from \propref{freq_clt} and Slutsky's
theorem.
\end{proof}
\end{cor}

Our main result, the consistency of $\gcovij$, follows from
\thmref{bayes_clt_main} and \corref{bayes_expectation_dist}.

\begin{thm}\thmlabel{ij_consistent}(See proof in \appref{ij_consistent}.)
Assume that the functions $\g(\theta)$, $\ell(\xn \vert \theta)$, and
$\g(\theta) \ell(\xn \vert \theta)$ are all third-order BCLT-okay, and let
\assuref{bayes_clt} hold. Then the infinitesimal jackknife variance estimate is
consistent. Specifically,
\begin{align*}
    \gcovij ={}& 
        \ggrad{1}(\thetahat) \infohat^{-1} \scorecovhat \infohat^{-1} \ggrad{1}(\thetahat)^\trans + \ordlogp{N^{-1}} 
        \plim \gcovtrue.
\end{align*}
\end{thm}


\section{High-dimensional global--local models}\seclabel{global_local}
We now turn to ``global--local'' models: a high-dimensional case with $\theta =
(\gamma, \lambda)$, where $\gamma$'s dimension is fixed, and the dimension of
$\lambda$ grows with $N$.  Arguably, this case is more realistic for
applications of MCMC than the finite-dimensional case, since global--local
marginal posteriors $\postg$ typically involve an intractable integral over the
latent variables, which is approximated using MCMC.

It will be useful to begin with a simple motivating example that captures
much of the intuition about why we care about global--local models, and
what is meant by frequentist variability under the exchangeability assumption.
We will use this example both to motivate the problem in the simplest possible
setting and also to have access to ground truth via simulation. In
\secref{experiments} we will consider more realistic examples.


\begin{ex}\exlabel{re}{(Nested random effects.)}
For a given number of groups $G$, let each datapoint $\x_n = \{\y_n, \a_n \}$
consist of a response $\y_n$ and a random ``group assignment'' $\a_n$, where
$\a_n = (\a_{n1}, \ldots, \a_{nG})$ such that if the $n$-th observation belongs
to group $g$ then $\a_{ng} = 1$ and the rest of the entries of $\a_n$ are zero.
Assume that the groups are equiprobable, so that $\expect{\fdist(\a_n)}{\a_{ng}}
= G^{-1}$. We will be particularly interested settings where $N / G$ does not
diverge as $N \rightarrow \infty$, so there are more groups as $N$ grows, and
$O(1)$ datapoints assigned to a typical group.

For our model,
let each group be associated with a scalar--valued ``random effect'' $\lambda_g
\in \mathbb{R}$. The local parameters are then $\lambda = (\lambda_1, \ldots,
\lambda_G)$, with $\lambdadim = G$. Suppose that there is a single scalar-valued
global parameter, $\gamma \in \mathbb{R}$, so $\gammadim = 1$. Let the model of
the data from group $g$ depend only on $\gamma$ and $\lambda_g$, so that
$\ell(\x_n \vert \gamma, \lambda) = \sumg \a_{ng} \ell(\y_n \vert \gamma,
\lambda_g, \a_{ng}=1)$.  Further, let the random effects be
IID \textit{a priori}, so the prior factorize as $\prior(\gamma,
\lambda) = \prior(\gamma) \prod_{g=1}^G \prior(\lambda_g)$, and the
components of $\lambda$ are independent {\em a posteriori} given $\gamma$.
\end{ex}

\Exref{re} captures two key features of the global--local models we consider:
that the IJ (and bootstrap) is estimating a covariance conditional on
fixed values of the random effects, and that the global parameter's posterior
concentrates.

First, observe that, under the assumed model, the datapoints $\x_n$ are IID
conditional on $\lambda$ and $\gamma$ since the group assignment indicators are
random.  Marginally over $\lambda$, the $\x_n$ are not IID, since the responses
$\y_n$ are correlated within a group.  Since the IJ (and bootstrap) estimate
covariances under the assumption of IID $\x_n$, they implicitly estimate the
sampling variance under repeated draws of the response and group
assignments, for a single, fixed value of $\lambda$ and $\gamma$.

Second, we expect $\p(\lambda_g | \xvec)$ does not concentrate for a typical
$g$, since a finite number of datapoints inform the posterior of $\lambda_g$.  
However, under typically regularity conditions, we do expect $\postg$ to
concentrate according as in \secref{finite_dim}, since the marginal log
likelihoods of the responses associated with group $g$ are IID.  The
following example details the reasoning.

\begin{ex}\exlabel{re_gamma}{(Concentration of nested random effects.)}
Let $\ng$ denote the indices of observations assigned to
group $g$, i.e. $\ng := \{n: \a_{ng} = g \}$. Taking $\x_n' := \{ \abs{\ng},
\y_n: n \in \ng \}$ to be the group of observations assigned to group $g$, we
have
\begin{align*}
\ell(\x_n' \vert \gamma)  = 
       \log \int \p(\lambda_g) \prod_{n \in \ng} \p(\x_n \vert \lambda_g, \gamma) d \lambda_g.
\end{align*}
If we additionally assume that set sizes $\abs{\ng}$ are IID Poisson random
variables,\footnote{By taking the group sizes to be IID Poisson, the total
number of data points $N' = \sumg \abs{\ng}$ becomes a (Poisson) random variable
rather than a fixed quantity.  But for a given $N$, by taking
$\expect{\fdist(\abs{ng})}{\abs{\ng}} = N / G$, we have $\expect{\fdist(N')}{N'}
= N$.  Since we condition on $N'$ in the posterior calculations, and since we
can $N' \rightarrow \infty$ as $N \rightarrow \infty$ with probability one, the
randomness in $N'$ does not affect the applicability of
\thmref{bayes_clt_main}.} then the grouped observations $\x_n'$ are IID with log
probability given by $\ell(\x_n' \vert \gamma)$.
\end{ex}


The results of \secref{finite_dim} show that the IJ applied to the marginal
observations $\ell(\x_n' \vert \gamma)$ 
given in \exref{re_gamma} will be a consistent estimator of the
covariance of $\expect{\p(\gamma \vert \x_1', \ldots, \x_G')}{\gamma}$ under IID samples
of $\x_n'$ --- but this is a different sampling distribution than IID $\x_n$.
Furthermore, in more complicated settings, the marginal log likelihood
$\ell(\x_n' \vert \gamma)$ is intractable, and we have access only to MCMC draws
of the joint log likelihood $\ell(\x_n \vert \gamma, \lambda)$.

In global--local models the dimension of $\theta = (\gamma, \lambda)$ grows with
$N$, and so we cannot apply the analysis of \secref{finite_dim} directly.
Instead, we will use \textit{von--Mises expansions} to study the limiting
behavior of \textit{low-dimensional posterior marginals}.  Our primary result
will be negative --- even when the posterior marginal the global parameters
$\gamma$ ``concentrates'' as in the BCLT of \secref{finite_dim}, the IJ
covariance estimate is inconsistent in general if the conditional posterior
$\postl$ does not concentrate as well.  

However, our negative results have a silver lining, which is that we conjecture
that the error in the IJ covariance estimates to be {\em quantitatively small}
when the variance of the local conditional posterior $\postl$ is itself
quantitatively small, even if not perfectly concentrated asymptotically.  
Our conjecture is based on a simplification of the residual from a von Mises
expansion in a class of high--dimensional exponential families.  We believe that
these arguments for the usefulness of the IJ covariance in the presence of
approximate concentration of $\postl$ support and are supported by our
experimental results of \secref{experiments} below, and are a promising subject
for future study.

Before proceeding, we make two notes on notation.  First, throughout this
section, we will restrict $\g(\theta)$ to be a scalar-valued function ($\gdim =
1$) for simplicity of notation.\footnote{In particular, this assumption will
allow us to use the mean value theorem rather than the integral form of a
Taylor series remainder, which is a considerable notational simplification.}
Second, in the global--local case of \secref{von_mises_high_dim}, we will take
derivative subscripts to represent partial derivatives with respect to $\gamma$
only, rather than $\theta$ as in the rest of the paper.  For example, we will
take
%
%
$\ellgrad{1}(\xn \vert \gamma, \lambda) =
\fracat{\partial \ell(\xn \vert \gamma, \lambda)}
       {\partial \gamma}{\gamma}$
%
%
rather than as a derivative with respect to both $\gamma$ and $\lambda$. All
other derivatives conventions as described in \secref{notation} will still
apply.

\subsection{A von Mises expansion of Bayesian posterior expectations}
\seclabel{von_mises}
To compute a von Mises expansion to a Bayesian posterior, we must express
posterior expectations as a functional of the data distribution.  Let $\gdist$
represent a generic (possibly signed) measure on a sigma algebra on the domain
of a datapoint $\xn$.  We will assume without further comment that all needed
quantities are measurable.  Define the ``generalized posterior'' functional
\begin{align}\eqlabel{expectation_functional}
\T(\gdist, N) := \frac{
    \int \g(\theta) \exp\left( N \int \ell(\xn \vert \theta) \gdist(d\xn)\right)
        \prior(\theta) d\theta
}
{
\int \exp\left( N \int \ell(\xn \vert \theta) \gdist(d\xn)\right)
    \prior(\theta) d\theta
}.
\end{align}

The definition \eqref{expectation_functional} expresses the usual posterior
expectation as a special case, as we now show. Let $\fndist$ denote the
empirical measure, placing mass $1/N$ on each datapoint in $\xvec$.  Since $\int
\ell(\xn \vert \theta) \fndist(d \xn) = \meann \ell(\x_n \vert \theta)$, we have
that $\T(\fndist, N) = \expect{\post}{\g(\theta)}$.  Analogously, we will write
$\expect{\postpert{\gdist, N}}{\g(\theta)}$ for posterior expectations and
covariances computed using the appropriate version of
\eqref{expectation_functional}, possibly with alternative functions in place of
$\g(\theta)$.

Note that, as $N$ grows, under sufficient regularity conditions, two things
occur in the posterior expectation $\T(\fndist, N)$.  As $N$ gets larger,
$\fndist \rightarrow \fdist$, and the posterior concentrates as the the weight
$N$ given to the integral $\int \ell(\xn \vert \theta) \fndist(d\xn)$ grows. The
definition \eqref{expectation_functional} decouples these two phenomena,
permitting expression of pseudo-posterior expectations such as $\T(\fdist, N)$,
which uses the limiting distribution $\fdist$, but the data weight $N$.  The
advantage of such a decomposition is that we can use uniform laws of large
numbers as $\fndist \rightarrow \fdist$ to deal with the data asymptotics, and
treat the posterior integral's asymptotics with the $N$ dependence.

To form a von Mises expansion for $\T(\fdist, N)$, one can parameterize a path
between $\fndist$ and $\fdist$ with a family of distributions $\ftdist$, indexed
by $\t \in [0, 1]$, such that $\ftdist$ equals $\fdist$ when $\t = 0$ and
$\fndist$ when $\t = 1$. Then, for a fixed $N$, if the map $\t \mapsto \T(\t,
N)$ is smooth, one can form the univariate Taylor series expansion:\footnote{One
modern perspective on von Mises expansions separates the smoothness of the
functional from probabilistic convergence of $\fndist$ to $\fdist$. This
approach is conceptually attractive, but requires careful justification of the
series expansion and convergence in an appropriate Banach space which embeds
both the empirical and limiting distribution functions (\cite{reeds:1976:thesis,
fernholz:1983:mises}).  We do not take this approach here.  Rather, our approach
to the von Mises expansion is more similar to the original work of
\cite{mises:1947:asymptotic}, in that we Taylor expand only along a path between
the population and sampling distributions, without defining an embedding normed
vector space.  One may also use a von Mises expansion informally to motivate a
particular estimator, and then study the estimator's error using other more
bespoke methods (\cite{mises:1947:asymptotic}, \cite[Chapter
20]{van:2000:asymptotic}, \cite[Chapter 6]{serfling:1980:approximation}).
In fact, this is what we have done in \thmref{ij_consistent} above.}
\begin{align}\eqlabel{t_expansion}
\T(\fndist, N) - \T(\fdist, N) ={}&
     \fracat{\partial \T(\ftdist, N)}{\partial \t}{\t=0} +
    \frac{1}{2} \fracat{\partial^2 \T(\ftdist, N)}{\partial \t^2}{\t=\ttil}
    \quad\textrm{for some }\ttil \in [0,1].
%
\end{align}
The advantage of \eqref{t_expansion} is that the first-order term in is often
simple, as we shall see shortly.  If one can then show that the residual
$\sqrt{N}$ time the second derivative converges to zero in probability, then the
asymptotic behavior of $\sqrt{N}\left( \T(\fndist, N) - \T(\fdist, N)\right)$ is
determined only by the relatively simple first-order term.




For our parameterized path, we will use the mixture distribution $\ftdist := \t
\fndist + (1 - \t) \fdist$ for scalar $\t \in [0, 1]$. Since the denominator of
\eqref{expectation_functional} is bounded away from zero for all $\t \in [0,1]$,
the map $\t \mapsto \T(\ftdist, N)$ is smooth, and we can form the expansion
given in \eqref{t_expansion}.  One advantage of the mixture path is that, under
\assuref{bayes_clt}, \thmref{bayes_clt_main} can be applied to the posterior
$\postft$ for $\t \in [0, 1]$ with essentially no modification
(see \lemref{bclt_okay_t} of \appref{ftdist} for more details).

For the mixture path, we can compute the terms in
\eqref{t_expansion} explicitly by scalar differentiation with respect to $\t$,
as given in the following lemma.

\begin{defn}\deflabel{centering} For a given $\t$, define the posterior centered quantities
    \begin{align*}
    \gbar(\theta) :={}& \g(\theta) - \expect{\postft}{\g(\theta)}
    &\quad\textrm{and}\quad
    \ellbar(\x_n \vert \theta) :={}& \ell(\x_n \vert \theta) -
    \expect{\postft}{\ell(\x_n \vert \theta)},
    \end{align*}
    and the sampling-centered quantities
    \begin{align*}
    \ellunderbar(\x_n \vert \theta) :={}
        \ell(\x_n \vert \theta) -
        \expect{\fdist(\xn)}{\ell(\xn \vert \theta)} \quad\textrm{ and }\quad
    \ellbarbar(\xn \vert \theta) :={}
        \ellbar(\xn \vert \theta) -
            \expect{\fdist(\xn)}{\ellbar(\xn \vert \theta)}.
    \end{align*}
    Quantities with overbars depend on $\t$, i.e., on the data distribution at which
    the posterior is evaluated, since posterior centering is at the corresponding
    $\ftdist$.  For compactness we suppress this dependence on $\t$ in the
    notation.
\end{defn}

\begin{lem}\lemlabel{expansion_terms} (See \appref{post_derivs} for a proof.)
    For the mixture distribution
    $\ftdist$, the terms in \eqref{t_expansion} are given by
    \footnote{Note that the second derivative here is
    of the original posterior expectation, which is based on $\fndist$, but the
    derivative is evaluated at $\ftdist$. Therefore the original $\ell(\x_n \vert
    \theta)$ occurs inside the expectation (rather than $\ell(\x_n \vert \theta,
    \t)$), but the posterior expectation is with respect to $\ftdist$, which is
    formed using $\ell(\x_n \vert \theta, \t)$.}
    \begin{align}
    \fracat{\partial \T(\ftdist, N)}{\partial \t}{\t=0} ={}
    \meann \infltil_n
    \textrm{ and }
    \fracat{\partial^2 \T(\ftdist, N)}{\partial \t^2}{\t}
    ={} \frac{1}{N^2} \sumnm \htil(\x_n, \x_m) 
    \textrm{ where} \nonumber \\
    \infltil_n := N \expect{\postf}{\gbar(\theta)\ellbarbar(\x_n \vert \theta)}
    \quad\textrm{and}\quad
    \htil(\x_n, \x_m) :=
    N^2 \expect{\postft}{
        \gbar(\theta)
        \ellbarbar(\x_n \vert \theta)
        \ellbarbar(\x_m \vert \theta)
    }. \eqlabel{v_def}
    \end{align}
    Note that the posterior centering in $\infltil_n$ is implicitly at $\postf$.
\end{lem}


The form of the terms in \lemref{expansion_terms} will allow us
to analyze the limiting behavior of the IJ covariance in more
complicated settings than our results in \secref{finite_dim}.  In particular,
by plugging \lemref{expansion_terms} into \eqref{t_expansion},
\begin{align}\eqlabel{post_expansion}
  \sqrt{N}\left(\expect{\post}{\g(\theta)}  - \expect{\postf}{\g(\theta)}\right) ={}&
  \frac{1}{\sqrt{N}} \sumn \infltil_n + \sqrt{N} \resid{T} \nonumber \\
  \quad\textrm{where}\quad
  \resid{T}(\t) :={}& \frac{1}{N^2} \sumnm \htil(\x_n, \x_m).
\end{align}
We will find
that $\gcovij$ recovers the variance of the limiting distribution of
$\frac{1}{\sqrt{N}} \sumn \infltil_n$, in both the finite--dimensional and
global--local models. Informally, we note that the summands in the first-order
term given in \lemref{expansion_terms} are closely related to our influence
scores, since
\begin{align*}
\expect{\postf}{\gbar(\theta)\ellbarbar(\x_n \vert \theta)} ={}&
\cov{\postf}{\g(\theta), \ellunderbar(\x_n \vert \theta)} := \frac{\infltil_n}{N}.
\end{align*}
In both the the finite--dimensional and global--local models, we will be able to
apply \thmref{bayes_clt_main} to show that the difference between $\infltil_n$
and $\infl_n$ vanishes asymptotically, so that $\gcovij$ consistently estimates
the variance of the limiting distribution of $\frac{1}{\sqrt{N}} \sumn
\infltil_n $.  However, we will show that $\sqrt{N}\resid{T}(\t) \plim 0$ in
finite--dimensional models but diverges in global--local models.

\subsection{Recovering the finite-dimensional consistency}

We briefly state assumptions under which the residual $\sqrt{N} \resid{T}(\t)
\plim{} 0$, proving a version of \thmref{ij_consistent} using our von--Mises
expansion.  Though the proof requires stronger assumptions than
\thmref{ij_consistent}  (requiring five derivatives and additional regularity
assumptions on the prior), it shows that the subsequent
\textit{failure} of the von Mises expansion in global--local models
is not vacuous.


\begin{assu}\assulabel{bayes_clt_vm}
    Assume that
    \begin{enumerate}
    \item The functions $\g(\theta)$, $\ell(\xn \vert \theta)$, and $\g(\theta)
    \ell(\xn \vert \theta)$ are all fifth-order BCLT-okay. 
    \item The prior expectation $\expect{\prior(\theta)}{
        \norm{\g(\theta)}_2^2
        \norm{\ell(\x_n \vert \theta)}_2^2
        \norm{\ell(\x_m \vert \theta)}_2^2
        }$
        is finite $\fdist$-almost
    surely.
    \item The dual sum $\frac{1}{N^2} \sumn \summ
        \expect{\prior(\theta)}{
            \norm{\g(\theta)}_2^2
            \ell(\x_n \vert \theta)^2
            \ell(\x_m \vert \theta)^2} = O_p(1).$
    \end{enumerate}
\end{assu}


\begin{thm}\thmlabel{finite_dim_resid} (See \appref{finite_dim_resid_proof} for a proof.)
Let \assuref{bayes_clt,bayes_clt_vm} hold. Then
\begin{align*}
\sup_{\t \in [0, 1]} \sqrt{N} \resid{T}(\t) \plim{} 0.
\end{align*}
If we additionally assume that $(\infltil_n)^2$ is uniformly integrable
with respect to $\fdist(\x_n)$ as $N \rightarrow \infty$, then
by \eqref{post_expansion}, it follows that
$\sqrt{N}\left(\expect{\post}{\g(\theta)} - \expect{\postf}{\g(\theta)} \right)
\dlim \normdist\left(0, \gcovtrue \right)$ and $\gcovij$ is consistent.  
%
\end{thm}

\subsection{Failure of the von Mises expansion in high dimensions}
\seclabel{von_mises_high_dim}
In this section, we argue that, in the global--local model, the von Mises
expansion residual $\resid{T}(\t)$ does not typically vanish as $N \rightarrow
\infty$, even though the influences scores $\infl_n$ are themselves $\ordp{1}$.
We show these results to be a general property of the Tower property of
expectations, and confirm the result in the specific case of \exref{re}.
However, we argue that the difficulty disappears as the conditional
posterior $\postl$ concentrates, even partially.

First, we note that no general $\sqrt{N}$--rate asymptotic results similar
to \thmref{finite_dim_resid} can be expected for the
expectation of functions of local variables whose posterior does not
concentrate.  To see this, take for the moment $\g(\theta) = g(\lambda)$,
and note that the first term in \eqref{post_expansion} becomes
\begin{align*}
\frac{1}{\sqrt{N}} \sumn \infltil_n ={}&
\meann \sqrt{N} \cov{\p(\lambda \vert \fdist, N)}
    {\g(\lambda), \expect{\p(\gamma \vert \lambda, \fdist, N)}
        {\ellbarbar(\x_n \vert \lambda, \gamma)}} 
        \rightarrow \infty,
\end{align*}
since covariances with respect to $\p(\lambda \vert \fdist, N)$ do not go to
zero as $N \rightarrow \infty$.  (This is in contrast to the finite--dimensional
case, for which posterior covariances go to zero at rate $1/N$ by
\thmref{bayes_clt_main}.) From \eqref{post_expansion} we then might expect that the
distribution of $\sqrt{N}\left(\expect{\post}{\g(\lambda)} -
\expect{\postf}{\g(\lambda)} \right)$ diverges, and it does not make sense to
attempt to estimate the limiting variance.

Suppose we then restrict our attention to functions $\g(\theta) = \g(\gamma)$,
and assume that the \textit{marginal posterior} $\postg$ does concentrate, in
the sense that we can apply results like \thmref{bayes_clt_main} to the marginal
posterior $\postg$, but $\postl$ does not concentrate. 

In \exref{re_gamma} we argued that nested random effect models are an example
when $\postg$ concentrates but $\postl$ does not.  It would be difficult to
state general conditions under which the marginal $\postg$ concentrates in more
complicated settings.  However, since our goal is to demonstrate
\emph{inconsistency} of the von Mises expansion, for our purposes it will
suffice to \emph{assume} that $\postg$ is as regular as would be needed to prove
\thmref{finite_dim_resid} (or something like it).  We will then argue that, even
if we did have adequate concentration of $\postg$, the residual would still be
large.  Specifically, we will make the following heuristic assumption, whose
function will be to render plausible the more specific and technical assumptions
below.
\begin{assu}\assulabel{postg_conc} (Informal)
Assume that $\postgt$ concentrates in the sense that there exists some
$\gamma_0$ such that $\expect{\postgt}{\gamma} \rightarrow \gamma_0$
as $N \rightarrow \infty$.  
Further, if the posterior--mean--zero functions $\bar{a}(\theta)$,
$\bar{b}(\theta)$, and $\bar{c}(\theta)$ are sufficiently regular, and $\post$
satisfies \assuref{bayes_clt,bayes_clt_vm}, then, as $N \rightarrow \infty$,
for some constants $C_2$ and $C_3$,
\begin{align*}
    N \expect{\post}{\bar{a}(\gamma) \bar{b}(\gamma)} 
    \rightarrow
    C_2 \bar{a}_{(1)}(\gamma_0) \bar{b}_{(1)}(\gamma_0)
    \quad\textrm{and}\quad
    N^2 \expect{\post}{\bar{a}(\theta) \bar{b}(\theta) \bar{c}(\theta)} 
    \rightarrow
    C_3 \bar{a}_{(1)}(\gamma_0) \bar{b}_{(1)}(\gamma_0) \bar{c}_{(1)}(\gamma_0).
\end{align*}
%
%
\end{assu}
\Assuref{postg_conc} amounts to a high--level summary of how
\thmref{bayes_clt_main} and \thmref{bayes_clt_vm} in \appref{bayes_clt_vm_proof}
are used in the proofs of \thmref{ij_consistent} and \thmref{finite_dim_resid}.
In the present section, we will simply assume that $\postgt$ satisfies
\assuref{postg_conc}, and argue that the IJ covariance will nevertheless be
inconsistent.

To use \assuref{postg_conc} to analyze $\resid{T}(\t)$, we must write the terms
of \lemref{expansion_terms} in terms of expectations over $\postg$ that we can
analyze.  By the Tower property,
\begin{align}
%
\infltil_n ={}&
N \expect{\postgf}{
    \gbar(\gamma) \expect{\postlf}{\ellbarbar(\x_n \vert \gl)} } 
\quad\textrm{and}\quad 
\eqlabel{highdim-infl} \\
\htil(\x_n, \x_m)
={}&
N^2 \expect{\postgt}{
    \gbar(\gamma)
    \expect{\postlt}{
        \ellbarbar(\x_n \vert \gl)
        \ellbarbar(\x_m \vert \gl)    
    }
}. \eqlabel{highdim-resid}
\end{align}
Here, we see that $\expect{\fdist(\x_n)}{\infltil_n} = 0$ and
$\infltil_n = \ordlogp{1}$, so under uniform square integrability
assumptions we might expect $\frac{1}{\sqrt{N}} \sumn \infltil_n = \ordlogp{1}$
as in the finite-dimensional case.  Thus, for the first term of the
von Mises expansion, little has changed in the high--dimensional case.

However, the residual takes the form of
\begin{align*}
    \resid{T}(\t) = 
    \frac{1}{N^2} \sumnm \htil(\x_n, \x_m) =
    \frac{1}{N^2} \sumnm N^2 \expect{\postgt}{
            \gbar(\gamma)
            \expect{\postlt}{
                \ellbarbar(\x_n \vert \gl)
                \ellbarbar(\x_m \vert \gl)    
            }
        }.
\end{align*}
For a particular $\x_m,\x_n$, $\htil(\x_n, \x_m)$ thus takes the form of $N^2$ times
a posterior covariance with respect to $\postgt$ and so is roughly $\ordp{N^2
N^{-1}} = \ordp{N}$.  Contrast this with the finite--dimensional case, in which
case we had
\begin{align*}
    \htil(\x_n, \x_m) =
    N^2 \expect{\postft}{
            \gbar(\theta)
                \ellbarbar(\x_n \vert \theta)
                \ellbarbar(\x_m \vert \theta)    
            }
\end{align*}
is a posterior expectation of three centered terms, and so was roughly
$\ordlogp{1}$.  We thus expect each term $\htil(\x_n, \x_m)$ in the
high--dimensional case to be roughly a factor of $N$ larger than in the
finite--dimensional case.

Using standard results from V--statistics, we can argue heuristically that this
extra factor of $N$ will be fatal. In particular, by evaluating at $\t = 1$ we
can eliminate the potentially complicated dependence between the posterior
$\postgt$ and the data $\x_n$, in which case $\resid{T}(1)$ takes the form of a
second--order V--statistic with kernel $\htil(\x_n, \x_m) =  N^2
\expect{\postgf}{ \gbar(\gamma) \expect{\postlf}{\ellbarbar(\x_n \vert \gl)
\ellbarbar(\x_m \vert \gl) } }$ (e.g., Section 5 of
\citet{serfling:1980:approximation}). The variance of the corresponding
U--statistic is given by $\expect{\fdistnm}{\htil(\x_n, \x_m)^2} / N^2$ (see
\appref{v_stats} for more details).  Consequently, we would expect
$\var{\xfdist}{\sqrt{N} \resid{T}(1)} = \ordlog{\expect{\fdistnm}{\htil(\x_n, \x_m)^2} /
N}$, which would diverge in the finite--dimensional case but diverge in the
high--dimensional case.

Of course, the preceding argument, though suggestive, does not prove
non--convergence.  In order to be more precise, we now define an exponential family
of high--dimensional global--local models, for which we can derive a precise
limiting distribution for $\sqrt{N} \resid{T}(1)$.

\subsection{High--dimensional global--local exponential families}
\seclabel{high_dim_exp_family}
In the present section, we provide a detailed analysis of $\sqrt{N}\resid{T}(1)$
for a particular class of high--dimensional exponential families.  Though
our analysis is limited to $\t = 1$ in order to decouple the posterior from
the data dependence, our results are highly suggestive of the following conclusions.

First, when $\cov{\postl}{\lambda}$ does not vanish as $N\rightarrow \infty$
and has non--vanishing dependence on $\gamma$, then the von--Mises expansion
residual does not vanish, and the IJ covariance estimate is inconsistent.

Conversely, there exists a constant $\kappa$ such that $\sqrt{N}(\resid{T}(1) -
\kappa)$ converges to an $\ordlogp{1}$ random variable, with both $\kappa$ and
the limiting variance proportional to covariances of the form
$\cov{\postl}{\cdot}$. Consequently, when the spread of $\postl$ is numerically
small, or depends only weakly on $\gamma$, we can expect $\sqrt{N}\resid{T}(1)$
to be numerically small as well, and the IJ covariance reasonably accurate.

We now define notation for global--local exponential families which will apply
for the present section.  
As in \exref{re}, assume that the datapoint $\x_n =
(\a_n, \y_n)$ consists of a group assignment indicator $\a_n$ and a ``response''
vector $\y_n \in \rdom{\ydim}$.  Assume further that, for each $g$, the model is
an exponential family with natural parameter $\eta_g(\gamma, \lambda) \in
\rdom{\ydim}$, so that
\begin{align*}
\ell(\y_n \vert \gamma, \lambda, \a_{ng} = 1) ={}
    \y_n^\trans \eta_g(\gamma, \lambda)
    \quad\Rightarrow\quad
\ell(\x_n \vert \gamma, \lambda) ={}
    \sumg a_{ng} \ell(\y_n \vert \gamma, \lambda, \a_{ng} = 1) ={}
    \sumg a_{ng} \y_n^\trans \eta_g(\gamma, \lambda).
\end{align*}
Unlike \exref{re}, we are \emph{not} assuming that the model is nested --- each
component of $\lambda$ may, in principle, occur in each group's likelihood. Note
also that we assume the log partition function is included in $\eta_g(\gamma,
\lambda)$ if necessary by taking one entry of $\y_n$ to be $1$, and the
corresponding entry of $\eta_g(\gamma, \lambda)$ to be the negative log
partition function. 

\begin{defn}\deflabel{high_dim_exp_terms}
    In the context of the high--dimensional exponential family, define
    \begin{align*}
        \m_g :={}& \expect{\fdist(\y_n | \a_{ng} = 1)}{\y_n} \in \rdom{\ydim}  &
        \S_g :={}& \expect{\fdist(\y_n | \a_{ng} = 1)}{\y_n \y_n^\trans} \in \rdom{\ydim} \times \rdom{\ydim} \\
        \mu_g(\gamma) :={}& \expect{\postlt}{\eta_g(\gl)} \in \rdom{\ydim} &
        \J_{gh}(\gamma) :={}& \cov{\postlt}{\eta_g(\gamma, \lambda), \eta_h(\gamma, \lambda)}
        \in \rdom{\ydim} \times \rdom{\ydim}\\
        \M_{gh} :={}& N^2 \expect{\postgf}{\gbar(\gamma)  \mubar_g(\gamma)  \mubar_h(\gamma)^\trans}
        \in \rdom{\ydim} \times \rdom{\ydim}&
        \L_{gh} :={}& N \expect{\postgf}{\gbar(\gamma)  \J_{gh}(\gamma)}
        \in \rdom{\ydim} \times \rdom{\ydim} \\
        \tilde{y}_{ng} :={}& \sqrt{G}  \a_{ng} \y_n - m_g / \sqrt{G} &
        \rho_{nm} :={}&  \frac{1}{G} \sumg \tilde{y}_{ng}^\trans \L_{gg} \tilde{y}_{mg}.
    \end{align*}
    As part of the definition, we assume that all the above quantities exist
    and are finite for any $N$ and $G$.
\end{defn}


For any particular $g$ and $h$, we have chosen the scaling so that we can expect
both $\M_{gh}$ and $\L_{gh}$ to be $\ordlog{1}$ by
\assuref{postg_conc} and concentration of the marginal posterior
$\postgf$. Further, note that $\expect{\fdist(\x_n)}{\tilde{y}_{ng}} ={} 0$ and
$\expect{\fdist{\x_n}}{\tilde{y}_{ng}\tilde{y}_{ng}^\trans} ={} \S_g$, so that,
asymptotically, we expect $\tilde{y}_{ng}^\trans \L_{gg} \tilde{y}_{ng}$ --- and
its average over $G$, $\rho_n$ --- to be an $\ordlog{1}$ random variable for a
typical $g$ and $n$.

Our key assumptions, \assuref{high_dim_resid_assu,high_dim_resid_moments},
essentially assume that the intuition of \assuref{postg_conc} holds
when averaged in various norms over $g$.


\begin{assu}\assulabel{high_dim_resid_assu}
Let $\littleop{1}$ denote a term that goes to zero in probability as
$G\rightarrow \infty$ and $N \rightarrow \infty$ along any sequence.
Let $\norm{\cdot}_F$ denote the Frobenius norm of the argument.
Assume that, as $N\rightarrow \infty$ and $G \rightarrow \infty$
along any sequence,
\begin{enumerate}
\item \itemlabel{m_matrix_frob}
    $\frac{1}{G^2} \sum_{g=1}^G \sum_{h=1}^G
        \norm{ \S_g^{1/2} \M_{gh} \S_h^{1/2}}_F^2
        = \ordlog{1}$.
\item \itemlabel{l_matrix_frob}
    $\frac{1}{G} \sum_{g=1}^G \sum_{h=1}^G
    \norm{\S_g^{1/2} \L_{gh} \S_h^{1/2}}_F^2
    = \ordlog{1}$.\footnote{
        Note the order $1/G$ in \assuref{high_dim_resid_assu} \itemref{l_matrix_frob},
        in contrast to the order $1/G^2$ in \itemref{m_matrix_frob}. 
        If the random effects are independent a posteriori given $\gamma$, then 
        $\J_{gh}(\gamma) = 0$ for $g \ne h$, and the sum is over only $G$ terms rather
        than $G^2$ terms.  \Assuref{high_dim_resid_assu} \itemref{l_matrix_frob}
        allows for some small amount of posterior covariance between random effects,
        as long as the joint sum is only order $G$ rather than $G^2$.        
    }
\item \itemlabel{m_matrix_avg_op}
    $\frac{1}{G} \sumg \trace{\S_g}  \norm{\M_{gg}}_F = \ordlog{1}$ 
\item \itemlabel{m_matrix_trace}
    $\frac{1}{G} \sum_{g=1}^G \trace{\S_g^{1/2} \M_{gg} \S_g^{1/2}}
    = \ordlog{1}$.
\end{enumerate}
\end{assu}


\begin{assu}\assulabel{high_dim_resid_moments}
Assume that, as $N\rightarrow \infty$ and $G \rightarrow \infty$
along any sequence, for some $\kappa < \infty$,
\begin{align*}
\frac{1}{G} \sum_{g=1}^G \trace{\S_g^{1/2} \L_{gg} \S_g^{1/2}}
    = \kappa + \ordlog{G^{-1/2}}.
\end{align*}
Further, assume that the average within--group third and fourth moments of
$\y_n$ satisfy:
\begin{align*}
\frac{1}{G} 
    \sumg \expect{\fdist(\y_n \vert \a_{ng} = 1)}{
        \norm{\y_n}_2^3 }
    \norm{\m_g}_2 \norm{\L_{gg}}_F^2
    = \ordlog{1}
\quad\textrm{and}\quad
\frac{1}{G} 
    \sumg \expect{\fdist(\y_n \vert \a_{ng} = 1)}{
        \norm{\y_n}_2^4 } \norm{\L_{gg}}_F^2
    = \ordlog{1}.
\end{align*}
\end{assu}
    

To further motivate \assuref{high_dim_resid_assu,high_dim_resid_moments},
we present a simple example before stating our result.

\begin{ex}\exlabel{poisson_re_resid}
\def\r{r}
Take the Poisson random effect regression model with responses $\r_n$, IID Gamma priors.
\begin{align*}
    \p(\lambda_g) =& \mathrm{Gamma}(\alpha, \beta) &
    \p(\r_g \vert \lambda_g, \gamma, \a_{ng}=1) =& \mathrm{Poisson}(\gamma \lambda_g).
\end{align*}
Then, up to a constant $C$ containing terms that do not involve both the data and
parameters,\footnote{
    Terms such as $-\log \y_n!$ that do not depend on the parameters
    vanish in the posterior centering, and terms that do not depend on the data
    vanish in the $\fdist$ centering.  Such terms can be omitted from the
    definitions of $\y_n$ and $\eta$.    
}
\begin{align*}
\ell(\x_n \vert \gamma, \lambda) ={}& 
    \sumg \a_{ng} \left( \r_n (\log(\gamma) + \log(\lambda_g)) - \gamma \lambda_g\right) + C
={}
\sumg \a_{ng} 
\begin{pmatrix}
    \r_n \\ 
    1
\end{pmatrix}^\trans
\begin{pmatrix}
    \log(\gamma) + \log(\lambda_g) \\ 
    - \gamma \lambda_g
\end{pmatrix} + C \Rightarrow \\
\y_{n} ={}& \begin{pmatrix}
    \r_n \\ 
    1
\end{pmatrix} 
\quad\textrm{and}\quad
\eta_g(\gamma, \lambda) ={}
\begin{pmatrix}
    \log(\gamma) + \log(\lambda_g) \\ 
    - \gamma \lambda_g
\end{pmatrix}.
\end{align*}
For a give $G$, all the quantities needed from $\fdist$ for
\defref{high_dim_exp_terms} can be represented as a deterministic set of values 
for $\expect{\fdist(\y_n \vert \a_{ng}=1)}{\r_n}$ and
$\var{\fdist(\y_n \vert \a_{ng}=1)}{\r_n}$.
In particular, since we are evaluating at $\postglf$, the posterior is
proportional to
$\exp\left( 
    \sumn \expect{\fdist(\x_n)}{\ell(\x_n \vert \gamma, \lambda)}
\right) = \exp\left( N \expect{\fdist(\x_n)}{\ell(\x_n \vert \gamma, \lambda)}\right)$,
with
\begin{align*}
N \expect{\fdist(\x_n)}{\ell(\x_n \vert \gamma, \lambda)} ={}&
\frac{N}{G} \sumg 
    \left( \expect{\fdist(\y_n \vert \a_{ng}=1)}{\r_n} 
        (\log(\gamma) + \log(\lambda_g)) - \gamma \lambda_g\right) + C,
\end{align*}
from which we can read off explicit expressions for $\mu_g(\gamma)$ and
$\J_{gh}(\gamma)$ and, in turn, $\M_{gh}$ and $\L_{gh}$ using
\assuref{postg_conc}.  Doing so, we find that $\M_{gh}$ is typically dense,
motivating the $G^2$ scaling in \assuref{high_dim_resid_assu}
\itemref{m_matrix_frob}.  Conversely, $\J_{gh} = 0$ for $g \ne h$ since the
entries of $\lambda$ are in fact nested as in \exref{re}, and so \emph{a
posteriori} independent given $\gamma$, motivating the $G$ scaling
in \assuref{high_dim_resid_assu} \itemref{l_matrix_frob}.
See \appref{poisson_re_resid} for details.
\end{ex}


\begin{thm}\thmlabel{resid_divergence}
    Let \assuref{high_dim_resid_assu, high_dim_resid_moments} hold.  Then, as
    $G\rightarrow\infty$ and $N\rightarrow\infty$,
    \begin{align*}
        \resid{T}(1) ={}  
        \frac{1}{N} \sumnm \rho_{nm} + \ordlog{N^{-1}}
        \quad\textrm{and}\quad
        \expect{\xfdist}{\abs{\resid{T}(1) - \kappa}} \rightarrow 0.
    \end{align*}
    Consquently, if $\kappa \ne 0$, then, with probability approaching one,
    \begin{align*}
        \sup_{\t \in [0,1]} 
            \sqrt{N} \abs{\resid{T}(\t)} \ge 
            \sqrt{N} \abs{\resid{T}(1)} \rightarrow \infty.
    \end{align*}
    Furthermore, if $N / G \rightarrow 0$, then
    \begin{align*}
        \resid{T}(1) = \frac{1}{N} \sumn \rho_{nn} + \ordlog{G^{-1/2}} + \ordlog{N^{-1}},
    \end{align*}
    where $\rho_{nn}$ are IID random variables,
    $\expect{\fdist(\x_n)}{\rho_{nn}} = \kappa + \ordlog{G^{-1/2}}$, and
    $\var{\fdist(\x_n)}{\rho_{nn}} = \ordlogp{1}$.  It follows that $\sqrt{N}
    \left(\resid{T}(1) - \kappa \right)$ converges by the central limit theorem
    to a Normal random variable. See \appref{resid_divergence} for a proof.
\end{thm}


\Thmref{resid_divergence} shows that $\sqrt{N}\resid{T}(1)$ can be expected to
be non--zero with high probability. However, \thmref{resid_divergence} also
states that $\sqrt{N}\resid{T}(1)$ is dominated to leading order by the
stochastic behavior of $\rho_n$, which is propotional to the conditional
posterior covariances $\cov{\postlf}{\eta_g(\gl), \eta_h(\gl)}$, and which does
not depend on the conditional posterior means $\mu_g(\gamma)$. Reducing
$\J_{gh}(\gamma)$ uniformly by a factor of $K$, for example, reduces $\sqrt{N}
\left(\resid{T}(1) - \kappa \right)$ by the same factor $K$.

\Thmref{resid_divergence} does not itself prove the inconsistency of the IJ
covariance, since it leaves open the possibility that $\sqrt{N}\resid{T}(\ttil)
\plim 0$ due to some complex stochastic relationship between the data and
posterior, or between $\ttil$ and the data, though examples such as
\exref{poisson_re_resid} make such a result seem unlikely.  We thus
only conjecture that the above results extend to $\resid{T}{\ttil}$,
leaving such analysis for future work.




\subsection{Weighted posterior}
\seclabel{von_mises_weights}
We will now discuss a well-known connection between the von Mises expansion
result of \secref{von_mises,von_mises_high_dim} to the bootstrap, primarily
to inform our real-world experiments of \secref{experiments}, for which the
ground truth asymptotic variance is unavailable and we compare the IJ covariance
with the more computationally demanding bootstrap covariance.  The connection
is interesting in its own right, however, since the negative
results of \secref{von_mises_high_dim} suggest that the bootstrap variance
will be inconsistent as well in the high-dimensional setting.

We can connect the, bootstrap, IJ covariance, and von Mises expansion in a
unified framework using the notion of {\em data weights}.  Let $\w = (\w_1,
\ldots, \w_N)$ be a vector of scalar weights, with $\onevec := (1, \ldots, 1)$
denoting the vector of all ones.  Define the {\em weighted posterior
expectation} as follows:
\begin{align}\label{eq:postw}
\expect{\postw}{\g(\theta)} =
    \frac{
    \int_{\thetadom} g(\theta)
        \exp\left(\sumn \w_n \ell(\x_n \vert \theta) \right) \prior(\theta)
                        d\theta}
         {\int_{\thetadom}
         \exp\left(\sum \w_n \ell(\x_n \vert \theta) \right) \prior(\theta)
         d\theta}.
\end{align}
Under this definition, $\expect{\postpert{\xvec, \onevec}}{\g(\theta)} =
\expect{\post}{\g(\theta)}$.  If we let $\w_n$ denote the number of times a
datapoint $n$ is selected when drawing indices with replacement from $[N]$, then
$\p(\w) = \textrm{Multinomial}\left(N^{-1}, N\right)$ and the bootstrap
covariance estimate of \algrref{boot} can be represented succinctly as
\begin{align*}
\gcovboot =
N \covhat{b \in [B]}{\expect{\postpert{\xvec, \w^b}}{\g(\theta)}}
\quad\textrm{where}\quad
\w^b \iid \p(\w).
\end{align*}

Of course, evaluating $\expect{\postw}{\g(\theta)}$ at a generic weight vector
$\w$ typically requires running a computationally expensive new MCMC chain as in
\algrref{boot}. Instead, we might form a linear approximation to the effect of
changing the posterior weights.  Define $\T(\w) :=
\expect{\postw}{\g(\theta)}$, and form the linear approximation
\begin{align}
%
\T(\w) ={}&
    \T(\onevec) +
    \fracat{\partial \T(\w)}{\partial\w^\trans}{\w=\onevec}(\w - \onevec) +
    \resid{\wtil}(\w) \quad\textrm{where} \eqlabel{boot_expansion} \\
\resid{\wtil}(\w) :={}&
    \frac{1}{2} (\w - \onevec)^\trans
    \fracat{\partial^2 T(\w)}{\partial \w_n \partial \w_m}{\wtil} (\w - \onevec)
    \quad\textrm{and} \nonumber \\
\wtil ={}& \onevec + \t (\w - \onevec) \textrm{ for some }\t \in [0,1].
\nonumber
\end{align}
Assuming that we can exchange integration and differentiation in the integrals
of \eqref{postw}, the first derivative is given by
\begin{align}\eqlabel{t_w_deriv}
\fracat{\partial \T(\w)}{\partial\w_n}{\wtil} =
\cov{\postwtil}{\g(\theta), \ell(\x_n \vert \theta)}
\quad\Rightarrow\quad
\fracat{\partial \T(\w)}{\partial\w^\trans}{\w=\onevec} = \frac{1}{N}\infl,
\end{align}
where we have taken $\infl = (\infl_1, \ldots, \infl_N)$ to be the $\gdim \times
N$ matrix of stacked influence scores as given in \eqref{ij_terms}.
Noting that, by standard properties of the multinomial distribution,
$\expect{\p(\w)}{\w} = \onevec$ and  $\cov{\p(\w)}{\w} = I_N + N^{-1}\onevec
\onevec^\trans$, it follows that we can also write the IJ covariance using the
weighted posterior as
\begin{align*}
\gcovij
= N \infl\, \cov{\p(\w)}{\w - \onevec} \infl^\trans =
\cov{\p(\w)}{
    \sqrt{N}
    \fracat{\partial \T(\w)}{\partial\w^\trans}{\w=\onevec}(\w - \onevec)}.
\end{align*}
The IJ covariance is thus precisely a linear approximation to the bootstrap
covariance of the rescaled statistic $\sqrt{N} (\T(\w) - \T(\onevec))$.

As in the von Mises expansion, we can evaluate the accuracy of this linear
approximation using the second order term of the Taylor series expansion. By
again differentiating \eqref{t_w_deriv} with respect to $\w$ by exchanging
integration and differentiation, we see that
\begin{align*}
\fracat{\partial^2 T(\w)}{\partial \w_n \partial \w_m}{\w}
= \expect{\postw}{\gbar(\theta)
    \ellbar(\x_n \vert \theta) \ellbar(\x_m \vert \theta)}
\end{align*}
so that
\begin{align}\eqlabel{boot_resid}
\sqrt{N} \resid{\wtil}(\w) ={}&
\frac{1}{N^2} \sumnm \sqrt{N} \htil(\w_n, \w_m | \x_n, \x_m, \wtil) \textrm{ where}\\
\htil(\w_n, \w_m | \x_n, \x_m, \wtil) :={}&
    N^{5/2} \expect{\postwtil}{\gbar(\theta)
      \ellbarbar(\x_n \vert \theta) \ellbarbar(\x_m \vert \theta)}
          (\w_n - 1) (\w_m - 1).
\end{align}

As with \eqref{v_def}, \eqref{boot_resid} takes the form of a V-statistic with a
kernel consisting of $N^{5/2}$ times a posterior expectation over $\gbar(\theta)
\ellbarbar(\x_n \vert \theta) \ellbarbar(\x_m \vert \theta)$.  All the analysis
of the residuals $\resid{\T}(\t)$ \secref{von_mises,von_mises_high_dim}
applies, only with $\postwtil$ replacing $\postft$.  This is, of course,
not a coincidence: \eqref{boot_expansion} is exactly analogous to
\eqref{t_expansion}, only with $\post$ in place of $\postf$ and $\postw$
in place of $\post$.

We thus expect the IJ to be a good approximation to the bootstrap in the same
settings where the IJ is a good approximation to the truth: namely, settings
where the posterior expectation $\expect{\postpert{\mathbb{G},N}}{\gbar(\theta)
\ellbarbar(\x_n \vert \theta) \ellbarbar(\x_m \vert \theta)}$ is of order
$N^{-2}$, uniformly for $\mathbb{G}$ in a neighborhood of $\fndist$.  In
particular, in high-dimensional models, we expect the IJ to form an inconsistent
approximation to the bootstrap, in the sense that the two do not coincide as $N
\rightarrow \infty$. This connection motivates our use of the bootstrap as a
``ground truth'' for our real-data experiments of \secref{experiments}.

Further, the failure of the residual to vanish may imply that the bootstrap
itself is also inconsistent in posteriors which fail to concentrate, given the
role of von Mises expansions in many proofs of consistency of the bootstrap
variance estimator (see, e.g., \cite[Chapter 3]{shao:2012:jackknife}). However,
we leave detailed exploration of the consistency of the bootstrap in
high-dimensional settings for future work.


\section{Experiments}
\seclabel{experiments}


We will now evaluate the performance of the IJ on several real-world models and
datasets. Since we are using real data, we will not have access to a ground
truth, so we will check the accuracy of the IJ covariance estimate by comparison
with the more computationally expensive bootstrap estimate.  All our estimates
will be computed with Monte Carlo (both MCMC and Monte Carlo bootstrap draws),
and we will account for this Monte Carlo error in our comparisons.

In each experiment, we will fix a number of MCMC samples, denoted $M$, and a
number of bootstraps, denoted $B$.  From a single run of MCMC, we compute the
Bayesian posterior estimate, $\gcovbayeshat$, and the IJ covariance estimate,
$\gcovijhat$.  Using \algrref{boot}, we draw $B$ different bootstrap samples,
computing $M$ MCMC samples for each, from which compute the bootstrap variance
estimate, $\gcovboothat$.

All three of $\gcovbayeshat$, $\gcovijhat$, and $\gcovboothat$ are estimated
with Monte Carlo sampling error. Since our experiments use a finite $M$,
$\gcovbayeshat$ and $\gcovijhat$ are MCMC estimators of $\gcovbayes$ and
$\gcovij$, respectively.  Similarly, our $B$ bootstrap samples are thought of as
a random draw from the set of $B_{max} := N^N / N!$ possible bootstrap samples,
and so $\gcovboothat$ is a ``doubly Monte-Carlo'' estimator of $\gcovboot :=
\lim_{B \rightarrow B_{max}} \lim_{M \rightarrow \infty} \gcovboothat$.  In
general, for $B < B_{max}$, we expect bootstrap sampling variability to remain
even in $\lim_{M \rightarrow \infty} \gcovboothat$. Similarly, for $M < \infty$,
we expect MCMC sampling variability to remain even in $\lim_{B \rightarrow
B_{max}} \gcovboothat$.  We will shortly discuss how to estimate these Monte
Carlo errors.

Even when the conditions of \thmref{ij_consistent} are met and the IJ covariance
is consistent, we do not expect $\gcovij = \gcovboot$ even with $B = B_{max}$
and $M = \infty$, since we have a finite number $N$ of data points, and the
bootstrap and IJ covariance are formally distinct (though consistent) estimators
of $\gcovtrue$.  The difference $\gcovij - \gcovboot$ should decrease as $N$
grows, but, for any finite $N$, the difference can be expected to exceed Monte
Carlo error for sufficiently large $M$ and $B$. Similarly, even in the correctly
specified case, we do not expect $\gcovbayes = \gcovhat$, as, again, these are
distinct (though consistent) estimators of $\info^{-1}$.

One might attempt to additionally account for frequentist sampling variability
of the original data from the target population in $\gcovij$ or $\gcovboot$ as
estimators of $\gcovtrue$, or of $\gcovbayes$ as an estimator of $\info^{-1}$.
However, we expect the quantities $\gcovij$, $\gcovboot$, and $\gcovbayes$ to be
highly correlated under random sampling of the data from the population
distribution, and it is necessary to take this correlation into account when
computing the variability of, say, the IJ ``error'' $\gcovijhat - \gcovboothat$.
The authors do not see a simple, computationally tractable way to account for
this correlation, so we neglect this error in our experiments, and account only
for the Monte Carlo error of the preceding paragraph.  Since we consider the IJ
covariance  to succeed when the difference $\gcovijhat - \gcovboothat$ is small
relative to random variability, neglecting the frequentist sampling error is
conservative.

To estimate the Monte Carlo standard errors, we model the individual entries of
the matrices $\gcovbayeshat$, $\gcovijhat$, and $\gcovboothat$ as normal
distributions with Monte Carlo standard deviations given by the matrices
$\sebayes$, $\seij$, and $\seboot$, respectively.\footnote{Since both
$\gcovbayeshat$ and $\gcovijhat$ are based on the same MCMC samples, their Monte
Carlo errors are correlated, typically positively.  However, to avoid cumbersome
exposition of the experimental results, we will report standard errors for
$\gcovbayeshat$ and $\gcovijhat$ as if they were independent, since accounting
for this correlation will tend to reduce the standard errors of the difference
$\gcovbayeshat - \gcovijhat$, and the difference $\gcovbayeshat - \gcovijhat$
tends to be large in our experiments, even relative to the conservative estimate
of independent standard errors.
}
To estimate $\sebayes$ and $\seij$, we block bootstrap the MCMC chains using
block sizes that are large relative to estimates of the MCMC mixing times.
Note that block bootstrapping here requires only re-arranging the original MCMC
chains, which is fast, and does not require drawing any new MCMC samples.  To
estimate $\seboot$, we use the fact that each bootstrap sample $\mathscr{G}^b$
from \algrref{boot} contains independent MCMC errors, and so the marginal
distribution of the $\mathscr{G}^b$ can be used to estimate the variability both
of the bootstrap sampling and the MCMC.  We do so using the delta method, as
follows.  For any two scalar random variables, $t_1$ and $t_2$ drawn according
to a joint distribution $p(t_1, t_2)$, we can define $h(a, b, c) :={} a - bc$
and write
\begin{align*}
\cov{p(t_1, t_2)}{t_1, t_2} &={}
    \expect{p(t_1, t_2)}{t_1 t_2} -
    \expect{p(t_1, t_2)}{t_1}\expect{p(t_1, t_2)}{t_2}\\
&={} h\left(\expect{p(t_1, t_2)}{t_1 t_2},
     \expect{p(t_1, t_2)}{t_1},
     \expect{p(t_1, t_2)}{t_2} \right).
\end{align*}
By applying the delta method to the function $h$, we can estimate the sampling
variability of $\cov{p(t_1, t_2)}{t_1, t_2}$ from the covariance of the vector
$(t_1 t_2, t_1, t_2)$, which, in turn, can be estimated from draws of $t_1$ and
$t_2$.  We apply the delta method in this way to each entry of the matrix
$\gcovboothat$.

\subsection{Simulated Poisson Experiments}
\seclabel{poisson_experiments}
\PoissonREGraph{}

We begin with a simulated experiment based on \exref{re}, in which we keep $N /
G$ fixed, increase $N$, and compare the IJ and bootstrap covariance to a ground
truth estimated by simulation.  We demonstrate the inconsistency of the IJ
covariance as $N$ grows when $N / G$ is small, as well as an improvement in the
accuracy of the IJ for larger values of $N / G$, long before $N / G$ approaches
infinity.

We consider $N$ responses $\y = (y_1, \ldots, \y_N)$ which are
Poisson-distributed conditional on $G$ scalar-valued random effects $\lambda =
(\lambda_1, \ldots, \lambda_G)$.  Let $\a_n \in \{1 ,\ldots, G \}$ denote the
random effect index for observation $n$, writing $\a = (\a_1, \ldots, \a_N)$.
Formally, our model is
\begin{align*}
\y_n \vert \lambda, \gamma, \a \iid{}&
    \mathrm{Poisson}(\exp(\gamma + \lambda_{\a_n}) ) \\
\exp\left(\lambda_g\right) \iid{}& \mathrm{Gamma}(\alpha, \beta)\\
\gamma \sim{}& \mathrm{Unif}(-\infty, \infty).
\end{align*}
We chose $\alpha$ and $\beta$ so that $\expect{\p(\lambda_g)}{\exp(\lambda_g)} =
\reZPriorMean$ and $\var{\p(\lambda_g)}{\exp(\lambda_g)} = \reZPriorSD^2$.   We will
take our quantity of interest to be $g(\gamma) = \gamma$ for simplicity.

For concreteness, one might imagine that $y_n$ describes a count of a
charismatic animal observed at site $\a(n)$, where each site has a different,
random baseline preponderance of the animal given by its corresponding
$\exp(\lambda_g)$. The IJ covariance measures the ``conditional variance'' of
collecting new observations at a randomly chosen collection of the same sites,
in contrast to the ``marginal variance'' of collecting new observations at new,
unseen sites.

Given values for $N$, $G$, we fix $\gamma = 1.5$,  and simulate $1 + \reNumSims$
datasets according to the model given above. We use the first simulated dataset
as the ``observed dataset,'' using it to compute the reported Bayesian posterior
variance and IJ variance, as well as to compute $\reNumBoots$ bootstrap
variance, all of $\expect{\post}{\theta}$.  We use the remaining $\reNumSims$
datasets to estimate the true frequentist variance of $\expect{\post}{\theta}$,
an estimate which we call $\hat{V}_{\mathrm{sim}}$.

\Figref{poisson_re_graph} shows the results.  When $N/G = 1$, i.e., there is
only one observation per random effect on average, we expect that $\p(\lambda_g
\vert \xvec)$ is fairly dispersed and, by the arguments of
\secref{von_mises_high_dim}, we expect the IJ covariance to be inconsistent.
As expected, the first row of \figref{poisson_re_graph} shows that the IJ
covariance under-estimates the ground truth simulated variance, even as
$N$ grows large.  In contrast, when $N/G = 10$, even though we do not expect
the IJ covariance to be consistent, we expect $\p(\lambda_g
\vert \xvec)$ to be more concentrated than in the $N/G=1$ case, and so,
in turn, for the IJ covariance to be more accurate.  Although there is
evidently some bias in the IJ covariance when $N/G=10$, the IJ variance is
clearly a better estimate than when $N/G=1$.

We note that the bootstrap is often within the Monte Carlo error bars of the
ground truth, even when $N/G=1$, although it is evident from the first row of
\figref{poisson_re_graph} that the bootstrap remains biased, even as $N$ grows,
as expected from the discussion in \secref{von_mises_weights}. For $N/G = 10$,
the performance of the bootstrap and IJ are essentially the same.

\subsection{Applied regression and modeling textbook}
\seclabel{arm_experiments}

\subsubsection{Models and data}

\ARMTable{}

We ran the experiments from \cite{gelman:2006:arm} for which data was available
in the Stan examples repository \cite{stan-examples:2017}, and for which
\texttt{rstanarm} \cite{rstanarm} was able to run in a reasonable amount of
time, resulting in comparisons between $\armNumCovsEstimated$ distinct
covariance estimates from $\armNumModels$ different models using
$\armNumDatasets$ distinct datasets. The median number of observations amongst
the models was $\armMedianNumObs$, and the range was from $\armMinNumObs$ to
$\armMaxNumObs$ exchangeable observations.  Since we fit our models in
\texttt{rstanarm}, we were able to use the function \texttt{rstanarm::log\_lik}
to compute $\ell(\x_n \vert \theta)$, so $\gcovijhat$ required essentially
no additional coding.  We used the default priors provided by \texttt{rstanarm}.

Every model we considered from \cite{gelman:2006:arm} were instances of
generalized linear models, including both ordinary and logistic regression as
well as models with only fixed effects and models that include random effects,
as shown in \tabref{arm_models}.  For parameters of interest, $g(\theta)$, we
took all fixed-effect regression parameters, the log of the residual variance
(when applicable), and the log of the random effect variance (when applicable).
We did not examine random effects, nor the random effect covariances for the few
models with multiple random effects, out of concern that a BCLT might not be
expected to hold for these parameters (in the notation of global--local models,
these parameters might be considered part of local parameters $\lambda$).

For models with no random effects, we define an exchangeable unit to be a tuple
containing a particular response and a regressor.  For random effects models, we
had to choose an exchangeable unit.  When the random effect had many distinct
levels, we considered the groups of observations within a level to be a single
exchangeable unit.  For models whose random effects had few distinct levels, we
took the exchangeable unit to either be a more fine-grained partition of the
data, or, in some cases treated the observations as fully independent.


\ARMGraphZ{}

\ARMGraphDiff{}

\subsubsection{Reporting metrics}


With $\armNumCovsEstimated$ covariances to report, we need to summarize our
results into a few high-level accuracy summaries.  We report the following two
summary measures, intended to capture the statistical and practical significance
of discrepancies between $\gcovbayeshat$, $\gcovijhat$ and $\gcovboot$.

First, to quantify the ``statistical significance'' of $\gcovijhat - \gcovboot$,
we compute the matrix $\zdiff$, whose $i,j$-th entry is simply the z-statistic
which we would compute to test equality between $\gcovijhat_{ij}$ and
$\gcovboot_{ij}$:
\begin{align*}
\zdiff_{ij} :=
\frac{\gcovijhat_{ij} - \gcovboothat_{ij}}
     {\sqrt{(\seij_{ij})^2 + (\seboot_{ij})^2}}.
\end{align*}
Of course, the z-statistics are not independent within the matrix $\zdiff$, so
we should think of aggregate summaries of elements of the various models'
$\zdiff$ matrices as heuristics which summarize the error $\gcovijhat -
\gcovboot$ relative to Monte Carlo error rather than formal hypothesis tests.

Second, we wish to quantify the ``practical significance'' of the differences
$\gcovijhat - \gcovboot$.  For example, if the Monte Carlo error is very small,
then $\zdiff$ might be very large even when $\gcovijhat - \gcovboot$ is small
for practical purposes.  Conversely, $\zdiff$ can be very small simply because
the Monte Carlo error is high, and we don't necessarily want to reward large
$\seij$.  To quantify the practical significance, we compute the relative
errors,
\begin{align*}
\normdiffij_{ij} :=
     \frac{\gcovijhat_{ij} - \gcovboothat_{ij}}
     {\left| \gcovboothat_{ij} \right| + \seboot_{ij}}
\quad\textrm{and}\quad
\normdiffbayes_{ij} :=
  \frac{\gcovbayeshat_{ij} - \gcovboothat_{ij}}
  {\left| \gcovboothat_{ij} \right| + \seboot_{ij}}.
\end{align*}
The matrices $\normdiffij$ and $\normdiffbayes$ measure the relative differences
between the bootstrap and $\gcovijhat$ or $\gcovbayeshat$, respectively.  We add
the standard deviation $\seboot$ in the denominator to avoid problems dividing
by zero.\footnote{Note that the relative errors $\normdiffij$ and
$\normdiffbayes$ do not attempt to remove or account for Monte Carlo error.  One
might imagine computing the relative error after deconvolving the Monte Carlo
error and covariance estimates using, say, nonparametric empirical Bayes.  We
investigated such an approach, but the resulting procedure was complicated and
the results qualitatively similar to \figref{normerr_graph} below, so we will
report only the relatively simple metrics $\normdiffij$ and $\normdiffbayes$.}


\subsubsection{Results}

We computed $M = \armNumMCMCSamples$ MCMC draws from each model, which took a
range of times, from $\armMinMCMCTimeSecs$ seconds to $\armMaxMCMCTimeMins$
minutes. The total sampling for the initial MCMC and IJ computation was
$\armTotalMCMCTimeMins$ minutes.  For each model, we computed $B =
\armNumBootstraps$ bootstrap samples, for a total bootstrap compute time of
approximately $\armTotalBootTimeHours$ hours.

The distribution of the measure of practical significance, $\normdiffij$ and
$\normdiffbayes$, is shown in \figref{normerr_graph}.  For smaller datasets, the
IJ covariance is not meaningfully closer to the bootstrap than the Bayesian
covariance is. However, for larger datasets, the IJ is a notably better
approximation to the bootstrap than is the Bayesian posterior covariance,
particularly for covariances involving at least one scale parameter.  Most of
the IJ covariance estimates are well within 50\% of the bootstrap covariance,
particularly for regression parameters.

The comparisons for $\zdiff$, our informal test of statistical significance are
shown in \figref{relerr_graph}. The IJ produces results similar to the bootstrap
for models with $N$ greater than the median and for regression parameters.  For
datasets with fewer exchangeable units, the IJ and bootstrap often differ to a
degree greater than can be accounted for by the standard error, probably because
there is not enough data for a BCLT to practically apply.  For larger datasets,
the tests lie outside the standard rejection region at roughly the nominal
rate, suggesting that the differences between $\gcovijhat$ and $\gcovboothat$
can mostly be accounted for by Monte Carlo error.  We stress again that the
tests are correlated within a matrix $\zdiff$, and such collective measures
should not be expected to be a valid family of hypothesis tests.
%
%

\subsection{Pilots data}
\seclabel{arm_pilots}

\PilotsSEGraph{}

The ``Pilots'' experiment from \cite[Section 13.5]{gelman:2006:arm}, originally
published in \cite{gawron:2003:flight}, is a particularly apt
illustration of our methods, so we will briefly discuss it in detail.  The data
and model consists of a non-nested linear random effects model, and exhibits the
following properties:
\begin{itemize}
\item Marginal maximum likelihood (\texttt{lme4}) fails due a singular fit,
which is attributable to low variance in one of the groups.  In contrast,
\texttt{rstanarm} is able to produce draws from the posterior with good MCMC
diagnostics.\footnote{ In the original \cite[Section 13.5]{gelman:2006:arm}, it
appears that the authors were able to run \texttt{lme4} on the same model we
discuss here without encountering a singular fit, though with a very small
estimated variance for the component which we find to be singular.  However,
running the identical command to theirs on a more recent version ($lmerVersion$)
of the \texttt{lme4} results in a singular fit.  We conjecture that some aspect
of \texttt{lme4} may have changed since the publication of
\cite{gelman:2006:arm}.}
\item The MCMC chain takes long enough that the bootstrap samples took a
non-trivial amount of compute time: $\armPilotMCMCTimeSecs$ seconds for MCMC,
leading to $\armPilotBootTimeMins$ minutes for $\armPilotNumBoots$ bootstrap
samples.  (The nonlinear scaling is a consequence of some bootstrap samples
taking longer to converge than the original data.)
\item The IJ and bootstrap covariances differ meaningfully from the Bayesian
posterior, and in a way that could affect the interpretation of the results.
\item For at least one parameter, the frequentist variance exceeds the Bayesian
posterior variance, despite the former being interpreted as a conditional
variance and the latter a marginal variance under correct specification.
\item The IJ and bootstrap differ more than can be accounted for by random
sampling, but are in qualitative agreement, particularly relative to Bayes.
\end{itemize}
We acknowledge that our detailed analysis of the Pilots experiment is
cherry-picking from the ARM models. Almost all the other models in the
\cite{gelman:2006:arm} could be fit by \texttt{lme4}, producing results in
qualitative agreement with the posterior computed by \texttt{rstanarm}.  Nor did
the models of \cite{gelman:2006:arm} typically exhibit such large discrepancies
between Bayesian and frequentist variability (though one can certainly detect
the discrepancies in aggregate, as we show in \secref{arm_experiments}).
Nevertheless, the Pilots experiment exhibits the kind of behavior that {\em may}
be possible, particularly in cases of non-nested random effects, which were not
widely represented in \cite{gelman:2006:arm}.

\PilotsIntervalsGraph{}

The dataset consists of $N = \armPilotNumObs$ distinct pilots (indexed here by
$n$), who were grouped into $K=\armPilotNumGroups$ (indexed here by $k$)
according to their amount of prior experience and education recovering from
emergency situations.  Each pilot was then tested in $J=\armPilotNumScenarios$
different emergency scenarios (indexed here by $j$), after each of which an
examiner determined whether the pilot was able to successfully recover the
aircraft.  The observed ``response'' $\y_{kj}$ is the proportion of pilots in
group $k$ who were able to successfully recover from scenario $j$.  The
study's authors are interested in the effect of training on the ability of
a pilot to recover, relative to the ordinary variability in the types of
emergency scenarios pilots might encounter.  

A non-nested linear random effects model is fit with a global offset and random
effects for both the airport and condition.  Our ``global'' parameters of
interest, $\gamma$, will be the global offset ($\mu$), and scale parameters for
the per-observation residual ($\sigma_y$) and airport and scenario random
effects ($\sigma_{grp}$ and $\sigma_{scn}$, respectively). The
$\armPilotNumScenarios \times \armPilotNumGroups$ actual random effects for each
scenario constitute our ``local'' parameters, $\lambda$.

Formally, the model is:
\begin{align*}
\gamma ={}& (
    \mu,
    \sigma_{y},
    \sigma_{grp},
    \sigma_{scn})
\\
\lambda ={}& (\lambda^{grp}_{1}, \ldots, \lambda^{grp}_J,
              \lambda^{scn}_{1}, \ldots, \lambda^{scn}_K)
\\
\y_n ={}& \mu +     \lambda^{grp}_{j[n]} +
    \lambda^{scn}_{k[n]} + \varepsilon_n
\\
\lambda^{grp}_{j} \iid{}& \normdist\left(0,  \sigma_{grp}^2\right)
\quad
\lambda^{scn}_{k} \iid{} \normdist\left(0,  \sigma_{scn}^2\right)
\quad\varepsilon_n \iid{} \normdist\left(0,  \sigma_{y}^2\right).
\end{align*}

In this model, we are interested in the relative sizes of the scale parameters
$\sigma_{grp}$, $\sigma_{scn}$, and $\sigma_{y}$, since these quantify the size
of the effect on pilot performance due to pilot group (i.e. education), the
difference between scenarios, and idiosyncratic, individual variation,
respectively. With this interpretation in mind, we take our quantities of
interest as $\g(\gamma) = \left(\log \sigma_{scn} - \log \sigma_{grp}, \log
\sigma_y - \log \sigma_{grp}, \log \sigma_y - \log \sigma_{scn}\right)$.

Since the data consists of forty different pilots, one might reasonably ask how
much our conclusions might have changed had we randomly selected new pilots from
the same population and asked them to compute the same set of tasks.  This very
reasonable notion of sampling variability corresponds to conditional sampling.
We expect the IJ and bootstrap to estimate this conditional sampling
variability, whereas the Bayesian posterior variance should estimate
(asymptotically) the marginal variance, i.e., the variability we would expect
from, in some sense, selecting both new pilots and a new set of scenarios.  We
thus expect the frequentist and Bayesian intervals to differ.

This model has a large number of fully crossed random effects, the number of
which ($\lambdadim = \armPilotNumGroups + \armPilotNumScenarios$) is large
relative to the number of data points ($N = \armPilotNumObs$), so by the
reasoning of \secref{von_mises_high_dim,von_mises_weights}, we expect the
bootstrap and IJ covariance to differ, even in the absence of MCMC noise, but
that the IJ covariance will be a reasonable approximation to the bootstrap.

In \figref{width_graph}, we see the width of the standard errors (i.e., the
square root of the variance, all divided by $\sqrt{N}$) for each of the sampling
methods. In some cases, the IJ and bootstrap covariances are fairly different,
as might be expected given the low number of observations per random effect for
this model.  \Figref{width_graph} shows that the differences in the uncertainty
estimates matter, particularly for the key quantity of interest, $\log
\sigma_{scn} - \log \sigma_{grp}$.  Although the Bayesian posterior credible
interval overlaps with zero --- a conclusion that admits the possibility no
difference between the size of the effects --- we see that we are very unlikely
to measure a point estimate of $\expect{\post}{\log \sigma_{scn} - \log
\sigma_{grp}} = 0$ under re-sampling of pilots while keeping the scenarios and
group definitions fixed.

\subsection{Leisleri's bats}
\seclabel{bats}

\subsubsection{Model and data}

\BatsData{}

Certain ecological datasets take the form of observation histories of labeled
individuals, and a key analysis task with such datasets is to infer population
quantities such as survival rates, taking into account the fact that individuals
may be present but not observed.  We consider one such dataset consisting
whether or not $\batsNumObs$ Leisleri's bats were observed on $\batsNumTimes$
distinct survey occasions, obtained from the datasets accompanying
\cite[Chapter 7]{kery:2011:bpa} (see also \cite{schorcht:2009:bats}).
Following \cite[Chapter 7]{kery:2011:bpa}, we model the data using a ``Cormack
Jolly Seber'' model with random time effects for survival probability.  We will
estimate the variability due to random sampling of the bats, for fixed
population parameters and at the same times.

We are motivated to compute frequentist covariances both because the simple
model is very likely an incomplete picture of the actual complicated population
dynamics, and because the bats are reasonably thought of as a random sample from
a larger population, e.g., from a contiguous section of forest with identical
properties that was not surveyed.

The analysis task is to infer the probability that a bat is observed given that
it is present, and the base rate and dispersion of the survival probabilities at
each time.  Formally, we are interested in the following quantities:
\begin{align*}
p_{\mathrm{mean}} :=&{}
    \textrm{ The (constant) probability of observation given presence}\\
\phi_{\mathrm{mean}} :=&{}
    \textrm{ The base probability that a bat survives from one period
            to the next}\\
\sigma_y :=&{} \textrm{ The variability in survival probability over time}.
\end{align*}

The generative model is then defined as follows. Each of $p_{\mathrm{mean}}$,
$\phi_{\mathrm{mean}}$ are given uniform priors in $(0, 1)$, and $\sigma_y$ is
given an improper uniform prior on $(0, \infty)$.  For each time period,
$t=1,\ldots,\batsNumTimes$, the survival probability for that time period
is determined by an IID ``random effect'' draw as follows:
\begin{align*}
\epsilon_t | \phi_{\mathrm{mean}}, \sigma_y \iid&
    \normdist\left(\cdot \vert
        \log\left(\frac{\phi_{\mathrm{mean}}}{1 - \phi_{\mathrm{mean}}}\right),
        \sigma_y\right)\\
\phi_{t} :=& \frac{\exp(\epsilon_t)}{1 + \exp(\epsilon_t)}.
\end{align*}
Given that a bat is observed in time period $t$, the probability that it is next
observed in the period $t + 1$ is then the probability that it survived times
the probability that it was observed, i.e., $\phi_{t} \cdot p_{\mathrm{mean}}$.
Similarly, the probability that it is next observed in the period $t + 2$ is the
probability that it survives to time $t + 2$, that it is not observed at time $t +
1$, and is observed at time $t + 2$, i.e., is $\phi_{t} \cdot \phi_{t + 1} \cdot
(1 - p_{\mathrm{mean}}) \cdot p_{\mathrm{mean}}$, and so on.  In this way, the
probability of the sequence of observations of the entire dataset can be
expressed in terms of the global model parameters $p_{\mathrm{mean}}$,
$\sigma_y$, and $\phi_{\mathrm{mean}}$.

Given the reasoning of the previous paragraph, it is computationally convenient
to express the dataset as a ragged array of counts, where each row corresponds
to a capture time, and each column corresponds to the first re-observation after
the capture time.  The final column counts all bats which were never
re-observed.  In this way, each individual's observation in all but the final
period increments by one exactly one cell of the ragged array by one.
Conversely, we can evaluate $\ell(x_n | \z, \theta)$ by adding up the log
probabilities of each cell in the ragged array to which individual $n$
contributes. Our particular dataset is shown in ragged array form in
\figref{bats_data}.

In our notation, $\theta = (p_{\mathrm{mean}}, \phi_{\mathrm{mean}}, \sigma_y)$,
$\z = (\epsilon_1, \ldots, \epsilon_T)$, and we take $g(\theta) =
(p_{\mathrm{mean}}, \phi_{\mathrm{mean}}, \log(\sigma_y))$.  We take a
datapoint, $\x_n$, to be all the observations for a single bat.  We are thus
evaluating the variability of sampling the bats, keeping the time periods and
their random effects fixed.

\subsubsection{Modeling results}

\BatsResults{}

As in the simpler example of \secref{arm_pilots}, the joint MAP estimate
(as computed by \texttt{rstan::optimizing}) estimates $\hat\sigma_y = 0$, and so
is degenerate (the Hessian of the log joint distribution at the MAP is
singular).  Furthermore, the returned MAP estimates $\hat{p}_{\mathrm{mean}} =
\batsOptMeanP$ and $\hat{\phi}_{\mathrm{mean}} = \batsOptMeanPhi$ are absurd,
probably due to the inability of the optimization algorithm to deal with the
poorly formed objective.  In order to account for survival probabilities that
are variable in time, we are thus constrained to use MCMC.

We ran MCMC model using \texttt{Stan} and a model implementation from the Stan
Examples Wiki \cite{stan-examples:2017}, slightly modified to track the log
probabilities $\ell(x_n | \z, \theta)$ needed for the IJ computation. We
computed a total of $\batsNumMCMCDraws$ draws, which took $\batsMCMCTime$
seconds. We ran the bootstrap for $B = \batsNumBoots$.  We computed standard
errors for the IJ and Bayesian covariances using the block bootstrap with
$\batsNumSEBlocks$ blocks and $\batsNumSEDraws$ draws.  As a robustness check,
we verified that the standard error estimates were similar those computed with
twice and four times as many blocks.

The base posterior mean estimates were
\begin{align*}
\expect{\post}{p_{\mathrm{mean}}} \approx \batsMeanMeanP \quad\quad
\expect{\post}{\phi_{\mathrm{mean}}} \approx \batsMeanMeanPhi \quad\quad
\expect{\post}{\log (\sigma_y)} \approx \batsMeanLogSigma,
\end{align*}
which closely match the results from \cite[Chapter 7.11]{kery:2011:bpa}.

\Figref{bats_result} shows the relationship between the diagonals of
$\gcovbayeshat$, $\gcovijhat$, and $\gcovboothat$.  The first row of
\figref{bats_result} shows the estimated variances, i.e., the diagonal entries
of the covariance matrices. For $p_{\mathrm{mean}}$ and $\phi_{\mathrm{mean}}$,
the frequentist variances are markedly different than the Bayesian posterior,
and the bootstrap and IJ agree closely.  For $\log (\sigma_y)$, the IJ
covariance is differs from the bootstrap by an amount greater than the estimated
standard errors, but the IJ is still a better estimate of the frequentist
variance than the Bayesian posterior. The second row of \figref{bats_result}
shows the estimated covariances, i.e., the off-diagonal entries of the
covariance matrices.  The covariances mostly differ from one another by less
than the estimated Monte Carlo error, except for the covariance between
$\phi_{\mathrm{mean}}$ and $\log \sigma_y$.

Computing the initial MCMC run took $\batsMCMCTime$ seconds, and the additional
time to compute $\gcovijhat$ was negligible.  In contrast, computing
$\gcovboothat$ took $\batsBootTime$ seconds.  The IJ produced results
very similar to the bootstrap, especially considering Monte Carlo error,
but approximately $\batsBootOverMCMCTime$ times faster.

\subsection{Election modeling}
\seclabel{election}

\subsubsection{Model and data}

The Economist website published predictions of the outcome of the 2020
presidential race using a Bayesian model implemented in \texttt{Stan}
\cite{economist:2020:election}. The model and data for some previous election
cycles were made publicly available on the website, and we examined the model's
predictions of the 2016 presidential election.  The dataset consisted of
$\electionNumObs$ national polls conducted over the course of the 2016 election.
We examined the election-day posterior mean vote percentage for the democratic
candidate both nationally and in each of the 51 US ``states'' (for the duration
of this section, I will describe Washington D.C. as a ``state'' for simplicity
of exposition), and consider the frequentist variability under re-sampling of
the $\electionNumObs$ polls. By re-sampling the polls rather than the individual
responses, we allow for the possibility that respondents are not necessarily
independent within a poll while still capturing the sampling variability within
the polls.

Based on \cite{linzer:2013:dynamic}, the model views the raw poll results as
noisy draws from the logistic transform of a random walk in time, plus random
effects to account for various aspects of the poll.  The random effects
accounted for include the noisy observations from a random walk of unobserved
``true'' vote percentages, pollster identity (there were $\electionNumPollsters$
distinct pollsters), the population queried (e.g., likely voters, registered
voters, etc.), and the mode of the poll (e.g. online, phone, etc.).  Since all
the polls were national, the vote proportions in the individual states are not
informed directly by the data.  Rather, the poll probabilities are drawn
randomly from the logistic transform of a national random walk (plus random
effects), and the national random walk is modeled in turn as a weighted sum of
individual state random walks that follow a 51-dimensional autoregressive
process with {\em a priori} fixed covariance.

The prior distribution on the ``true'' vote percentage random walks is centered
on a fundamentals prediction on election day, and then follows an autoregressive
process backwards in time.  Effectively, polls are shrunk more towards the
fundamentals prediction near election day, and shrunk less the further the poll
is from the election.  The random walk prior also effectively smooths the time
series of predictions.  The fundamentals-based election day prediction is,
again, estimated from data outside the model and treated as fixed in the present
analysis.  See \cite{gelman:1993:elections} for the reasoning being this
time-reversed random walk prior.  
%


The distinction between $\gamma$ and $\lambda$ in the election model is less
clear than in explicit random effects like those of \secref{bats,
arm_experiments}. Our quantities of interest, the election day vote percentages,
can be derived entirely from the state-level\footnote{Washington DC is counted
as a ``state,'' so $\gamma$ is length 51.} random walks on election day. So we
take
\begin{align*}
\gamma :={}& \textrm{State-level random walk values on election day}\\
\lambda :={}& \textrm{All other parameters (including all other random walk values)}\\
\rho_s :={}& \left(\frac{\exp(\gamma_s)}{1 + \exp(\gamma_s)} \right)
 = \textrm{Expected state vote proportions}\\
w_s :={}& \textrm{State population proportions (fixed {\em a priori})} \\
\rho_{\mathrm{nat}} :={}& \sum_{s=1}^{51} w_s \rho_s
    = \textrm{Expected national vote proportion} \\
g(\gamma) :={}& \left(\rho_{\mathrm{nat}}, \rho_1, \ldots, \rho_{51}\right).
\end{align*}

We take a datapoint, $\x_n$, to be the value of one of the $\electionNumObs$
national polls.  Note that, by doing so, we consider multiple observations from
the same pollster to be exchangeable.  Conceptually, we are conditioning on the
identities of the pollsters, and randomly drawing a new population of polls from
the fixed set of pollsters.

In this case, the data is literally a random sample from a larger population,
and the variability under that sampling procedure is naturally an interesting
quantity. Furthermore, the model is surely an incomplete description of the true
underlying relationship between the voting population and the polls.
Consequently, it is particularly interesting to ask whether the posterior
uncertainty is greater than or less than the frequentist variability due to the
random sampling that we know gave rise to the observed data.  In particular, if
the frequentist sampling variability is much smaller than the posterior
uncertainty, then the posterior uncertainty is dominated by model uncertainty
that is not being strongly informed by the data.

\subsubsection{Modeling results}

We computed $\electionNumMCMCDraws$ MCMC draws using \texttt{Stan}, which took
approximately $\electionMCMCHours$ hours of compute time.   We computed
$\electionNumBoots$ bootstraps, for a total bootstrap compute time of
approximately $\electionBootHours$ hours. As expected, the bootstrap took
roughly $\electionBootOverMCMC \approx \electionNumBoots$ times longer than the
initial MCMC run.  There was considerable variability in the MCMC runtime of the
bootstrap, and a few slow runs of MCMC increased the average bootstrap time
relative to the run on the original dataset.

\Figref{election_result} shows results for the national vote, and the states
$\mathrm{\electionBestState}$ and $\mathrm{\electionWorstState}$, which had the
smallest and largest relative error between the IJ and bootstrap, respectively.
Most importantly, observe that the Bayesian variance is over an order of
magnitude larger than both estimates of the frequentist variance.  In this
model, the posterior uncertainty is dominated by the uncertainty in its latent
parameters, strongly suggesting that additional polling data {\em from the same
distribution as the data} will be relatively unhelpful for reducing the
posterior variability.  In particular, one might expect to require more sampling
near the election date, and more repeated sampling from existing pollsters; such
sampling is different from more samples from the same distribution that gave
rise to our empirical distribution on $\x_n$ because we are effectively
conditioning on the distribution of the polls amongst pollsters and time points.

\POTUSStatesGraph{}

\POTUSScatterGraph{}

In \figref{election_result}, we see from the error bars that the IJ and
bootstrap differ by more than can be accounted for by Monte Carlo sampling, even
in the national vote and best-performing state, $\mathrm{\electionBestState}$.
Nevertheless, the qualitative conclusion drawn by $\gcovijhat$ is correct, which
is that the frequentist variability is much smaller than the Bayesian posterior
variability in this case.  \Figref{election_state_result} shows that the rest
of the states show the same qualitative performance.

The reason for the discrepancy between the IJ and the bootstrap in this case is
not clear to the authors, but we can put forward some possibilities.  One
possibility is that a BCLT simply does not apply to this dataset for lack of
data.  Due to the time series nature of the model, one might expect that early
polls do not have a strong effect on the prediction on election day, and that
the effective number of samples informing our election-day quantity of interest
is in fact much smaller than $\electionNumObs$.  Similarly, the fact that each
pollster has its own random effect, but many pollsters are observed only a few
times in this dataset, may mean that the marginal likelihood of the data is not
strongly affected by the national vote on election day.  Any such arguments made
for the national vote would apply {\em a fortiori} for the states, since the
data consisted only of national-level polls, and an individual state's
parameters are necessarily even less informative of the data than any
national-level parameters.  Another possibility, as discussed in the beginning
of \secref{experiments}, could be that $\gcovij$ and $\gcovboot$ (of which our
computed $\gcovijhat$ and $\gcovboothat$ are Monte Carlo estimates) may be quite
different from one another due to $N$ being sufficiently far from $\infty$.

Nevertheless, we emphasize that the IJ provided the qualitatively correct result
in this case --- namely, that the frequentist variance is a small fraction of
the posterior variance, and essentially no additional time relative to the
initial MCMC run, compared over a thousand hours of compute time required for
the bootstrap.

\section{Conclusion}
We define and analyze the IJ covariance estimate of the frequentist variance of
Bayesian posterior expectations.  The IJ covariance is computationally
appealing, since it can be computed easily from only a single set of MCMC draws,
in contrast to the bootstrap, which, in general, requires many distinct MCMC
chains.  In finite dimensional problems obeying a central limit theorem, we show
that the IJ covariance is consistent, since it consistently estimates the
sandwich covariance matrix.  However, in high-dimensional problems in which some
parameters' posterior do not concentrate, we prove using a von Mises calculus
that the IJ covariance is inconsistent in general, even as an estimator for
parameters which concentrate marginally.  Nevertheless, we argue theoretically
and demonstrate in real-life settings that the IJ covariance remains a useful
estimator even when the posterior does not fully concentrate, particularly
given its computational simplicity.

\clearpage

\ifbool{jrssb}{
    \bibliographystyle{abbrvnat}
    \bibliography{references}
} {
    \printbibliography{}
}

\newpage

\ifbool{jrssb}{
    \begin{appendices}
} {
    \appendix
}
\section*{Appendices}


\section{Additional notation and lemmas}

\subsection{Additional notation}

\subsubsection{Uniform laws of large numbers}

Several of our key assumptions take the form of \emph{uniform laws of large
numbers} (ULLNs), a concept which we make concrete in \defref{ulln}.
Typically, ULLNs require that the sample average $\meann \phi(\x_n, \theta)$
converge to $\expect{\fdist(\xn)}{\phi(\x_n, \theta)}$ uniformly over $\theta$
in some compact set.
%
\begin{defn}\deflabel{ulln}
For $\delta > 0$, denote the $\delta$-ball in $\theta$ as $\thetaball{\delta} :={}
\left\{\theta: \norm{\theta - \thetatrue}_2 \le \delta \right\}$. We say a
``uniform law of large numbers'', or ``ULLN'', holds for a function
$\phi(\theta, \xn)$ if, for some $\delta > 0$,
\begin{align*}
\sup_{\theta \in \thetaball{\delta}}
\norm{\meann \phi(\theta, \x_n) -
\expect{\fdist(\xn)}{\phi(\theta, \xn)}}_2
\plim 0,
\end{align*}
where $\expect{\fdist(\xn)}{\phi(\theta, \xn)}$ is finite.
\end{defn}


In general, a ULLN requires that the set of functions $\left\{\theta \mapsto
\phi(\theta, \xvec): \norm{\theta - \thetatrue}_2 < \delta \right\}$ be a
``Glivenko-Cantelli class'' with finite covering number (\cite[Chapter
19]{van:2000:asymptotic}, \cite{vaart:2013:empiricalprocesses}). A common
condition for a uniform law of large numbers to hold is that $\phi(\theta, \xn)$
be dominated by an absolutely integrable function of $\xn$ alone in an open
neighborhood of $\thetatrue$.

%
%
%

\subsubsection{Derivative arrays}

Throughout our proofs, we will be forming and manipulating higher-order
Taylor expansions of functions of $\theta$.  The arrays of derivatives
in these expressions will always be summed against a direction, typically
$\tau = \sqrt{N}(\theta - \thetahat)$.  In order to avoid cumbersome
summation notation for all such terms, we will implicitly sum derivative arrays
against $\tau$ or $\theta$.  or example, for
a function $\psi(\theta, \x_n) \in \mathbb{R}^{D_\psi}$, we will write
\begin{align*}
\psigrad{3}(\thetahat, \x_n)\tau^3 =
\sum_{a=1}^D \sum_{b=1}^D \sum_{c=1}^D
    \fracat{\partial^3 \psi(\theta, \x_n)}
           {\partial\theta_{a} \partial\theta_{b} \partial\theta_{c}}
           {\thetahat} \tau_{a} \tau_{b} \tau_{c}.
\end{align*}
Occasionally we will need to sum several derivative arrays against $\tau$,
particularly in \secref{region_r1}.  Since the directions $\tau$ will always be
the same, we will combine the $\tau$ in such expressions with no ambiguity. For
example, for the term involving $\likhatk{3}(\thetahat)$ in the statement of
\thmref{bayes_clt_main}, we will write
\begin{align*}
\psigrad{1}(\thetahat, \x) \likhatk{3}(\thetahat) \tau^4 =
\sum_{a=1}^{D} \sum_{b=1}^{D} \sum_{c=1}^{D} \sum_{d=1}^{D}
    \psigrad{1:a}(\thetahat, \x) \likhatk{3:b,c,d}(\thetahat)
    \tau_a \tau_b \tau_c \tau_d.
\end{align*}

When necessary, higher
derivatives will be indexed in the parentheses.  For example,
$\psigrad{2:i,j}(\theta, \x_n)$ will denote the length $D_\psi$ vector
\begin{align*}
\fracat{\partial^2 \psi(\theta, \x_n)}
    {\partial\theta_i \partial\theta_j}{\theta, \x_n}.
\end{align*}
%

We denote by $\ind{A}$ the indicator function taking value $1$ when the
event $A$ occurs and $0$ otherwise.

We will use order notation to describe the behavior of quantities as $N
\rightarrow \infty$.  For example, we will use $z = O(\phi(N))$ to mean there
exists an $N^*$ and an $A > 0$ such that $N > N^* \Rightarrow z < A \phi(N)$. In
our notation, $A$ and $N^*$ in this definition will never depend on the data,
$\xvec$, though they may depend on unobserved population quantities.
Dependence on the data $\xvec$ in order notation will always be represented with
$O_p$ notation.  For example, $z = O_p(\phi(N))$ will mean that there exists an
$A > 0$ such that for any $\epsilon > 0$, there is a corresponding $N^*$ and
such that
\begin{align*}
N > N^* \quad\Rightarrow\quad
    \expect{\xfdist}{\ind{z < A \phi(N)}} > 1 - \epsilon.
\end{align*}
Here, $A$ does not depend on $\epsilon$.  The meaning is that the event $z =
O(\phi(N))$ occurs with arbitrarily high probability as $N$ goes to infinity.
The notation $\tilde{O}$ and $\tilde{O}_p$ will have the same meaning as $O$ and
$O_p$ notation respectively, but with factors of $\log N$ possibly ignored.

\subsection{Lemmas and Definitions}\applabel{lemmas}

\begin{defn}\deflabel{taylor_residual}
For a $k$-times differentiable function $\phi(\theta)$, a fixed $\thetahat$, and
and $\theta$, define the residual functional
\begin{align*}
\tresid{k}{\phi, \thetahat, \theta} :=
    \frac{1}{(k-1)!} \int_0^1 (1 - t)^{k-1}
        \phigrad{k}(\thetahat + t (\theta - \thetahat)) dt.
\end{align*}
\end{defn}
The residual $\tresid{k}{\phi, \thetahat, \theta}$ is an array of the same
dimension as the $k$-th order derivative, $\phigrad{k}(\theta)$.  Observe that,
for any norm $\norm{\cdot}$,
\begin{align*}
    \norm{\tresid{k}{\phi, \thetahat, \theta}} \le
        \frac{1}{(k-1)!} \sup_{\thetatil: \norm{\thetatil - \thetahat}_\infty \le
                                          \norm{\theta - \thetahat}_\infty}
                            \norm{\phigrad{k}(\thetatil)}.
\end{align*}
\begin{lem}\lemlabel{taylor_residual}
A $K$-times continuously differentiable function $\phi(\theta)$ has Taylor
series residual
\begin{align*}
\phi(\theta) =
    \sum_{k=0}^{K-1} \frac{1}{k!} \phigrad{k}(\thetahat) (\theta - \thetahat)^k +
    \tresid{K}{\phi, \thetahat, \theta} (\theta - \thetahat)^K.
\end{align*}
\begin{proof}
Let $\phi_d(\theta)$ denote the $d$-th component of the vector $\phi(\theta)$
(contrast this notation with for the $d$-th order derivative
$\phigrad{d}(\theta)$). Consider the scalar to scalar map $t \mapsto
\phi_d(\thetahat + t (\theta - \thetahat))$. Taking the integral remainder form
of the $K - 1$-th order Taylor series expansion of this scalar map
\cite[Appendix B.2]{dudley:2018:analysis} and stacking the result into a vector
gives the desired result.
\end{proof}
\end{lem}

Expanding \defref{loglik_def}, define the following quantities.

\begin{defn}\deflabel{loglik_details_def}
\begin{align*}
\likhatk{k}(\theta) :={}&
   \frac{1}{N} \sumn \ellgrad{k}(\x_n \vert \theta)
&
\likk{k}(\theta) :={}& \int \ellgrad{k}(\x \vert \theta) d\fdist(\x)
\\
\normalizer :={}& \det\left(2\pi \infohat^{-1} \right)^{1/2} &
\normhatarg{\tau} :={}&
    \normalizer^{-1} \exp\left(-\frac{1}{2} \infohat \tau^2 \right)
    = \normdist(\tau \vert 0, \infohat^{-1})
\\
\infoev :={}& \textrm{The minimum eigenvalue of }\info &
\infoevhat :={}& \textrm{The minimum eigenvalue of }\infohat
\end{align*}
\end{defn}
Additionally, our residuals will be defined in terms of the following
truncated normal distribution.

\begin{defn}\deflabel{residual_distribution}
Let $\normresid{\lambda}{S}(\tau)$ denote the $\normdist\left(\tau | 0,
\lambda^{-1} I_D \right)$ distribution truncated to the set $S$,
where $I_D$ denotes the $D \times D$ identity matrix.
Specifically,
\begin{align*}
\normresid{\lambda}{S}(\tau) :=
\frac{
    \exp\left( - \frac{\lambda}{2} \norm{\tau}_2^2 \right)
}{
    \int_{S}
        \exp\left( - \frac{\lambda}{2} \norm{\tau}_2^2 \right) d\tau
}.
\end{align*}
\end{defn}

We next show that the domination condition of the derivatives of the log
likelihood --- specifically, \assuref{bayes_clt}
\itemref{finite_dim,loglik_smooth,loglik_ulln} --- suffices to apply various
uniform laws.

\begin{lem}\lemlabel{ulln}
Suppose that there exists $\deltaulln > 0$, $M(\xn)$, and $K \ge 1$ such that
\begin{align}
\sup_{\theta \in \thetaball{\deltaulln}}
    \norm{\ellgrad{K}(\xn \vert \theta)}_2 \le M(\xn)
\quad\textrm{with}\quad
\expect{\fdist(\xn)}{M(\xn)^2} \le \infty.
\eqlabel{bounded_lipschitz}
\end{align}
Then
\begin{align}
\sup_{\theta \in \thetaball{\deltaulln}}
    \norm{\meann \left( \ellgrad{k}(\x_n \vert \theta) -
           \expect{\fdist(\xn)}{\ellgrad{k}(\xn \vert \theta)} \right)
           }_2 &\plim{} 0
\nonumber\\&\hspace{-3em}
\textrm{for }k \in \{0, \ldots, K - 1\}
          \eqlabel{glivenko_cantelli} \\
\sup_{\theta \in \thetaball{\deltaulln}}
    \norm{\frac{1}{\sqrt{N}} \sumn \left( \ellgrad{k}(\x_n \vert \theta) -
           \expect{\fdist(\xn)}{\ellgrad{k}(\xn \vert \theta)}\right)
           }_2 &={} \ordp{1}
\nonumber\\&\hspace{-3em}
   \textrm{for }k \in \{0, \ldots, K - 1\}
           \eqlabel{donsker} \\
\sup_{\theta \in \thetaball{\deltaulln}}
    \meann \norm{\ellgrad{k}(\x_n \vert \theta)}_2^2
     &={} \ordp{1}
\nonumber\\&\hspace{-3em}
\textrm{for }k \in \{0, \ldots, K\}
     \eqlabel{means_op1}.
\end{align}
\begin{proof}
First, note that
\begin{align*}
\sup_{\theta \in \thetaball{\deltaulln}}
    \meann \norm{\ellgrad{K}(\x_n \vert \theta)}_2^2 \le{}&
\meann \sup_{\theta \in \thetaball{\deltaulln}}
    \norm{\ellgrad{K}(\x_n \vert \theta)}_2^2
\quad\textrm{and}\\
\expect{\fdist(\xn)}{
    \sup_{\theta \in \thetaball{\deltaulln}}
        \norm{\ellgrad{K}(\xn \vert \theta)}^2_2
} \le{}& \expect{\fdist(\xn)}{M(\xn)^2} < \infty,
\end{align*}
so \eqref{means_op1} holds for $k = K$ by the strong law of large numbers.

We now prove the remaining relations inductively. For a scalar-valued,
differentiable function $f(\theta)$, and any two $\theta, \theta'$ in
$\thetaball{\deltaulln}$, we can write
\begin{align*}
\abs{f(\theta) - f(\theta')} ={}&
\abs{
    \int_0^1 \fracat{\partial f(\theta' + \t (\theta - \theta'))}
                    {\partial \t}{\t} d\t
}
\\={}&
\abs{
    \int_0^1 f_{(1)}(\theta' + \t (\theta - \theta')) (\theta - \theta') d\t
}
\\\le{}&
    \norm{\theta - \theta'}_2
    \sup_{\thetatil \in \thetaball{\deltaulln}}
        \norm{f_{(1)}(\thetatil)}_2.
\end{align*}
Applying the previous formula componentwise to $\ellgrad{K-1}(\xn \vert
\theta)$, and loosening the right hand side to contain all of
$\norm{\ellgrad{K}(\xn \vert \thetatil)}_2$ rather than only a particular
component, gives
\begin{align*}
\norm{\ellgrad{K-1}(\xn \vert \theta)}_\infty \le{}&
    \norm{\theta - \theta'}_2
    \sup_{\thetatil \in \thetaball{\deltaulln}}
        \norm{\ellgrad{K}(\xn \vert \thetatil)}_2
\\\le{}&
\norm{\theta - \theta'}_2 M(\xn).
\end{align*}
It follows that each component of $\ellgrad{K-1}(\xn \vert \theta)$ is both
Glivenko-Cantelli and Donsker, and so \eqref{glivenko_cantelli,donsker} hold for
$k = K - 1$ (see \cite[Example 19.7]{van:2000:asymptotic}).  Further,
\begin{align*}
\sup_{\theta \in \thetaball{\deltaulln}}
    \norm{\ellgrad{K-1}(\xn \vert \theta)}_2
\le{}&
\sup_{\theta \in \thetaball{\deltaulln}}
    \sqrt{\thetadim} \norm{\ellgrad{K-1}(\xn \vert \theta)}_\infty
\le{}
    2 \sqrt{\thetadim}  \deltaulln M(\xn).
\end{align*}
It follows that \eqref{means_op1} holds for $k = K - 1$.

Additionally, \eqref{bounded_lipschitz} holds for $k=K-1$, so
the remaining results follow by proceeding inductively on $K-1$, $K-2$,
and so on.

\end{proof}

\end{lem}

\subsection{Derivatives of posterior expectations}\applabel{post_derivs}
\def\et#1{\expect{\p(\theta \vert \t)}{#1}}
\def\phibar{\bar{\phi}}

Consider an exponentially tilted distribution,
\begin{align*}
\p(\theta \vert \t) :=
    \frac{\exp\left(
        \ell(\theta) + \t \phi(\theta)
        \right)}
        {\int \exp\left(
            \ell(\theta') + \t \phi_k(\theta')
            \right) d\theta'}
\end{align*}
which we assume is defined for all $\t$ in a neighborhood of the zero vector.
When we can exchange the order of integration and differentiation,
\begin{align*}
\fracat{d \et{\g(\theta)}}{d \t}{\t}
={}&
\et{\g(\theta) \phi(\theta)} - \et{\g(\theta)} \et{\phi(\theta)}.
\end{align*}
Applying the previous formula to itself gives
\begin{align*}
\fracat{d^2 \expect{\p(\theta \vert \t)}{\g(\theta)}}{d \t^2}{\t}
={}&
\et{\g(\theta) \phi(\theta) \phi(\theta)} -
\et{\g(\theta) \phi(\theta)} \et{\phi(\theta)} -
\\&
\et{\g(\theta)} \left(
    \et{\phi(\theta) \phi(\theta)}  - \et{\phi(\theta)}\et{\phi(\theta)}
\right) -
\\&
\et{\phi(\theta)} \left(
    \et{\g(\theta) \phi(\theta)} - \et{\g(\theta)} \et{\phi(\theta)}
\right).
\end{align*}
If, for some $\t_0$, we set
\begin{align*}
\gbar(\theta) :={} \g(\theta) - \expect{\p(\theta \vert \t_0)}{\g(\theta)} \\
\phibar(\theta) :={} \phi(\theta) -
    \expect{\p(\theta \vert \t_0)}{\phibar(\theta)},
\end{align*}
then the above formulas simplify at $\t_0$ to
\begin{align*}
\fracat{d \et{\g(\theta)}}{d \t}{\t=\t_0} ={}&
    \et{\gbar(\theta) \phibar(\theta)} \quad\textrm{and}\quad
\fracat{d^2 \expect{\p(\theta \vert \t)}{\g(\theta)}}{d \t^2}{\t=\t_0} ={}
    \et{\gbar(\theta) \phibar(\theta) \phibar(\theta)}.
\end{align*}
Importantly, the validity of the preceding expression for the second derivative
requires that the centering is at a fixed $\t_0$ --- the same $\t_0$ at which
the derivative is evaluated --- and does not depend on $\t$.

\section{Expansion of posterior moments}\applabel{bayes_clt_proof}
The proof of \thmref{bayes_clt_main} will be broken up into several parts. Our
first step is to derive an explicit series expansion for posterior expectations
of the form $\expect{\post}{\phi(\theta)}$, culminating in \lemref{bayes_clt}.
In \lemref{resid_op1}, we then consider the dependence of $\phi(\theta)$ on a
subset of datapoints, writing $\phi(\theta, \xvec_s)$ (recalling
\defref{index_sets}), and showing that the results of \lemref{bayes_clt} can be
summed over $\indexset$. The residual of the expansion in \lemref{bayes_clt} is
designed so that \lemref{resid_op1} follows easily.

To prove \lemref{bayes_clt}, we follow classical proofs of the BCLT, keeping
explicit track of the residual.  We first express the desired expectation as the
ratio of integrals with respect to $\tau$.  We then divide the domain of
integration into three regions according to their distance from $\thetahat$ and
treat each separately. Finally, we combine the bounds on the integrals over
different regions for our final result.

Throughout, we will use the change of variable
\begin{align*}
\theta := \thetahat + \sqrt{N} \tau
\quad\Leftrightarrow\quad
\tau = N^{-1/2} (\theta - \thetahat).
\end{align*}
We will switch the domain of integration between $\tau$ and $\theta$ without
re-writing the arguments.  For example, for some function $\phi(\theta)$, we
will write the shorthand $\int \phi(\theta) d\tau$ for the more precise $\int
\phi\left(\sqrt{N} \tau + \thetahat\right) d\tau$. Similarly, we will overload
set notation: if $A$ is a set of $\theta$, we will write $\tau \in A$ as a
shorthand for $\tau \in \left\{\tau: \thetahat + \sqrt{N} \tau \in A \right\}$.
Although slight abuses of notation, these conventions will keep our expressions
more compact, as well as serve as a reminder of the $\theta$ domain in which the
integrand is smooth by assumption.


\subsection{General index sets}\seclabel{index_sets}

To succinctly state our theorems for double sums, it will be convenient to
introduce a more general index notation.

\begin{defn}\deflabel{index_sets}
    Let $\indexset$ denote a set of sets of data indices, where each member $s \in
    \indexset$ has the same cardinality, and for $s \in \indexset$, let $\xvec_s$
    denote the set of corresponding datapoints.  We write $\phi(\xvec_s)$ for a
    function that depends on $|s|$ datapoints.  For example, if $\indexset = \{1,
    \ldots, N \} =: [N]$, then $\xvec_s = \x_s$, and $\means \phi(\xvec_s) = \meann
    \phi(\x_n)$.  Similarly, if $\indexset = [N] \times [N]$, then $\xvec_s =
    (\x_{s_1}, \x_{s_2})$, and $\means \phi(\xvec_s) = \frac{1}{N^2} \sumnm
    \phi(\x_n, \x_m)$.
\end{defn}

We now state a generalization of \defref{bclt_okay} that applies to 
generic index sets as defined in \defref{index_sets}.

\begin{defn}\deflabel{bclt_okay_index}
For an index set $\indexset$, a data-dependent function
$\phi(\theta, \xvec_s)$ is $K$-th order BCLT-okay if,
for each $k = 1, \ldots, K$,
\begin{enumerate}
\item\itemlabel{bclt_okay_as}
    $\theta \mapsto \phigrad{k}(\theta, \xvec_s)$ is BCLT-okay
    $\fdist$-almost surely
\item\itemlabel{prior_op1}
    $\means \expect{\prior(\theta)}{\norm{\phi(\theta, \xvec_s)}^2_2} = O_p(1)$.
\item\itemlabel{grad_op1}
For some $\delta > 0$,
$\sup_{\theta \in \thetaball{\delta}}
    \means \norm{\phigrad{k}(\theta, \xvec_s)}^2_2 = O_p(1)$.
\end{enumerate}
\end{defn}

Note that a function is BCLT--okay according to \defref{bclt_okay} precisely if
it is BCLT--okay according to \defref{bclt_okay_index} with index set $\indexset
= [N]$.  Although $\indexset = [N]$ and $\indexset = [N]\times[N]$ will be the
only examples of $\indexset$ we will use in the present paper, we believe it is
clearer to keep the notation general.   

In particular, \thmref{bayes_clt_main} is stated in terms of $\indexset = [N]$,
but we will prove that it holds for generic index sets of the form
\defref{bclt_okay_index}.  Formally, we will prove \thmref{bayes_clt_main} under
the following \assuref{core_bclt_assu}, noting that \assuref{core_bclt_assu} is
assumed by \thmref{bayes_clt_main} for $\indexset = [N]$.

\begin{assu}\assulabel{core_bclt_assu}
Assume that $\phi(\theta, \xvec_s)$ is third-order BCLT-okay
according to \defref{bclt_okay_index}, and that \assuref{bayes_clt} holds.
\end{assu}

\subsection{Proof of \propref{freq_clt}}\applabel{freq_clt_proof}
\

\begin{lem}\lemlabel{thetahat_consistent}
Under \assuref{bayes_clt}, $\thetahat \plim \thetatrue$, and
$\norm{\infohat^{-1}}_{op} < 2 \infoev^{-1}$ , with probability approaching one.
\begin{proof}
%
By \assuref{bayes_clt} \itemref{strict_opt}, $\info$ is positive definite with
$\norm{\info^{-1}}_{op} = \norm{\likk{2}(\thetatrue)}_{op} = \infoev^{-1}$.
Choose $\delta$ small enough that $\likk{2}(\theta)$ is continuous by
\assuref{bayes_clt} \itemref{loglik_smooth}, then choose $\delta$ smaller still
so that, by continuity of the operator norm,
\begin{align*}
\sup_{\theta \in \thetaball{\delta}}
    \norm{\likk{2}(\theta)^{-1}}_{op} < \sqrt{2} \infoev^{-1}.
\end{align*}
Next, choose $\delta$ sufficiently small that, by the ULLN of
\assuref{bayes_clt} \itemref{loglik_ulln} and \lemref{ulln},
\begin{align*}
\sup_{\theta \in \thetaball{\delta}}
    \norm{\likhatk{2}(\theta)^{-1}}_{op} < 2 \infoev^{-1},
\end{align*}
with probability approaching one.  For the remainder of the proof,  we will
condition on the event that the preceding display holds.

Using \lemref{taylor_residual} to Taylor expand $\likhatk{1}$ around
$\thetahat$ and evaluate at $\thetatrue$ gives
\begin{align}\eqlabel{likscore_taylor_expansion}
\likhatk{1}(\thetatrue) ={}&
    \likhatk{1}(\thetahat) +
    \tresid{2}{\likhat, \thetahat, \thetatrue} (\thetahat - \thetatrue).
\end{align}
Now,the minimum eigenvalue of the matrix $\tresid{2}{\likhat, \thetahat, \theta}$
is controlled by
\begin{align*}
\inf_{v: \norm{v}_2=1}
    v^T \tresid{2}{\likhat, \thetahat, \theta} v ={}&
\inf_{v: \norm{v}_2=1}
    \int_0^1 (1 - t)
        v^T \likhatk{2}(\thetahat + t (\thetatrue - \thetahat)) v dt \\
\ge&
    \inf_{v: \norm{v}_2=1}
    \inf_{\theta \in \thetaball{\delta}} v^T \likhatk{2}(\theta) v \\
\ge& \frac{\infoev}{2}.
\end{align*}
Consequently, the matrix $\tresid{2}{\likhat, \thetahat, \theta}$ is invertible
with $\norm{\tresid{2}{\likhat, \thetahat, \theta}^{-1}}_{op} \le
2 \infoev^{-1}$. Applying to the Taylor expansion, and using the fact that
$\likhatk{1}(\thetahat) = 0$, gives that
\begin{align*}
\norm{\thetahat - \thetatrue}_2 \le 4 \infoev^{-1} \norm{\likhatk{1}(\thetatrue)}_2.
\end{align*}
Finally, by a ULLN applied to $\likhatk{1}$, we have
$\likhatk{1}(\thetatrue) \plim \likk{1}(\thetatrue) = 0$, concluding the proof.
\end{proof}
\end{lem}

\paragraph{Proof of \propref{freq_clt}.}
\begin{proof}
The proof uses \eqref{likscore_taylor_expansion} in the proof of
\lemref{thetahat_consistent}.  First, let $\delta_N := \norm{\thetahat - \thetatrue}_2$,
so that, by \lemref{thetahat_consistent}, $\delta_N \plim 0$.

Then, observe that
\begin{align*}
\MoveEqLeft
\norm{\tresid{2}{\likhat, \thetahat, \theta} - \likk{2}(\thetatrue)}_2
\\={}&
\norm{\int_0^1 (t - 1)
\left(
    \likhatk{2}(t (\theta - \thetahat) + \thetahat) dt - \likk{2}(\thetatrue)
\right) dt
}_2
\\ \le&
\sup_{\theta \in \thetaball{\delta_N}}
    \norm{\likhatk{2}(\theta) - \likk{2}(\thetatrue)}_2
\\ \le&
\sup_{\theta \in \thetaball{\delta_N}}
    \norm{\likhatk{2}(\theta) - \likhatk{2}(\thetatrue)}_2 +
\sup_{\theta \in \thetaball{\delta_N}}
    \norm{\likhatk{2}(\thetatrue) - \likk{2}(\thetatrue)}_2
\\ &\plim 0,
\end{align*}
where the limit to zero follows from continuity of $\likhatk{2}$ via
\assuref{bayes_clt} \itemref{loglik_smooth} for the first term and the ULLN via
\assuref{bayes_clt} \itemref{loglik_ulln} and \lemref{ulln} for the second term.
It follows that
\begin{align*}
\tresid{2}{\likhat, \thetahat, \theta}^{-1} \plim
    \likk{2}(\thetatrue)^{-1} = -\info^{-1}.
\end{align*}

Next, by \assuref{bayes_clt} \itemref{loglik_ulln} and \lemref{ulln},
$\expect{\fdist(\x)}{\norm{\ellgrad{1}(\x \vert \thetatrue)}^2_2}$ is finite
with probability approaching one, so $\cov{\fdist(\x)}{\ellgrad{1}(\x \vert
\thetatrue)} = \scorecov$ is finite. By smoothness of $\lik$
(\assuref{bayes_clt} \itemref{loglik_smooth}) and the definition of
$\thetatrue$, $\expect{\fdist(\x)}{\ellgrad{1}(\x \vert \thetatrue)} = 0$. The
$\x_n$ are IID, so, by a central limit theorem,
\begin{align*}
\sqrt{N} \likhatk{1}(\thetatrue) =
    \frac{1}{\sqrt{N}} \sumn \ellgrad{1}(\x_n \vert \thetatrue)
\dlim \normdist(\cdot | 0, \scorecov).
\end{align*}
By \eqref{likscore_taylor_expansion} in the proof of
\lemref{thetahat_consistent}, with probability approaching one,
\begin{align*}
\sqrt{N}(\thetahat - \thetatrue) ={}&
    \tresid{2}{\likhat, \thetahat, \theta}^{-1}
        \sqrt{N} \likhatk{1}(\thetatrue) \\
&\dlim \normdist(\cdot | 0, \info^{-1} \scorecov \info^{-1}),
\end{align*}
where the final line follows from Slutsky's theorem.
\end{proof}

\subsection{Ratio of integrals}\seclabel{bclt_ratio}
We need to consider the asymptotic behavior of $\expect{\post}{\phi(\theta, \xvec_s)}$,
which can be written as the ratio of the integrals
\begin{align*}
\expect{\post}{\phi(\theta, \xvec_s)} ={}&
\frac{\int \phi(\theta, \xvec_s)
        \exp\left(N \likhat(\theta)\right)
        d \theta}
  {\int \exp\left(N\likhat(\theta)\right) d \theta}.
\end{align*}
Recall that $\thetahat := \mathrm{argmax}_\theta \likhat(\theta)$ and
define $\thetatil(t) := t (\theta - \thetahat) + \thetahat$.
We will employ the following Taylor series expansion around $\thetahat$.
\begin{align*}
\likhat(\theta) ={}&
    \likhat(\thetahat) +
    \likhatk{1}(\thetahat)(\theta - \thetahat) +
    \frac{1}{2}\likhatk{2}(\thetahat)(\theta - \thetahat)^2 +
\\&
    \frac{1}{6}\likhatk{3}(\thetahat)(\theta - \thetahat)^3 +
    \tresid{4}{\likhat, \thetahat, \theta} (\theta - \thetahat)^4,
\end{align*}
recalling the definition of the integral form of the Taylor series residual
$\tresid{4}{\likhat, \thetahat, \theta}$ given in \defref{taylor_residual}.

Note that $\likhat(\thetahat)$ is constant in $\theta$, and that
$\likhatk{1}(\thetahat) = 0$.  Plugging in and changing the domain of
integration to $\tau$, we see
\begin{align*}
\expect{\post}{\phi(\theta, \xvec_s)}
={}&
\frac{\int \phi(\theta, \xvec_s)
       \exp\left(
            N \likhat(\theta) - N \likhat(\thetahat)
        \right)
       d \theta}
     {\int
     \exp\left(
        N \likhat(\theta) - N \likhat(\thetahat)
     \right) d \theta}
\\={}&
\frac{\int \phi(\theta, \xvec_s)
       \exp\left(
            N \likhat(\theta) - N \likhat(\thetahat)
        \right)
       \sqrt{N} d \tau}
     {\int
     \exp\left(
        N \likhat(\theta) - N \likhat(\thetahat)
     \right) \sqrt{N} d \tau}.
\end{align*}
Since the factor of $\sqrt{N}$ cancels in the preceding display, it will suffice
to consider the following generic integral.
\begin{defn}\deflabel{bclt_integral}
For a target function $\phi$, define the integral functional
\begin{align*}
I(\phi) :={}&
\int \phi(\theta, \xvec_s)
   \exp\left(
       N \likhat(\theta) - N \likhat(\thetahat)
   \right)
   d \tau \\
={}&
\int \phi(\theta, \xvec_s) \times
\\&\quad
   \exp\left(
       \frac{1}{2} \likhatk{2}(\thetahat)\tau^2 +
       N^{-1/2} \frac{1}{6}\likhatk{3}(\thetahat) \tau^3 +
       N^{-1} \tresid{4}{\likhat, \thetahat, \theta} \tau^4
   \right)
   d \tau.
\end{align*}
\end{defn}
Under \defref{bclt_integral}, $\expect{\post}{\phi(\theta, \xvec_s)} = I(\phi) /
I(1)$.  We will now consider the behavior of $I(\phi)$ for a
generic third-order BCLT-okay $\phi$,
returning to consider the behavior of the ratio in \secref{bclt_region_comb}.

\subsection{Regions}
For some $\delta_1$ and $\delta_2$ to be determined later, we will divide
consideration of $I(\phi)$ into the following three disjoint regions.

\begin{defn}\deflabel{region_division}
For any $\delta_1 > 0$ and $\delta_2 > 0$ and $N$ sufficiently large, define
the disjoint regions
\begin{align*}
R_1 :={}& \left\{\theta:
    \norm{\theta - \thetahat}_2 < \delta_1 \frac{\log N}{\sqrt{N}} \right\}\\
R_2 :={}& \left\{\theta:
    \delta_1 \frac{\log N}{\sqrt{N}} \le
        \norm{\theta - \thetahat}_2 < \delta_2 \right\} \\
R_3 :={}& \left\{\theta:
    \delta_2 \le \norm{\theta - \thetahat}_2  \right\}.
\end{align*}
Recalling our slight abuse of set notation, the same regions in the domain of $\tau$
are,
\begin{align*}
\tau \in R_1 \Rightarrow{}& \tau \in \left\{\tau:
    \norm{\tau}_2 < \delta_1 \log N \right\}\\
\tau \in R_2 \Rightarrow{}& \left\{\tau:
    \delta_1 \log N \le
        \norm{\tau}_2 < \delta_2 \sqrt{N} \right\} \\
\tau \in R_3 \Rightarrow{}& \left\{\tau:
    \delta_2 \sqrt{N} \le \norm{\tau}_2  \right\}.
\end{align*}
Similarly, we will use the regions as domains of integration for $\tau$ as well
as $\theta$. Correspondingly, we can define
\begin{align*}
I_1(\phi) :={}& \int_{R_1} \phi(\theta, \xvec_s)
   \exp\left(N \likhat(\theta) - N \likhat(\thetahat)\right) d \tau \\
I_2(\phi) :={}& \int_{R_2} \phi(\theta, \xvec_s)
  \exp\left(N \likhat(\theta) - N \likhat(\thetahat)\right) d \tau \\
I_3(\phi) :={}& \int_{R_3} \phi(\theta, \xvec_s)
 \exp\left(N \likhat(\theta) - N \likhat(\thetahat)\right) d \tau.
\end{align*}
\end{defn}
Since the boundaries of the regions coincide, $I(\phi) = I_1(\phi) + I_2(\phi) +
I_3(\phi)$ for $N$ when the regions $R_1$, $R_2$, and $R_3$ are disjoint, which
will occur for $N$ sufficiently large (see \defref{regular_event} below for more
discussion of this point). Note that the union $R_1 \bigcup R_2$ is a $\delta_2$
ball around $\thetahat$, and that $R_1$ shrinks to $\thetahat$ as $N \rightarrow
\infty$.

Throughout the proof, we will be choosing $\delta_1$ and $\delta_2$ to satisfy
certain conditions.  Our assumptions will always take the form of assuming that
$\delta_1$ is \textit{large enough}, and that $\delta_2$ is \textit{small
enough}.  For any particular $\delta_1$ and $\delta_2$, we then will require
that $N$ is large enough that the regions are disjoint.


\subsection{Nearly certain events}
Our entire analysis will be conditioned on the occurrence of several events,
each of which occurs with probability approaching one as $N$ goes to infinity.
Outside of these events we provide no guarantees.  As a consequence, the
residuals of our approximation can be used to show convergence in probability
or distribution (via Slutsky's theorem), but not convergence in expectation
nor almost sure convergence.

\begin{defn}\deflabel{regular_event}
For any any $\epsilon_U > 0$, let $\regevent$
denote the event that all of the following occur
for a particular $N$, $\delta_1$, $\delta_2$.
(``R'' is for ``regular''.)
\begin{enumerate}
\item \eventlabel{disjoint_regions} The regions $R_1$, $R_2$, and $R_3$ are
disjoint, and $I(\phi) = I_1(\phi) + I_2(\phi) + I_3(\phi)$.
\item \eventlabel{thetahat_close}
$\norm{\thetahat - \thetatrue}_\infty < \epsilon_U$.
\item \eventlabel{prior_close}
    $|\prior(\thetahat) - \prior(\thetatrue)| < \epsilon_U$.
\item \eventlabel{ulln} For any scalar-valued function $\rho(\theta, \x_n)$
with $ \expect{\fdist(\xn)}{\rho(\theta, \xn)} = 0$
satisfying a ULLN with $\delta = \deltaulln$,
\begin{align*}
\sup_{\theta \in R_1 \bigcup R_2} \meann
\abs{\rho(\theta, \x_n)}
    \le \epsilon_U.
\end{align*}
%
%
\item \eventlabel{infoev} $\infoevhat > \infoev / 2$.
\item \eventlabel{strict_opt} There exists $\epsilon_L > 0$ such that
\begin{align*}
\sup_{\theta \in R_3}
\left( \sumn \ell(\x_n \vert \thetahat) -
       \sumn \ell(\x_n \vert \theta)
\right) > N \epsilon_L.
\end{align*}
\end{enumerate}
\end{defn}


\begin{lem}\lemlabel{regular}
Let \assuref{bayes_clt} hold. For any $\delta_1 > 0$, any $\pi_0 \in (0,1)$, and
any $\delta_2 > 0$ and $\epsilon_U > 0$ satisfying  $\deltaulln > \epsilon_U +
\delta_2$ and $\delta_2 / 2 \ge \epsilon_U$, there exists an $N^*$ such $N >
N^*$ implies that event $\regevent$ occurs with $\fdist$-probability at least $1 -
\pi_0$.
\begin{proof}
\Eventref{disjoint_regions} follows from the fact that, for any positive
$\delta_1$ and $\delta_2$, eventually $\delta_2 \sqrt{N} > \delta_1 \log N$.

\Eventref{thetahat_close} follows from \lemref{thetahat_consistent}.

\Eventref{prior_close} follows from the fact that the prior density is
continuous by \assuref{bayes_clt} \itemref{prior_smooth} and,
by \lemref{thetahat_consistent}, we can take
$\thetahat$ to be arbitrarily close to $\thetatrue$.

For \eventref{ulln}, observe that $\norm{\thetahat - \thetatrue}_2 < \epsilon_U$
implies that
\begin{align*}
\theta \in& R_1 \bigcup R_2 \Rightarrow\\
\norm{\theta - \thetahat}_2 \le& \delta_2 \Rightarrow\\
\norm{\theta - \thetatrue}_2 \le&
    \delta_2 + \norm{\thetahat - \thetatrue}_2
\le \delta_2 + \epsilon_U \le \deltaulln.
\end{align*}
Therefore,
\begin{align*}
\sup_{\theta \in R_1 \bigcup R_2} \meann \abs{\rho(\theta, \x_n)} \le&
\sup_{\theta \in \thetaball{\deltaulln}}
    \meann \abs{\rho(\theta, \x_n)} \plim 0.
\end{align*}
\Eventref{infoev} follows from \lemref{thetahat_consistent}.

For \eventref{strict_opt}, write
\begin{align*}
\MoveEqLeft
\frac{1}{N} \sup_{\theta \in R_3}
\left( \sumn \ell(\x_n \vert \thetahat) -
       \sumn \ell(\x_n \vert \theta)
\right) =
\likhat(\thetahat) - \likhat(\thetatrue) +
    \sup_{\theta \in R_3}
        \left(\likhat(\thetatrue) - \likhat(\theta)\right).
\end{align*}
Since $\norm{\thetahat - \thetatrue}_2 \le \epsilon_U \le \delta_2 / 2$
by assumption of the present lemma,
\begin{align*}
\theta \in& R_3 \Rightarrow\\
\norm{\theta - \thetahat + \thetahat - \thetatrue}_2 \ge & \delta_2 \Rightarrow\\
\norm{\theta - \thetahat}_2 + \norm{\thetahat - \thetatrue}_2
    \ge& \delta_2 \Rightarrow\\
\norm{\theta - \thetahat}_2 \ge& \delta_2 - \frac{\delta_2}{2} = \frac{\delta_2}{2},
\end{align*}
so the set $R_3$, is contained in $\thetadom \setminus \thetaball{\delta_2 /
2}$. Consequently, by \assuref{bayes_clt} \itemref{strict_opt}, there exists an
$\epsilon_L$ such that
\begin{align*}
\sup_{\theta: \norm{\theta - \thetahat}_2 \ge \delta_2}
    \left(\likhat(\thetatrue) - \likhat(\theta)\right) \le&
\sup_{\theta \in \thetadom \setminus \thetaball{\delta_2 / 2}}
    \left(\likhat(\thetatrue) - \likhat(\theta)\right) \le
-2 \epsilon_L.
\end{align*}
By the continuity of $\lik$, and a ULLN applied to $\likhat(\theta)$, we can
take $\norm{\thetahat - \thetatrue}_2$ sufficiently small that that
$|\likhat(\thetahat) - \likhat(\thetatrue)| < \epsilon_L$.  It follows that,
with probability approaching one,
\begin{align*}
\frac{1}{N} \sup_{\theta: \norm{\theta - \thetahat}_2 \ge \delta_2}
\left( \sumn \ell(\x_n \vert \thetahat) -
       \sumn \ell(\x_n \vert \theta)
\right) \le \epsilon_L - 2 \epsilon_L = -\epsilon_L.
\end{align*}
The result follows by multiplying both sides of the inequality by $N$.
\end{proof}
\end{lem}

Note that \lemref{regular} requires $\delta_2$ to be sufficiently small (so that
$\deltaulln > \delta_2$), which in turn requires $\epsilon_U$ to be small (so
that $\deltaulln > 3 \epsilon_U$).  For any particular $\deltaulln > \delta_2$,
we can always apply \lemref{regular} with arbitrarily small $\epsilon_U$, at the
cost of requiring a large enough $N^*$ so that the probability convergence
results apply.
Throughout, we will choose $\epsilon_U$ \textit{small enough} that we obtain our
desired results.

\subsection{Region $R_3$}

The integral over $R_3$ goes to zero by strict optimality of $\thetatrue$
and the proper behavior of the prior.

\begin{lem}\lemlabel{region_r3}
Under \assuref{core_bclt_assu}, when $\regevent$ occurs with
$\epsilon_U < \prior(\thetatrue) / 2$,
\begin{align*}
\norm{I_3(\phi)}_2 =
O(\exp(-N \epsilon_L)) \expect{\prior(\theta)}{\norm{ \phi(\theta, \xvec_s) }_2 }.
\end{align*}
\begin{proof}
We have
\begin{align*}
\norm{I_3(\phi)}_2
={}& \norm{ \int_{R_3} \phi(\theta, \xvec_s)
       \exp\left(
            N \likhat(\theta) - N \likhat(\thetahat)
        \right)
       d \theta }_2 \\
\le&
\sup_{\theta \in R_3}
\exp\Bigg(
    N \Bigg(
         \frac{1}{N} \sumn \ell(\x_n \vert \theta) -
         \frac{1}{N} \sumn \ell(\x_n \vert \thetahat)
     \Bigg)
 \Bigg) \times
 \int_{R_3} \norm{\phi(\theta, \xvec_s)}_2
 \frac{\prior(\theta)}{\prior(\thetahat)} d \theta \\
={}&
    \exp(-N \epsilon_L)
    \frac{1}{\prior(\thetahat)}
    \expect{\prior(\theta)}{
        \norm{\phi(\theta, \xvec_s)}_2 \mathbb{I}(\theta \in R_3)
    }
    \quad\textrm{(\defref{regular_event}, \eventref{strict_opt})}\\
\le&
    \exp(-N \epsilon_L)
    \frac{1}{\prior(\thetatrue) - \epsilon_U}
    \expect{\prior(\theta)}{\norm{\phi(\theta, \xvec_s)}_2 \mathbb{I}(\theta \in R_3)}
\quad\textrm{(\defref{regular_event}, \eventref{prior_close})}\\
\le&
\exp(-N \epsilon_L) 2 \prior(\thetatrue)^{-1}
\expect{\prior(\theta)}{\norm{\phi(\theta, \xvec_s)}_2 }.
\end{align*}
In the last step, we take $\epsilon_U < \prior(\thetatrue) / 2$ in $\regevent$
(specifically, \defref{regular_event} \eventref{prior_close}).
\end{proof}
\end{lem}
%

\subsection{Region $R_2$}
Region $R_2$ is dealt with by choosing $\delta_2$ sufficiently small that
$I_2(\phi)$ is dominated by a particular Gaussian integral.  The we will make
that Gaussian integral small by choosing $\delta_1$ sufficiently large.

\begin{lem}\lemlabel{region_r2}
Let \assuref{core_bclt_assu} hold. Take $\delta_1 \ge \sqrt{8 / \infoev}$.
There exist $\delta_2$ and $\epsilon_U$, defined as functions only of the
population quantities $\infoev$ and $\likk{k}(\theta)$, such that, for any
$\xvec$ such that $\regevent$ occurs,
\begin{align*}
\norm{I_2(\phi)}_2 \le&
    O\left(N^{-3}\right)
    \expect{\normresid{\infoev / 2}{R_1 \bigcup R_2}(\tau)}
           {\norm{\phi(\theta, \xvec_s)}_2 }.
\end{align*}
\end{lem}
\begin{proof} 
The main part of the proof consists in showing that there exists a $\delta_2$
sufficiently small so that
\begin{align}\eqlabel{r2_target_bound}
\norm{I_2(\phi)}_2 \le& \int_{R_2}
    \exp\left(
        -\frac{\infoev}{2} \norm{\tau}_2^2
        \right) \norm{\phi(\theta, \xvec_s)}_2 d\tau,
\end{align}
The approach will be to bound the difference between $N \likhat(\theta) - N
\likhat(\thetahat)$ and $\frac{1}{2} \likk{2}(\thetatrue)$ so that the entire
expression inside the integral's exponent can be written as  $A \tau^2$ for a
positive definite matrix $A$.

For the duration of this proof, define the following residual matrix:
\begin{align*}
\Lambda(\theta, \xvec, N) :={}&
    \frac{1}{2} \left(\likhatk{2}(\thetahat) - \likk{2}(\thetatrue)\right) + \\
&   N^{-1/2} \frac{1}{6}\left(
            \likhatk{3}(\thetahat) -
            \likk{3}(\thetatrue)
        \right) \tau + \\
&    N^{-1/2} \frac{1}{6}\likk{3}(\thetatrue) \tau +
    N^{-1} \tresid{4}{\likhat, \thetahat, \theta} \tau^2.
\end{align*}
Under this definition,
\begin{align*}
N \likhat(\theta) - N \likhat(\thetahat) ={}&
\left( \frac{1}{2} \likk{2}(\thetatrue) + \Lambda(\theta, \xvec, N) \right) \tau^2.
\end{align*}
We will now bound the entries of the $D \times D$ matrix $\Lambda(\theta, \xvec,
N)$.

When $\regevent$ occurs, we have
\begin{align*}
\norm{\likhatk{2}(\thetahat) - \likk{2}(\thetatrue)}_2 \le&
\norm{\likhatk{2}(\thetahat) - \likk{2}(\thetahat)}_2 +
\norm{\likk{2}(\thetahat) - \likk{2}(\thetatrue)}_2 \\
\le&
\sup_{\theta \in R_2}
    \norm{\likhatk{2}(\theta) - \likk{2}(\theta)}_2 +
\sup_{\theta \in R_2} \norm{\likk{3}(\theta)}_2
    \norm{\thetahat - \thetatrue}_2 \\
\le& \left(1 + \sup_{\theta \in R_2} \norm{\likk{3}(\theta)}_2\right)
    \epsilon_U.
\quad\textrm{\defref{regular_event}, \eventref{ulln, thetahat_close}}
\end{align*}
Similarly,
\begin{align*}
\norm{\left(\likhatk{3}(\thetahat) - \likk{3}(\thetatrue)\right) \tau}_2
\le&
\norm{\likhatk{3}(\thetahat) - \likk{3}(\thetatrue)}_2
\norm{\tau}_2
\\
\le& \left(1 + \sup_{\theta \in R_2} \norm{\likk{4}(\theta)}_2\right)
    \norm{\tau}_2 \epsilon_U.
\end{align*}
Next, applying Cauchy-Schwartz gives
\begin{align*}
\norm{\likk{3}(\thetatrue) \tau}_2 \le
    \norm{\likk{3}(\thetatrue)}_2 \norm{\tau}_2.
\end{align*}
Finally, let us consider the Taylor series residual term.  Let $\thetatil(t) = t
(\theta - \thetahat) + \thetahat$.  Because we can write
\begin{align*}
\tresid{4}{\likhat, \thetahat, \theta} ={}&
\frac{1}{6} \int_0^1 (t - 1)^3 \likhatk{4}(\thetatil(t)) dt \\
={}&
\frac{1}{6} \int_0^1 (t - 1)^3 \left(
    \likhatk{4}(\thetatil(t)) - \likk{4}(\thetatil(t))
\right) dt +
\frac{1}{6} \int_0^1 (t - 1)^3 \likk{4}(\thetatil(t)) dt,
\end{align*}
we have, when $\regevent$ occurs,
\begin{align*}
\norm{\tresid{4}{\likhat, \thetahat, \theta} \tau^2 }_2 \le&
\norm{\tresid{4}{\likhat, \thetahat, \theta}}_2 \norm{\tau}_2^2 \\
\le&
\frac{1}{6} \left(
    \sup_{\theta \in R_2}
        \norm{\likhatk{4}(\theta) - \likk{4}(\theta)}_2 +
    \sup_{\theta \in R_2} \norm{\likk{4}(\theta)}_2 \right) \norm{\tau}_2^2 \\
\le&
\frac{1}{6} \left(
    \epsilon_U + \sup_{\theta \in R_2} \norm{\likk{4}(\theta)}_2 \right)
    \norm{\tau}_2^2.
\quad\textrm{\defref{regular_event}, \eventref{ulln}}
\end{align*}

Combining the above results, together with the fact that, in $R_2$, $\sqrt{N}
\norm{\tau}_2 \le \delta_2$, gives
\begin{align}
\norm{\Lambda(\theta, \xvec, N)}_2 \le&
    \frac{1}{2} \norm{\likhatk{2}(\thetahat) - \likk{2}(\thetatrue)}_2 +
    \nonumber\\&
        N^{-1/2} \frac{1}{6}\norm{
            \likhatk{3}(\thetahat) -
            \likk{3}(\thetatrue)
        } \norm{\tau}_2 +
    \nonumber\\&
        N^{-1/2} \frac{1}{6} \norm{\likk{3}(\thetatrue)}_2 \norm{\tau}_2 +
    \nonumber\\&
    N^{-1} \norm{\tresid{4}{\likhat, \thetahat, \theta}}_2 \norm{\tau}_2^2
\nonumber\\\le&
    \frac{1}{2} \left(1 + \sup_{\theta \in R_2} \norm{\likk{3}(\theta)}_2\right)
        \epsilon_U +
   \frac{1}{6} \left(1 + \sup_{\theta \in R_2} \norm{\likk{4}(\theta)}_2\right)
    \delta_2 \epsilon_U + \nonumber\\
&   \frac{1}{6} \norm{\likk{3}(\thetatrue)}_2 \delta_2 +
   \frac{1}{6} \left(
        \epsilon_U + \sup_{\theta \in R_2} \norm{\likk{4}(\theta)}_2 \right)
    \delta_2^2. \eqlabel{r2_lambda_bound}
\end{align}
The preceding bound contains only the $\epsilon_U$, $\delta_2$, and the
population quantities $\sup_{\theta \in R_2} \norm{\likk{3}(\theta)}_2$,
$\sup_{\theta \in R_2} \norm{\likk{4}(\theta)}_2$, and $\infoev$. The quantities
$\sup_{\theta \in R_2} \norm{\likk{3}(\theta)}_2$ and $\sup_{\theta \in R_2}
\norm{\likk{4}(\theta)}_2$ depend on $\delta_2$ themselves, but are finite and
decreasing in $\delta_2$.

Consequently, when $\regevent$ occurs, there exist an $\epsilon_U$ and
$\delta_2$ sufficiently small so that
\begin{align}\eqlabel{re2_lambda_inf_bound}
\norm{\Lambda(\theta, \xvec, N)}_\infty \le
    \sqrt{D^2} \norm{\Lambda(\theta, \xvec, N)}_2 \le \frac{\infoev}{2},
\end{align}
independently of $\xvec$ (except that $\regevent$ occurs). The bound
\eqref{re2_lambda_inf_bound} cannot necessarily be acheived simply by
setting $\epsilon_U$ to be small (i.e., by relying on the consistency of
$\thetahat$ and the ULLN) because of the terms $\sup_{\theta \in R_2}
\norm{\likk{4}(\theta)}_2 \delta_2^2$ and $\norm{\likk{3}(\thetatrue)}_2
\delta_2$ in \eqref{r2_lambda_bound}---it is necessary to set $\delta_2$
small as well.

Recall from \defref{loglik_def} that $\info =
-\frac{1}{2}\likk{2}(\thetatrue)$, and from \defref{loglik_details_def} that
$\infoev > 0$ is the minumum eigenvalue of $\info$.  These facts, combined
with \eqref{re2_lambda_inf_bound} give that, for any $\tau$,
\begin{align*}
\left(\frac{1}{2}\likk{2}(\thetatrue) + \Lambda(\theta, \xvec, N)\right) \tau^2
    \le& -\infoev \norm{\tau}_2^2 + \frac{\infoev}{2} \norm{\tau}_2^2 \\
    ={}& -\frac{\infoev}{2} \norm{\tau}_2^2.
\end{align*}
Plugging into the definition of $I_2(\phi)$ given in \defref{region_division}
gives the desired \eqref{r2_target_bound}.

To complete the proof, we use the fact that the normal density is bounded
in $R_2$.

\begin{align*}
\norm{I_2(\phi)}_2
\le&
\int_{R_2}
    \exp\left(
        -\frac{\infoev}{2} \norm{\tau}_2^2
        \right) \norm{\phi(\theta, \xvec_s)}_2  d\tau
\\={}&
\int_{R_2}
    \exp\left(
        -\frac{\infoev}{4} \norm{\tau}_2^2
        -\frac{\infoev}{4} \norm{\tau}_2^2
        \right) \norm{\phi(\theta, \xvec_s)}_2 d\tau
\\\le&
\sup_{\tau \in R_2}
\exp\left(
    -\frac{\infoev}{4} \norm{\tau}_2^2
\right)
\int_{R_2}
    \exp\left(
        -\frac{\infoev}{4} \norm{\tau}_2^2
        \right) \norm{\phi(\theta, \xvec_s)}_2 d\tau
\\\le &
\exp\left(
    -\frac{\infoev}{4} \delta_1^2 \log N
\right)
\int_{R_1 \bigcup R_2}
    \exp\left(
        -\frac{\infoev}{4} \norm{\tau}_2^2
        \right) \norm{\phi(\theta, \xvec_s)}_2 d\tau.
\end{align*}
In the final line, we have used the fact that $(\log N)^2 \ge \log N$ for all $N
\ge 3$. Taking $\delta_1^2 \ge 12 / \infoev$ and multiplying by the $O(1)$
normalizer for the $\normresid{\infoev/2}{R_1 \bigcup R_2}(\tau)$ completes the proof.
\end{proof} 

By inspecting the proof of \lemref{region_r2}, one can see that that any
polynomial rate of convergence can be achieved by taking $\delta_1$ sufficiently
large. Indeed, if the lower upper boundary of $R_1$ were of the form $\tau \le
N^k$ for some $k < 1/2$, then $\norm{I_2(\phi)}_2$ would decay exponentially and
$R_1$ would still shrink to a point in the $\theta$ domain.  However, such a
rapidly growing boundary between $R_1$ and $R_2$ comes at a cost in the analysis
of region $R1$.  In particular, it is the use of a logarithmically growing
boundary is what allows us to give polylog bounds in the subsequent
\lemref{region_r1_final_bound} and in the final result of
\thmref{bayes_clt_main}.

\subsection{Region $R_1$}\seclabel{region_r1}
It is region $R_1$ that will dominate the value of the integral $I(\phi)$.  We
will first perform a Taylor series expansion of all the terms around their
sample versions, and then take probability limits of the resulting expressions.

First, we will show that integrals over $\normhat$ in $R_1$ are sufficiently
close to the corresponding Gaussian integrals over $\mathbb{R}^D$ for
sufficiently large $\delta_1$.


\begin{lem}\lemlabel{region_r1_gaussian_moments}
Let \assuref{core_bclt_assu} hold.  Let $\tau^k$ represent the
$D$-dimensional array of outer products of $\tau$.  When $\regevent$ occurs and
$\delta_1 \ge \sqrt{24 / \infoev}$ then, for any $k \ge 0$,
\begin{align*}
\int_{R_1} \normhatarg{\tau} \tau^k d\tau =
    \expect{\normhat}{\tau^k} + O(N^{-3}).
\end{align*}
\begin{proof}
By definition,
\begin{align*}
\int_{R_1} \normhatarg{\tau} \tau^k d\tau  - \expect{\normhat}{\tau^k}
={}&
\int_{R_2 \bigcup R_3} \normhatarg{\tau} \tau^k d\tau
\end{align*}
We will now use a technique similar to that of the proof of
\lemref{region_r2} to control the norm of the final integral.
Recall that, on $\regevent$, the minimum eigenvalue of $\infohat$, $\infoevhat$,
is at least $\frac{1}{2} \infoev$ (\defref{regular_event} \eventref{infoev}),
so that $\normalizer^{-1} = O(1)$.
\begin{align*}
\MoveEqLeft
\norm{
    \int_{R_2 \bigcup R_3}
            \normhatarg{\tau} \tau^k d\tau
}_2
\\\le&
    \int_{R_2 \bigcup R_3}
            \normhatarg{\tau} \norm{\tau}^k_2 d\tau
\\={}&
\normalizer^{-1} \int_{R_2 \bigcup R_3}
        \exp\left(-\frac{1}{2} \infohat \tau^2 \right)
            \norm{\tau}^k_2 d\tau
\\\le&
\normalizer^{-1} \int_{R_2 \bigcup R_3}
        \exp\left(-\frac{1}{4} \infoev \norm{\tau}_2^2 \right)
            \norm{\tau}^k_2 d\tau
\\\le&
\sup_{\tau \in R_2 \bigcup R_3}
    \exp\left(-\frac{1}{8} \infoev \norm{\tau}_2^2 \right)
\normalizer^{-1} \int_{R_2 \bigcup R_3}
        \exp\left(-\frac{1}{8} \infoev \norm{\tau}_2^2 \right)
            \norm{\tau}^k_2 d\tau
\\={}&
\exp\left(-\frac{1}{8} \infoev \delta_1^2 (\log N)^2 \right)
\normalizer^{-1} \int_{R_2 \bigcup R_3}
        \exp\left(-\frac{1}{8} \infoev \norm{\tau}_2^2 \right)
            \norm{\tau}^k_2 d\tau
\\={}& O(N^{-3}),
\end{align*}
where the final line takes $\delta_1^2 \ge \frac{24}{\infoev}$, that $\log N <
(\log N)^2$ for $N \ge 3$, that $\normalizer^{-1} = O(1)$, and that normal
moments are finite.
\end{proof}
\end{lem}


Next, we series expand the likelihood term.

\begin{lem}\lemlabel{region_r1_loglik_expansion}
Let \assuref{core_bclt_assu} hold.  For any $\xvec$ such that $\regevent$
occurs,
\begin{align*}
\MoveEqLeft
\exp\left(N \likhat(\theta) - N \likhat(\thetahat)\right) = \\
& \exp\left(\frac{1}{2} \infohat \tau^2 \right) \times
\left( 1 + N^{-1/2} \frac{1}{6} \likhatk{3}(\theta) \tau^3 +
    \ordlog{N^{-1}} \right).
\end{align*}
\begin{proof}

Write the log likelihood term as
\begin{align}
\MoveEqLeft
\exp\left(N \likhat(\theta) - N \likhat(\thetahat)\right) =  \nonumber \\
& \exp\left(\frac{1}{2} \likhatk{2}(\thetahat) \tau^2 \right) \times \nonumber \\
& \exp\left(
        N^{-1/2} \frac{1}{6} \likhatk{3}(\thetahat) \tau^3 +
        N^{-1} \frac{1}{6} \tresid{4}{\likhat(\theta), \thetahat, \theta} \tau^4
    \right). \eqlabel{r1_loglike_term_factorization}
\end{align}
We will expand the second term of \eqref{r1_loglike_term_factorization} using
the expansion of the  exponential function around 0, which holds for all
$z$:
\begin{align*}
\exp(z) ={}& 1 + z + z^2 + \frac{1}{2} \int_0^1 (1 - t)^2 \exp(tz) dt z^3 \\
={}& 1 + z + O(|z|^2).
\end{align*}
To match \eqref{r1_loglike_term_factorization}, we need to evaluate at
\begin{align*}
z = N^{-1/2} \frac{1}{6} \likhatk{3}(\thetahat) \tau^3 +
    N^{-1} \tresid{4}{\likhat(\theta), \thetahat, \theta} \tau^4.
\end{align*}
When $\regevent$ occurs (\defref{regular_event} \eventref{ulln}), we
have
\begin{align*}
\norm{N^{-1} \tresid{4}{\likhat(\theta), \thetahat, \theta} \tau^4}_2 \le&
N^{-1} \sup_{\theta \in R_1}
    \norm{\tresid{4}{\likhat(\theta), \thetahat, \theta}}_2
    \sup_{\tau \in R_1}\norm{\tau^3}_2 \\
={}& \ordlog{N^{-1}}\\
\norm{N^{-1/2} \frac{1}{6} \likhatk{3}(\thetahat) \tau^3}_2 \le&
N^{-1/2} \frac{1}{6} \norm{\likhatk{3}(\thetahat)}_2
    \sup_{\tau \in R_1}\norm{\tau^3}_2\\
={}& \ordlog{N^{-1/2}} \Rightarrow\\
|z| ={}& \ordlog{N^{-1/2}},
\end{align*}
so that
\begin{align}
\exp\left(
        N^{-1/2} \frac{1}{6} \likhatk{3}(\thetahat) \tau^3 +
        N^{-1} \frac{1}{6} \tresid{4}{\likhat(\theta), \thetahat, \theta} \tau^4
    \right) = \nonumber\\
1 + N^{-1/2} \frac{1}{6} \likhatk{3}(\theta) \tau^3 +
    \ordlog{N^{-1}}
. \eqlabel{r1_exp_expansion}
\end{align}
Combining \eqref{r1_loglike_term_factorization, r1_exp_expansion} and
recognizing the definition of $\infohat$ gives the desired result.
\end{proof}
\end{lem}


\begin{lem}\lemlabel{r1_normalizer_expansion}
Let \assuref{core_bclt_assu} hold.  When $\regevent$ occurs and $\delta_1$ is as
given in \lemref{region_r1_gaussian_moments},
\begin{align*}
\int_{R_1} \exp\left(
    N \likhat(\theta) - N \likhat(\thetahat)
\right) d\tau = \normalizer + \ordlog{N^{-1}}.
\end{align*}
\begin{proof}
Recall that, by definition,
\begin{align*}
\normalizer = \int \exp\left(-\frac{1}{2} \infohat \tau^2 \right) d\tau.
\end{align*}
Expanding using \lemref{region_r1_loglik_expansion} gives
\begin{align*}
\MoveEqLeft
\left| \int_{R_1} \exp\left(
    N \likhat(\theta) - N \likhat(\thetahat)
\right) d\tau - \normalizer \right|
\\={}&
\left|
\int_{R_1} \exp\left(\frac{1}{2} \infohat \tau^2 \right)
\left( 1 + N^{-1/2} \frac{1}{6} \likhatk{3}(\theta) \tau^3 +
    \ordlog{N^{-1}} \right) d\tau -
\int \exp\left(-\frac{1}{2} \infohat \tau^2 \right) d\tau
\right|
\\={}&
\left|
\int_{R_1} \exp\left(\frac{1}{2} \infohat \tau^2 \right)
\left( N^{-1/2} \frac{1}{6} \likhatk{3}(\theta) \tau^3 +
    \ordlog{N^{-1}} \right) d\tau -
\int_{R_2 \bigcup R_3} \exp\left(-\frac{1}{2} \infohat \tau^2 \right) d\tau
\right|
\\\le&
\frac{1}{6} N^{-1/2} \sup_{\theta \in R_1} \norm{
    \likhatk{3}(\theta) }_2
\expect{\normhat}{\ind{\tau \in R_1}\tau^3} +
\sup_{\tau \in R_2 \bigcup R_3}
    \exp\left(-\frac{1}{4} \infoev \norm{\tau}_2^2 \right) +
\ordlog{N^{-1}}
\\={}&
O(N^{-2}) +
\exp\left(-\frac{1}{4} \infoev \delta_1^2 (\log N)^2 \right) +
\ordlog{N^{-1}}
\\={}& \ordlog{N^{-1}},
\end{align*}
where, in the penultimate line, we use \lemref{region_r1_gaussian_moments},
including the given lower bound on $\delta_1$.
\end{proof}
\end{lem}



In \lemref{region_r1_final_bound} to follow, we integrate the result of
\lemref{region_r1_loglik_expansion} using \lemref{region_r1_gaussian_moments} to
produce our final results for region $R_1$. Recall from
\defref{loglik_details_def} that $\normhat = \normalizer^{-1}
\exp\left(\frac{1}{2} \infohat \tau^2 \right)$ is the density of the Gaussian
distribution on $\tau$ with mean zero and covariance $\infohat^{-1} =
(-\likhatk{2}(\thetahat))^{-1}$.

Note that the statement of \lemref{region_r1_final_bound} combines vector-valued
quantities like $\phigrad{2}(\thetahat, \xvec_s) \expect{\normhat}{\tau^2}$ with
scalar-valued quantities like $\norm{\phigrad{2}(\thetahat, \xvec_s)}_2$.  This is
for notational convenience since the scalar-valued terms will simply become part
of the residual, and their component values do not matter.  Formally, the
equation can be taken to hold componentwise.
\begin{lem}\lemlabel{region_r1_final_bound}
Let \assuref{core_bclt_assu} hold.  When $\delta_1$ is as given in
\lemref{region_r2}, then, for any $\xvec$ such that $\regevent$
occurs,
\begin{align*}
\MoveEqLeft
I_1(\phi) =
    \phi(\thetahat, \xvec_s)
    \int_{R_1} \exp\left(
        N \likhat(\theta) - N \likhat(\thetahat)
    \right) d\tau +
\nonumber\\ &\quad
    N^{-1} \normalizer \left(
        \frac{1}{2} \phigrad{2}(\thetahat, \xvec_s) \expect{\normhat}{\tau^2} +
        \frac{1}{6} \likhatk{3}(\thetahat) \phigrad{1}(\thetahat, \xvec_s)
            \expect{\normhat}{\tau^4} \right) +
\nonumber\\ &\quad
    \ordlog{N^{-2}} \bigg(
        1 +
        \norm{\phigrad{1}(\thetahat, \xvec_s)}_2 +
        \norm{\phigrad{2}(\thetahat, \xvec_s)}_2 +
        \expect{\normresid{\infoev / 2}{R_1}(\tau)}{
            \norm{\tresid{3}{\phi(\cdot, \xvec_s), \theta, \thetahat}}_2}
    \bigg).
\end{align*}
\begin{proof}
The result will follow by Taylor expanding $\phi$, combining with
\lemref{region_r1_loglik_expansion}, and then integrating using the fact that
Gaussian integrals in region $R_1$ are approximately equal to Gaussian integrals
over the whole domain, as proved in \lemref{region_r1_gaussian_moments}.


We first Taylor expand $\phi(\theta, \xvec_s)$ around $\thetahat$ and substitute
$\tau = \sqrt{N}(\theta - \thetahat)$.
\begin{align}
\phi(\theta, \xvec_s) ={}&
    \phi(\thetahat, \xvec_s) +
    N^{-1/2} \phigrad{1}(\thetahat, \xvec_s) \tau + \nonumber \\
&
    N^{-1} \frac{1}{2} \phigrad{2}(\thetahat, \xvec_s) \tau^2 +
    N^{-3/2} \tresid{3}{\phi(\cdot, \xvec_s), \theta, \thetahat} \tau^3.
    \eqlabel{r1_phi_expansion}
\end{align}


We will now combine \eqref{r1_phi_expansion} with
\lemref{region_r1_loglik_expansion}.
\begin{align}
\MoveEqLeft
\exp\left(N \likhat(\theta) - N \likhat(\thetahat)\right)
\left(
    \phi(\theta, \xvec_s) - \phi(\thetahat, \xvec_s)
\right) = \nonumber \\
\MoveEqLeft
\exp\left(
    \frac{1}{2} \infohat \tau^2
\right) \times \nonumber\\
&
\left(
    N^{-1/2} \phigrad{1}(\thetahat, \xvec_s) \tau +
    N^{-1} \frac{1}{2} \phigrad{2}(\thetahat, \xvec_s) \tau^2 +
    N^{-3/2} \tresid{3}{\phi(\cdot, \xvec_s), \theta, \thetahat} \tau^3
\right)  \times \nonumber\\
&
\left(
    1 +
    N^{-1/2} \frac{1}{6} \likhatk{3}(\thetahat) \tau^3 +
    \ordlog{N^{-1}}
\right) = \nonumber\\
\MoveEqLeft
\exp\left(
    \frac{1}{2} \infohat \tau^2
\right) \times \nonumber\\
&
\Bigg(
    N^{-1/2} \phigrad{1}(\thetahat, \xvec_s) \tau +
\nonumber\\ &\quad
    N^{-1} \left(
        \frac{1}{2} \phigrad{2}(\thetahat, \xvec_s) \tau^2 +
        \frac{1}{6} \likhatk{3}(\thetahat) \phigrad{1}(\thetahat, \xvec_s) \tau^4 \right) +
\nonumber\\ &\quad
    N^{-3/2} \left(
        \tresid{3}{\phi(\cdot, \xvec_s), \theta, \thetahat} \tau^3 +
        \frac{1}{12} \likhatk{3}(\thetahat)  \phigrad{2}(\thetahat, \xvec_s) \tau^5
    \right) +
\nonumber\\ &\quad
    \ordlog{N^{-3/2}} \bigg(
        \phigrad{1}(\thetahat, \xvec_s) \tau +
        N^{-1/2}  \frac{1}{2} \phigrad{2}(\thetahat, \xvec_s) \tau^2 +
        N^{-1} \tresid{3}{\phi(\cdot, \xvec_s), \theta, \thetahat} \tau^3
    \bigg)
\Bigg) = \nonumber\\
\MoveEqLeft
\exp\left(
    \frac{1}{2} \infohat \tau^2
\right) \times \nonumber\\
&
\Bigg(
    N^{-1/2} \phigrad{1}(\thetahat, \xvec_s) \tau +
    N^{-1} \left(
        \frac{1}{2} \phigrad{2}(\thetahat, \xvec_s) \tau^2 +
        \frac{1}{6} \likhatk{3}(\thetahat) \phigrad{1}(\thetahat, \xvec_s) \tau^4 \right) +
\nonumber\\ &\quad
    \ordlog{N^{-3/2}} \bigg(
        \phigrad{1}(\thetahat, \xvec_s) \tau +
        \left( \ord{1} \tau^5 + \ord{N^{-1/2}}  \tau^2\right)
            \phigrad{2}(\thetahat, \xvec_s) +
\nonumber\\ &\quad\quad\quad\quad\quad\quad\quad
        \ord{1} \tresid{3}{\phi(\cdot, \xvec_s), \theta, \thetahat} \tau^3
    \bigg)
\Bigg).
\eqlabel{r1_product_expansion1}
\end{align}

In the final line of \eqref{r1_product_expansion1} we absorbed some higher-order
quantities into the lower-order $\ordlog{N^{-3/2}}$ term for convenience of
expression, and used the fact that, when $\regevent$ occurs, for $k=1,2,3$,
$\norm{\likhatk{k}(\thetahat)}_2 = \ord{1}$ when $\regevent$ occurs, since
\begin{align*}
\norm{\likhatk{k}(\thetahat) - \likk{k}(\thetahat)}_2 \le&
    \sup_{\theta \in R_1} \norm{\likhatk{k}(\theta) - \likk{k}(\theta)}_2 \le
    \epsilon_U,
\end{align*}
so
\begin{align*}
\norm{\likhatk{k}(\thetahat)}_2 \le&
    \norm{\likhatk{k}(\thetahat) - \likk{k}(\thetahat)}_2 +
    \norm{\likk{k}(\thetahat) - \likk{k}(\thetatrue)}_2 +
    \norm{\likk{k}(\thetatrue)}_2 \\
\le&
    \epsilon_U +
    \sup_{\theta \in R_1} \norm{\likk{k}(\theta) - \likk{k}(\thetatrue)}_2 + \norm{\likk{k}(\thetatrue)}_2
= \ord{1},
\end{align*}
by continuity of $\likk{k}(\theta)$.

Next, we integrate both sides of \eqref{r1_product_expansion1} using
\lemref{region_r1_gaussian_moments} to replace Gaussian integrals over $R_1$
with Gaussian moments over the whole domain, up to an order $N^{-3}$ which can
be absorbed in the $\ordlog{N^{-2}}$ error term.
\begin{align*}
I_1(\phi) ={}&
\int_{R_1} \phi(\theta, \xvec_s)
   \exp\left(N \likhat(\theta) - N \likhat(\thetahat)\right) d \tau
\\
={}&
\phi(\thetahat, \xvec_s)
\int_{R_1} \exp\left(
    N \likhat(\theta) - N \likhat(\thetahat)
\right) d\tau +
\\{}&
N^{-1} \left(
    \frac{1}{2} \phigrad{2}(\thetahat, \xvec_s)
        \expect{\normhat}{\tau^2} +
    \frac{1}{6} \likhatk{3}(\thetahat) \phigrad{1}(\thetahat, \xvec_s)
        \expect{\normhat}{\tau^4} \right) +
\nonumber\\ &
    \ordlog{N^{-2}} \bigg(
        1 +
        \norm{\phigrad{1}(\thetahat, \xvec_s)}_2 +
        \norm{\phigrad{2}(\thetahat, \xvec_s)}_2 +
\nonumber\\ &
    \hspace{5em}
        \int_{R_1}
            \exp\left(-\frac{1}{2} \infohat \tau^2 \right)
            \tresid{3}{\phi(\cdot, \xvec_s), \theta, \thetahat}
            d\tau
\bigg).
\end{align*}

Finally, since $\tresid{3}{\phi(\cdot, \xvec_s), \theta, \thetahat}$ depends
explicitly on $\tau$, we must control this term in a more convenient form.  We
will argue as in \lemref{region_r1_gaussian_moments}.   Recall that, when
$\regevent$ occurs, $\infoevhat \ge \frac{\infoev}{2}$, so that $\normalizer =
O(1)$ and
\begin{align*}
\MoveEqLeft
\norm{\normalizer \int_{R_1} \exp\left(-\frac{1}{2} \infohat \tau^2 \right)
    \tresid{3}{\phi(\cdot, \xvec_s), \theta, \thetahat} d\tau }_2
\\\le&
\normalizer \int_{R_1} \exp\left(-\frac{1}{2} \infohat \tau^2 \right)
    \norm{\tresid{3}{\phi(\cdot, \xvec_s), \theta, \thetahat}}_2  d\tau
\\ \le&
\normalizer \int_{R_1} \exp\left(-\frac{\infoev}{4} \norm{\tau}_2^2 \right)
    \norm{\tresid{3}{\phi(\cdot, \xvec_s), \theta, \thetahat}}_2  d\tau.
%
\end{align*}
The result follows by multiplying by the $O(1)$ normalizer of
the $\normresid{\infoev / 2}{R_1}$ distribution.
\end{proof}
\end{lem}

\subsection{Combining the regions}\seclabel{bclt_region_comb}

We now combine the results of the previous three sections to evaluate
$I(\phi)$.
The accuracy of our BCLT expansion of $\expect{\post}{\phi(\theta, \xvec_s)}$ will
depend on the following residual function.
%
\begin{defn}\deflabel{bclt_resid}
Let $\infoev > 0$ denote the minimum eigenvalue of $\info$, and let $\delta > 0$
be an arbitrarily small number. For a function $\phi(\theta, \xvec_s)$ that
is third-order BCLT-okay, define the residual function
\begin{align*}
\resid{0}(\xvec_s) :={}&
    1 +
    \norm{\phi(\thetahat, \xvec_s)}_2 +
    \norm{\phigrad{1}(\thetahat, \xvec_s)}_2 +
    \norm{\phigrad{2}(\thetahat, \xvec_s)}_2
\\
\resid{1}(\xvec_s) :={}&
    \expect{\normresid{\infoev / 2}{R_1}(\tau)}{
        \norm{\tresid{3}{\phi(\cdot, \xvec_s), \theta, \thetahat}}_2}
\\
\resid{2}(\xvec_s) :={}&
    \expect{\normresid{\infoev / 2}{R_1 \bigcup R_2}(\tau)}{\norm{\phi(\theta, \xvec_s)}_2}
\\
\resid{3}(\xvec_s) :={}& \expect{\prior(\theta)}{\norm{\phi(\theta, \xvec_s)}_2}
\\
\resid{\phi}(\xvec_s) :={}&
    \resid{0}(\xvec_s) + \resid{1}(\xvec_s) + \resid{2}(\xvec_s) + \resid{3}(\xvec_s).
\end{align*}
\end{defn}

%
\begin{lem}\lemlabel{bayes_clt}
Let \assuref{bayes_clt} hold, and assume that $\phi$ is third-order BCLT-okay.
With $\fdist$-probability approaching one,
\begin{align*}
\MoveEqLeft
\expect{\post}{\phi(\theta, \xvec_s)} - \phi(\thetahat, \xvec_s) =
\\&
N^{-1} \left(
    \frac{1}{2} \phigrad{2}(\thetahat, \xvec_s) \expect{\normhat}{\tau^2} +
    \frac{1}{6} \likhatk{3}(\thetahat) \phigrad{1}(\thetahat, \xvec_s)
        \expect{\normhat}{\tau^4} \right) +
\\&
\ordlog{N^{-2}} \resid{\phi}(\xvec_s),
\end{align*}
where $\resid{\phi}(\xvec_s)$ is given in \defref{bclt_resid}.
\begin{proof}
First, combine the results of \lemref{region_r3, region_r2, region_r1_final_bound},
consolidating higher orders into lower orders and expanding the domain of
integration in the residual $\resid{\phi}(\xvec_s)$.
\begin{align}
I(\phi) ={}&
    \phi(\thetahat, \xvec_s)
    \int_{R_1} \exp\left(
        N \likhat(\theta) - N \likhat(\thetahat)
    \right) d\tau +
\nonumber\\ &\quad
    N^{-1} \normalizer \left(
        \frac{1}{2} \phigrad{2}(\thetahat, \xvec_s) \expect{\normhat}{\tau^2} +
        \frac{1}{6} \likhatk{3}(\thetahat) \phigrad{1}(\thetahat, \xvec_s)
            \expect{\normhat}{\tau^4} \right) +
\nonumber\\ &\quad
    \ordlog{N^{-2}} \resid{\phi}(\xvec_s).
\eqlabel{i_phi_expansion}
\end{align}
For convenience in the final bound of \thmref{bayes_clt_main} we have included the
term $\norm{\phi(\thetahat, \xvec_s)}_2$ in the definition of the residual even
though it does not occur in \lemref{region_r1_final_bound}.

The final expectation is given by the ratio
\begin{align*}
\expect{\post}{\phi(\theta, \xvec_s)} ={}& \frac{I(\phi)}{I(1)}.
\end{align*}
For the remainder of this proof, for compactness of notation, we define
\begin{align*}
\zeta :={}&
    \int_{R_1} \exp\left(N \likhat(\theta) - N \likhat(\thetahat) \right) d\tau \\
={}& \normalizer + \ordlog{N^{-1}} \\
={}& \normalizer (1 + \ordlog{N^{-1}}),
\end{align*}
which follows from \lemref{r1_normalizer_expansion} and the fact that
$\normalizer$ is bounded away from zero when $\regevent$ occurs. When
$\phi(\theta, \xvec_s) \equiv 1$, then $\resid{0}(\xvec_s) = \resid{1}(\xvec_s) = 0$,
so by \lemref{region_r3, region_r2},
\begin{align*}
I(1) = \zeta\left(1 + \ord{N^{-3}}\right).
\end{align*}
We now use the expansion, valid for $|z| < 1$ as $z \rightarrow 0$,
\begin{align}\eqlabel{inverse_expansion}
    \frac{1}{1 + z} = 1 - z + O(z^2).
\end{align}
Using \eqref{i_phi_expansion} for $I(\phi)$ and
applying \eqref{inverse_expansion} to $1 / I(1)$ and to $1 / \zeta$
gives
\begin{align*}
\frac{I(\phi)}{I(1)} ={}&
\frac{I(\phi)}{\zeta( 1 + \ord{N^{-3}})} \\
={}& \frac{I(\phi)}{\zeta} + \ord{N^{-3}} I(\phi) \\
={}&
\phi(\thetahat, \xvec_s) +
N^{-1} \frac{\normalizer}{\zeta} \left(
    \frac{1}{2} \phigrad{2}(\thetahat, \xvec_s) \expect{\normhat}{\tau^2} +
    \frac{1}{6} \likhatk{3}(\thetahat) \phigrad{1}(\thetahat, \xvec_s)
        \expect{\normhat}{\tau^4} \right) +
\\&
\ordlog{N^{-2}} \resid{\phi}(\xvec_s)
\\={}&
\phi(\thetahat, \xvec_s) +
N^{-1} \left(
    \frac{1}{2} \phigrad{2}(\thetahat, \xvec_s) \expect{\normhat}{\tau^2} +
    \frac{1}{6} \likhatk{3}(\thetahat) \phigrad{1}(\thetahat, \xvec_s)
        \expect{\normhat}{\tau^4} \right) +
\\&
\ordlog{N^{-2}} \resid{\phi}(\xvec_s).
\end{align*}
Note that we have absorbed terms from the expansion of $I(\phi)$ that are linear
in $|\phi(\thetahat, \xvec_s)|$, $\norm{\phigrad{2}(\thetahat, \xvec_s)}$, and
$\norm{\phigrad{3}(\thetahat, \xvec_s)}$ into the residual $\resid{\phi}(\xvec_s)$.
To do so we have used the fact that, when $\regevent$ occurs,
$\likhatk{3}(\thetahat) = \likhatk{3}(\thetatrue) + O(1)$.

Finally, the expansion holds with $\fdist$-probability approaching one by
\lemref{regular}.
\end{proof}
\end{lem}
%

\subsection{Proof of that the BCLT residual is square-summable}\applabel{resid_proofs}

All that remains to prove \thmref{bayes_clt_main} is to show that the residual is
$O_p(1)$ when summed.  The form of $\resid{\phi}(\xvec_s)$ is awkward, but it is
designed specifically to commute with suprema over sums, allowing us to apply
ULLNs to the remainder of our BCLT, as stated formally in
\lemref{integral_commuting_sums,taylor_commuting_sums} below.

Observe that residual expressions of the form $\sup_{\theta \in R_1 \bigcup R_2}
\norm{\phi(\theta, \x)}_2$ would not be as useful, since
\begin{align*}
\meann \sup_{\theta \in R_1 \bigcup R_2} \norm{\phi(\theta, \x_n)}_2 \ge
\sup_{\theta \in R_1 \bigcup R_2} \meann \norm{\phi(\theta, \x_n)}_2,
\end{align*}
and ULLN bound on the latter does not imply a bound on the former, since the
supremum may depend on $\x_n$ in a way that causes the sample mean to diverge.


\begin{lem}\lemlabel{integral_commuting_sums}
The integral of \defref{bclt_resid} commutes with sums.
For a scalar-valued function $\xi(\theta, \xvec_s)$ and an index set
$\indexset$,
\begin{align*}
\means \expect{\normresid{\infoev / 2}{R_1 \bigcup R_2}(\tau)}
              {\xi(\theta, \xvec_s)}^2 \le
\sup_{\theta \in R_1 \bigcup R_2} \means \xi(\theta, \xvec_s)^2.
\end{align*}
\begin{proof}
\begin{align*}
\means \expect{\normresid{\infoev / 2}{R_1 \bigcup R_2}(\tau)}
              {\xi(\theta, \xvec_s)}^2
\le{}&
\means \expect{\normresid{\infoev / 2}{R_1 \bigcup R_2}(\tau)}
            {\xi(\theta, \xvec_s)^2}
\\={}&
\expect{\normresid{\infoev / 2}{R_1 \bigcup R_2}(\tau)}
              {\means \xi(\theta, \xvec_s)^2}
\\\le{}&
\sup_{\tau \in R_1 \bigcup R_2}
    \means \xi(\thetahat + \tau / \sqrt{N}, \xvec_s)^2
\\={}&
\sup_{\theta \in R_1 \bigcup R_2} \means \xi(\theta, \xvec_s)^2.
\end{align*}
\end{proof}
\end{lem}


\begin{lem}\lemlabel{taylor_commuting_sums}
The Taylor series residual of \defref{bclt_resid} commutes with sums.
When $\theta \in R_1 \bigcup R_2$, for a k-th order continuously differentiable
function $\xi(\theta, \xvec_s)$ and an index set $\indexset$,
\begin{align*}
\means \norm{\tresid{k}{\xi(\cdot, \xvec_s), \theta, \thetahat}}_2^2 \le
\sup_{\theta \in R_1 \bigcup R_2}
    \means \norm{\xi_{(k)}(\theta, \xvec_s)}_2^2.
\end{align*}
\begin{proof}
Recall from \defref{taylor_residual} that
\begin{align*}
\MoveEqLeft
\means \norm{\tresid{k}{\xi(\cdot, \xvec_s), \theta, \thetahat}}_2^2
=\\{}&
\means \norm{
    \frac{1}{(k-1)!} \int_0^1 (1 - t)^{k-1}
    \xi_{(k)}(\thetahat + t (\theta - \thetahat), \xvec_s) dt
}_2^2
\le\\&
\means \frac{1}{(k-1)!} \int_0^1
\norm{\xi_{(k)}(\thetahat + t (\theta - \thetahat), \xvec_s)}_2^2 dt
=\\&
\frac{1}{(k-1)!} \int_0^1
\means \norm{\xi_{(k)}(\thetahat + t (\theta - \thetahat), \xvec_s)}_2^2 dt
\le\\&
\frac{1}{(k-1)!}
\sup_{\theta \in R_1 \bigcup R_2}
    \means\norm{\xi_{(k)}(\theta, \xvec_s)}_2^2.
\end{align*}
The result follows because $k \ge 1$.
\end{proof}
\end{lem}



%
\begin{lem}\lemlabel{resid_op1}
Under \assuref{core_bclt_assu}, $\meann\resid{\phi}(\x_n)^2 = O_p(1)$.
\begin{proof}
By \lemref{regular}, and in particular \defref{regular_event} \eventref{ulln},
together with the ULLN assumed in \defref{bclt_okay}, we have
\begin{align*}
\sup_{\theta \in R_1 \bigcup R_2}
    \meann \norm{\phigrad{k}(\theta, \x_n)}_2^2 = O_p(1),
\end{align*}
for $k=0,1,2,3$.  It follows immediately that $\meann
\norm{\phigrad{k}(\thetahat, \x_n)}_2 = O_p(1)$ and $\meann
\norm{\phigrad{k}(\thetahat, \x_n)}_2^2 = O_p(1)$ for $k=0, 1, 2, 3$. From this
and Cauchy-Schwartz it follows that $\meann \resid{0}(\x_n)^2 = O_p(1)$.

Applying \lemref{integral_commuting_sums} with $\indexset = \{1, \ldots, N \}$ and
$\xi(\theta, \xvec_s) = \norm{\phi(\theta, \x_n)}_2$ gives
$\meann \resid{2}(\x_n) = O_p(1)$ and $\meann \resid{2}(\x_n)^2 = O_p(1)$.

Similarly, $\meann \resid{1}(\x_n)^2 = O_p(1)$ follows from
\lemref{integral_commuting_sums} and then \lemref{taylor_commuting_sums}.

The terms $\meann \resid{3}(\x_n)$ and $\meann \resid{3}(\x_n)^2$ are controlled
by assumption in \defref{bclt_okay}.

Finally, the fact that $\meann \resid{\phi}(\x_n)^2 = O_p(1)$ follows
by Cauchy-Schwartz, since $\resid{\phi}(\x_n)$ is the sum of products
of terms, the squares of each of which have $O_p(1)$ sum.
\end{proof}
\end{lem}
%


\subsubsection{Proof of \thmref{ij_consistent}}\applabel{ij_consistent}
\begin{proof}
    We can re-arrange the covariance in the definition of $\infl_n$ as the product
    of terms centered at $\thetahat$:
    \begin{align*}
    \infl_n ={}& N \expect{\post}{
        (\g(\theta) - \g(\thetahat))
        (\ell(\x_n \vert \theta) - \ell(\x_n \vert \thetahat))
    } -
    \\&
    2 N \left(
        \g(\thetahat) - \expect{\post}{\g(\theta)}
    \right)
    \left(
        \ell(\x_n \vert \thetahat) - \expect{\post}{\ell(\x_n \vert \theta)}
    \right).
    \end{align*}
    Apply \thmref{bayes_clt_main} three times, with $\phi_\g(\theta) := \g(\theta)$,
    $\phi_\ell(\theta, \x_n) := \ell(\x_n \vert \theta)$, and $\phi_{\ell\g}(\theta,
    \x_n) := (\g(\theta) - \g(\thetahat)) (\ell(\x_n \vert \theta) - \ell(\x_n \vert
    \thetahat))$, all of which are third-order BCLT-okay by assumption. Let the
    residuals be denoted as $\resid{\g}$ (with no $\x_n$ dependence),
    $\resid{\ell}(\x_n)$, and $\resid{\ell\g}(\x_n)$, respectively.   Note that the
    first derivative of $\theta \mapsto \phi_{\ell\g}(\theta, \x_n)$ is zero at
    $\thetahat$.
    
    Write the leading term in the expansion of $\phi_\ell(\theta,
    \x_n)$ as
    \begin{align*}
    L(\x_n) :=
    \frac{1}{2} \ellgrad{2}(\thetahat, \x_n) \infohat^{-1} +
    \frac{1}{6} \ellgrad{1}(\thetahat, \x_n) \likhatk{3}(\thetahat) \hat{M}.
    \end{align*}
    Under \assuref{bayes_clt}, $\thetahat \plim \thetatrue$  and $\infohat^{-1}$ has
    bounded operator norm with probability approaching one (\propref{freq_clt}).
    These facts imply that $\meann \norm{L(\x_n)}_2^2 = \ordp{1}$ by Cauchy-Schwarz
    and a ULLN (\assuref{bayes_clt} \itemref{loglik_ulln} and \lemref{ulln}).
    
    Then \thmref{bayes_clt_main}, together with the fact that $\thetahat = \ordp{1}$
    and $\g(\theta)$ is continuous, gives
    \begin{align*}
    \infl_n ={}&
    \ggrad{1}(\thetahat) \infohat \ellgrad{1}(\x_n \vert \thetahat) +
    \ordlog{N^{-1}} \resid{\ell\g}(\x_n) +
    \\& N
        \left( \ordp{N^{-1}} L(\x_n) + \ordlogp{N^{-2}} \resid{\ell}(\x_n) \right)
        \left( \ordp{N^{-1}} + \ordlogp{N^{-2}} \resid{\g} \right)
    \\=&
    \ggrad{1}(\thetahat) \infohat \ellgrad{1}(\x_n \vert \thetahat) +
    \\&
    \ordlog{N^{-1}} \left(
    \resid{\ell\g}(\x_n) + L(\x_n) +
    \resid{\g} L(\x_n) +  \resid{\ell}(\x_n) +
    \resid{\ell}(\x_n) \resid{\g}
    \right).
    \end{align*}
    In the final line, we have combined lower orders into the $\ordp{N^{-1}}$ term
    for conciseness of expression.
    
    We now show that the sums of the squares of the residuals in the preceding
    display go to zero. By \thmref{bayes_clt_main}, each of $\meann
    \resid{\ell}(\x_n)^2$, and $\meann \resid{\ell\g}(\x_n)^2$ are $O_p(1)$.
    Similarly, $\meann \norm{\infohat^{-1} \ellgrad{1}(\x_n \vert \thetahat)}_2^2 =
    \ordp{1}$ by a ULLN and the boundedness of the operator norm of $\infohat^{-1}$.
    By Cauchy-Schwartz, we then have $\meann \ggrad{1}(\thetahat)^\trans
    \infohat^{-1} \ellgrad{1}(\x_n \vert \thetahat) \resid{\infl}(\x_n) = O_p(1)$.
    
    By repeated application of Cauchy-Schwartz, we
    can thus sum $\infl_n$ and $\infl_n \infl_n^\trans$ to get
    \begin{align*}
    \sumn \infl_n ={}&
     \ggrad{1}(\thetahat) \infohat \sumn \ellgrad{1}(\x_n \vert \thetahat) +
     \ordlogp{N^{-1}} = \ordlogp{N^{-1}} \\
    \meann \infl_n \infl_n^\trans ={}&  \ggrad{1}(\thetahat) \infohat^{-1}
     \scorecovhat
     \infohat^{-1} \ggrad{1}(\thetahat)^\trans + \ordlogp{N^{-1}}
    = \gcovmaphat + \ordlogp{N^{-1}}.
    \end{align*}
    Consistency of $\gcovij$ then follows from \corref{bayes_expectation_dist}.
\end{proof}

\newpage
\section{Theoretical details for the von Mises expansion}\applabel{higher_order}


\subsection{A higher--order posterior moment expansion}\applabel{bayes_clt_vm_proof}

The residual in the von--Mises expansion \eqref{post_expansion}
involves a posterior expectation of a product of three
centered quantities.  To analyze this term precisely, we require
a higher--order version of \thmref{bayes_clt_main}.


\subsubsection{A higher--order version of \thmref{bayes_clt_main}}


%
\begin{thm}\thmlabel{bayes_clt_vm}
    Let \assuref{bayes_clt} hold, and assume that $\phi(\theta, \xvec_s)$ is
    fifth-order BCLT-okay for an index set $\indexset$.  Further, assume that all of
    $\phi(\thetahat, \xvec_s)$, $\phigrad{1}(\thetahat, \xvec_s)$. and
    $\phigrad{2}(\thetahat, \xvec_s)$ are zero.\footnote{The assumption that the
    first two derivatives of $\phi$ are zero at $\thetahat$ is not essential to the
    series expansion; this assumption simply eliminates a large number of unneeded
    terms while still sufficing to analyze $\resid{T}(\t)$.}
    Let $\hat{M}_4$ denote the array of fourth moments and $\hat{M}_6$ the array of
    sixth moments of the Gaussian distribution $\normdist(0, \infohat^{-1})$. With
    $\fdist$-probability approaching one, for each $s \in \indexset$,
    \begin{align*}
    \expect{\post}{\phi(\theta, \xvec_s)} ={}&
        N^{-2} \left(
            \frac{1}{36} \phigrad{3}(\thetahat, \xvec_s) \likhatk{3}(\thetahat)
                \hat{M}_6 +
            \frac{1}{24} \phigrad{4}(\thetahat, \xvec_s) \hat{M}_4
    \right) +
    \nonumber\\ &
    \ordlog{N^{-3}} \resid{\phi}(\xvec_s),
    \end{align*}
    where $\means \resid{\phi}(\xvec_s)^2 = O_p(1)$, and the moment arrays are
    summed against the partial derivatives as in \thmref{bayes_clt_main}.
\begin{proof}

The reasoning is the same as \lemref{region_r1_final_bound}.  Define
\begin{align*}
\resid{4}(\xvec_s) :=
\norm{\phigrad{3}(\thetahat, \xvec_s)}_2 +
\norm{\phigrad{4}(\thetahat, \xvec_s)}_2 +
\expect{\normresid{\infoev / 2}{R_1}(\tau)}{
    \norm{\tresid{5}{\phi(\cdot, \xvec_s), \theta, \thetahat}}_2}.
\end{align*}
Then, for any $\xvec$ such that event $\regevent$ of \defref{regular_event}
occurs,
\begin{align*}
\expect{\post}{\phi(\theta, \xvec_s)} ={}&
    N^{-2} \left(
        \frac{1}{36} \phigrad{3}(\thetahat, \xvec_s) \likhatk{3}(\thetahat)
            \expect{\normhat}{\tau^6} +
        \frac{1}{24} \phigrad{4}(\thetahat, \xvec_s)
            \expect{\normhat}{\tau^4}
\right) +
\nonumber\\ &
\ordlog{N^{-3}}
\left(
\resid{2}(\xvec_s) + \resid{3}(\xvec_s) + \resid{4}(\xvec_s)
\right).
\end{align*}

As in \lemref{region_r1_loglik_expansion}, Taylor expand $\phi(\theta, \xvec_s)$
around $\thetahat$ and substitute $\tau = \sqrt{N}(\theta - \thetahat)$:
\begin{align*}
\phi(\theta, \xvec_s) ={}&
N^{-1} \tau^2 \left(
N^{-1/2} \frac{1}{6} \phigrad{3}(\thetahat, \xvec_s) \tau +
N^{-1} \frac{1}{24} \phigrad{4}(\thetahat, \xvec_s) \tau^2 +
N^{-3/2} \tresid{5}{\phi(\cdot, \xvec_s), \theta, \thetahat} \tau^3
\right).
\end{align*}
By comparing with \eqref{r1_phi_expansion}, the result follows by exactly
the same proof as that of \thmref{bayes_clt_main}, where we
\begin{itemize}
\item Use $\frac{1}{6} \phigrad{3}(\thetahat, \xvec_s)$ in place of
$\phigrad{1}(\thetahat, \xvec_s) $,
\item Use $\frac{1}{24} \phigrad{4}(\thetahat, \xvec_s)$ in place of
$\frac{1}{2} \phigrad{2}(\thetahat, \xvec_s)$,
\item Use $\tresid{5}{\phi(\cdot, \xvec_s), \theta, \thetahat}$
in place of $\tresid{3}{\phi(\cdot, \xvec_s), \theta, \thetahat}$,
\item Observe that $\resid{2}(\xvec_s)$ and $\resid{3}(\xvec_s)$ are
$\ord{N^{-3}}$ when $\regevent$ occurs, and
\item Incorporate the leading factor of $N^{-1} \tau^2$ before taking
expectations.
\end{itemize}
The above formal modifications lead directly to the desired result.
\end{proof}
\end{thm}

\subsection{Proof of \thmref{finite_dim_resid}}\applabel{finite_dim_resid_proof}
\begin{proof}


To analyze the residual, we will first use \thmref{bayes_clt_main,bayes_clt_vm}
to analyze $\resid{T}(\t)$ for fixed $\fndist$ and $\t$.  For compactness of
notation, we omit $\t$ for the application of \thmref{bayes_clt_vm}, returning
to it later to analyze the limiting behavior of each term as $N \rightarrow
\infty$.

As in the proof of \thmref{ij_consistent}, we re-write each term of $\htil(\x_n,
\x_m \vert \t)$ of \eqref{v_def} as a quantity centered at $\thetahat(\t)$
rather than the posterior expectation.  Define
\begin{align*}
\ghat(\theta) :={}& \g(\theta) - \g(\thetahat)
&
\ellhatunderbar(\xn \vert \theta) :={}&
    \ellunderbar(\xn \vert \theta) - \ellunderbar(\xn \vert \thetahat).
\end{align*}
For compactness, additionally define
\begin{align*}
\varepsilon_{\g} :={}& \g(\thetahat) - \expect{\post}{\g(\theta)}
\\
\varepsilon_{n} :={}&
    \ellunderbar(\x_n \vert \thetahat) -
    \expect{\post}{\ellunderbar(\x_n \vert \theta)},
\\
\varepsilon_{m} :={}&
    \ellunderbar(\x_m \vert \thetahat) -
    \expect{\post}{\ellunderbar(\x_m \vert \theta)}
\end{align*}
and
\begin{align*}
\phi^{\g n}(\theta) :={}& \ghat(\theta) \ellhatunderbar(\x_n \vert \theta)
\\
\phi^{\g m}(\theta) :={}&  \ghat(\theta) \ellhatunderbar(\x_m \vert \theta)
\\
\phi^{n m}(\theta) :={}&
    \ellhatunderbar(\x_n \vert \theta) \ellhatunderbar(\x_m \vert \theta)
\\
\phi^{\g n m}(\theta) :={}&
    \ghat(\theta) \ellhatunderbar(\x_n \vert \theta)
    \ellhatunderbar(\x_m \vert \theta).
\end{align*}
Using the fact that $\gbar(\theta) =
\ghat(\theta) + \varepsilon_\g$ (and so on), we can rearrange
\begin{align}
\MoveEqLeft
\expect{\post}{
    \gbar(\theta)
    \ellbarbar(\x_n \vert \theta)
    \ellbarbar(\x_m \vert \theta)
    } = \expect{\post}{\phi^{\g n m}(\theta)} +
\nonumber\\{}&
\expect{\post}{\phi^{\g n}(\theta)} \varepsilon_m +
\expect{\post}{\phi^{\g m}(\theta)} \varepsilon_n +
\expect{\post}{\phi^{n m}(\theta)} \varepsilon_\g +
\nonumber\\{}&
4 \varepsilon_\g \varepsilon_n \varepsilon_m.
\eqlabel{resid_centered}
\end{align}
We have $\sqrt{N} \resid{\t} = \meannm N^{5/2} \expect{\post}{ \gbar(\theta)
\ellbarbar(\x_n \vert \theta) \ellbarbar(\x_m \vert \theta) }$, to which each
term of \eqref{resid_centered} contributes one sum.

We will begin by considering the term $\meann N^{5/2} \expect{\post}{\phi^{\g n
m}(\theta)}$. Assume for the moment that $\phi_{\g n m}(\theta)$ is BCLT-okay
for $\indexset = [N] \times [N]$ (we will prove this later).  The derivatives
$\phigrad{k}^{\g n m}(\theta)$ are linear combinations of products of the form
$\ghatgrad{a}(\theta) \ellhatunderbargrad{b}(\x_n \vert \theta)
\ellhatunderbargrad{c}(\x_m \vert \theta)$, having $a + b + c = k$, $0 \le a,b,c
\le k$, and coefficient $k! / (a! b! c!)$.  In particular, $\phi_{\g n
m}(\theta)$ and its first two partial derivatives are all zero at $\thetahat$,
and
\begin{align*}
\phigrad{3}^{\g n m}(\thetahat) ={}&
    \ggrad{1}(\thetahat)
    \ellunderbargrad{1}(\x_n \vert \thetahat)
    \ellunderbargrad{1}(\x_m \vert \thetahat)
\\
\phigrad{4}^{\g n m}(\thetahat) ={}&
    3 \ggrad{2}(\thetahat)
    \ellunderbargrad{1}(\x_n \vert \thetahat)
    \ellunderbargrad{1}(\x_m \vert \thetahat) +
\\&
    3 \ggrad{1}(\thetahat)
    \ellunderbargrad{2}(\x_n \vert \thetahat)
    \ellunderbargrad{1}(\x_m \vert \thetahat) +
\\&
    3 \ggrad{1}(\thetahat)
    \ellunderbargrad{1}(\x_n \vert \thetahat)
    \ellunderbargrad{2}(\x_m \vert \thetahat).
\end{align*}
We can then apply \thmref{bayes_clt_vm} to write
\begin{align}
N^{5/2} \expect{\post}{\phi^{\g n m}(\theta)} ={}&
\sqrt{N} A_{111}
\ggrad{1}(\thetahat)
\ellunderbargrad{1}(\x_n \vert \thetahat)
\ellunderbargrad{1}(\x_m \vert \thetahat) \likhatk{3}(\thetahat) \hat{M}_6 +
\nonumber\\&
\sqrt{N} \sum_{ijk} A_{ijk}
    \ggrad{i}(\thetahat)
    \ellunderbargrad{j}(\x_n \vert \thetahat)
    \ellunderbargrad{k}(\x_m \vert \thetahat) \hat{M}_4 +
\nonumber\\&
\ordlog{N^{-1/2}} \resid{\g n m}(\x_n, \x_m),
\eqlabel{phi_gnm_bclt}
\end{align}
for $\{i, j, k\}$ ranging over $\{1, 1, 2\}$, $\{1, 2, 1\}$, and $\{2, 1, 1\}$
and suitably chosen constants $A_{ijk}$.

Let us now re-introduce $\t$ dependence, writing $\thetahat(\t)$ in place of
$\thetahat$.
By \thmref{bayes_clt_vm},
\begin{align*}
\meannm \ordlog{N^{-1/2}} \resid{\g n m}(\x_n, \x_m) \plim{}& 0,
\end{align*}

By \lemref{thetahat_consistent}, the terms $\hat{M}_4$ and $\hat{M}_6$ in
\thmref{bayes_clt_vm}
converge in probability to constants, and by smoothness of $\g(\theta)$ and
convergence by \lemref{thetahat_consistent} of $\thetahat$ to a constant,
$\sup_{\t \in [0,1]} \norm{\g(\thetahat(\t))}_2$ and $\sup_{\t \in [0,1]}
\likhatk{3}(\thetahat(\t))$ are both $\ordp{1}$.  So it suffices to analyze, for
$j \in \{1,2\}$,
\begin{align*}
\MoveEqLeft
\sup_{\t \in [0,1]}
\norm{
\frac{1}{N^2} \sum_{n=1}^N \sum_{m=1}^N \sqrt{N}
    \ellunderbargrad{1}(\x_n \vert \thetahat(\t))
    \ellunderbargrad{j}(\x_m \vert \thetahat(\t))
}_2
={}\\&
\sup_{\t \in [0,1]}
\norm{
\frac{1}{\sqrt{N}} \sumn
    \ellunderbargrad{1}(\x_n \vert \thetahat(\t))
\meanm
    \ellunderbargrad{j}(\x_m \vert \thetahat(\t))
}
\le{}\\&
\sup_{\t \in [0,1]} \norm{
\frac{1}{\sqrt{N}} \sumn
    \ellunderbargrad{1}(\x_n \vert \thetahat(\t))
}_2
\sup_{\t \in [0,1]} \norm{
\meanm
    \ellunderbargrad{j}(\x_m \vert \thetahat(\t))
}_2
\\&\plim{} 0,
\end{align*}
where the final line follows from \lemref{ulln}, and the fact that
$\sup_{\t \in [0,1]} \thetahat(\t) \plim \thetatrue$ by
\lemref{bclt_okay_t}.

It follows that
\begin{align*}
\sup_{\t \in [0,1]}
\norm{
\meannm
    N^{5/4} \expect{\postft}{\phi^{\g n m}(\theta)}}_2 \plim 0.
\end{align*}

Analogous reasoning applied to the remaining terms of \eqref{resid_centered},
but using \thmref{bayes_clt_main} instead of \thmref{bayes_clt_vm}, gives that
each term involving a $\ell(\x_n \vert \theta)$ is of order $N^{-3/2}$ when
summed over $N$ --- a of $N$ for the expansion of \thmref{bayes_clt_main}, and a
factor of $\sqrt{N}$ from the fact that $\sup_{\t \in [0,1]}\frac{1}{\sqrt{N}}
\sumn \ellgrad{k}(\x_n \vert \thetahat(\t)) = \ordp{1}$ by \lemref{ulln}.
Similarly, $\expect{\postfn}{\phi^{n m}(\theta)}$ is order $N^{-2}$,
since we get one factor of $\sqrt{N}$ for each of
$\ell(\x_n \vert \theta)$ and $\ell(\x_m \vert \theta)$.
Specifically,
\begin{align*}
\sup_{\t \in [0,1]}
\norm{
\meann N^{3/2} \expect{\postft}{\phi^{\g n}(\theta)}
}_2 =& \ordp{1}
\\
\sup_{\t \in [0,1]}  \norm{
    \meann  N^{3/2} \varepsilon_n
}_2 =& \ordp{1}
\\
\sup_{\t \in [0,1]}
\norm{
\meannm N^2 \expect{\postft}{\phi^{n m}(\theta)}
}_2 =& \ordp{1}
\\
\sup_{\t \in [0,1]} N \varepsilon_\g =& \ordp{1}.
\end{align*}
Combining, we see that
\begin{align*}
\sup_{\t \in [0,1]}
\norm{\meannm
N^{5/2} \expect{\postft}{
    \gbar(\theta)
    \ellbarbar(\x_n \vert \theta)
    \ellbarbar(\x_m \vert \theta)
    }}_2 = \ordp{N^{-1/2}} \plim 0.
\end{align*}

It follows each of term of \eqref{resid_centered} converges to zero in
probability.

It remains to show that $\phi^{\g n m}(\theta)$ is fifth-order BCLT-okay
according to \defref{bclt_okay} with $\indexset = [N] \times [N]$, uniformly in
$\t$, using \assuref{bayes_clt_vm}.

First, \defref{no_data_bclt_okay} \itemref{eprior_finite,cont_diffable} are
satisfied by \assuref{bayes_clt_vm} and \assuref{bayes_clt}, respectively, and
so \defref{bclt_okay} \itemref{bclt_okay_as} is satisfied. \Defref{bclt_okay}
\itemref{prior_op1} is also satisfied by \assuref{bayes_clt_vm}.

Finally, for \defref{bclt_okay} \itemref{grad_op1},
\begin{align*}
\MoveEqLeft
\sup_{\theta \in \thetaball{\delta}}
    \meannm \norm{\g(\theta)
                  \ell(\x_n \vert \theta) \ell(\x_m \vert \theta)}_2^2
\le{}\\&
\sup_{\theta \in \thetaball{\delta}}
    \norm{g(\theta)}
\sup_{\theta \in \thetaball{\delta}}
    \meann \norm{\ell(\x_n \vert \theta)}_2^2
\sup_{\theta \in \thetaball{\delta}}
    \meanm \norm{\ell(\x_m \vert \theta)}_2^2,
\end{align*}
which is $\ordp{1}$ by \assuref{bayes_clt} and \lemref{ulln} for sufficiently
small $\delta$.  Since, by \lemref{bclt_okay_t}, $\thetahat(\t) \in
\thetaball{\delta}$ with probability approaching one for any $\delta > 0$,
uniformly in $\t \in [0,1]$, it follows that   $\phi^{\g n m}(\theta)$ is
fifth-order BCLT-okay.


Finally, assume that $(\infltil_n)^2$ is uniformly integrable
with respect to $\fdist(\x_n)$ as $N \rightarrow \infty$.  Then
\begin{align*}
    \frac{1}{\sqrt{N}} \meann \infltil_n - \expect{\fdist{\x_n}}{\infltil_n}
    \dlim 
    \normdist\left(0, \lim_{N\rightarrow\infty} \var{\fdist(\x_n)}{\infltil_n}
    \right)
\end{align*}
by the FCLT, using the fact that $\infltil_n$ depens on $\x_n$ only
through the argument $\ell(\x_n \vert \theta)$ and not through the
posterior expectation, which is evaluated at $\fdist$.

By \lemref{bclt_okay_t}, \thmref{ij_consistent} applies
to $\infltil_n$ as well as to $\infl_n$, and that the two
converge to the same limit, since $\thetahat(0)$ and $\thetahat(1)$
both converge in probability to the same quantity.  From this, it follows
that
\begin{align*}
    \lim_{N\rightarrow\infty} \covhat{[N]}{\infltil_n} \rightarrow
    \gcovtrue = \lim_{N\rightarrow\infty} \var{\fdist(\x_n)}{\infltil_n}.
\end{align*}

Finally, the limiting distribution of $\sqrt{N}\left(\expect{\post}{\g(\theta)}
- \expect{\postf}{\g(\theta)} \right)$ follows from plugging into
\eqref{post_expansion} and applying Slutsky's theorem.

\end{proof}


\subsection{Supporting results for the von Mises expansion}\applabel{von_mises_helpers}
\subsubsection{Centered expectations}

\begin{lem}\lemlabel{centered_expectations}
    Concentration of centered expectations.  Consider any $\post$-integrable
    functions $a(\theta)$, $b(\theta)$, and $c(\theta)$.
    Define
    \begin{align*}
    \bar{a}(\theta) := a(\theta) - \expect{\post}{a(\theta)}
    \textrm{, }\quad
    \hat{a}(\theta) := a(\theta) - a(\thetahat)
    \textrm{, and }
    \varepsilon_a := a(\thetahat) - \expect{\post}{a(\theta)}.
    \end{align*}
    with analogous definitions for $b$ and $c$.  Then
    \begin{align*}
    \expect{\post}{\bar{a}(\theta) \bar{b}(\theta)} ={}&
    \expect{\post}{\hat{a}(\theta) \hat{b}(\theta)} - \varepsilon_a \varepsilon_b.
    \end{align*}
    and
    \begin{align*}
    \MoveEqLeft
    \expect{\post}{\bar{a}(\theta) \bar{b}(\theta) \bar{c}(\theta)} =
    \expect{\post}{\hat{a}(\theta) \hat{b}(\theta) \hat{c}(\theta)} +
    \\&
        \expect{\post}{\hat{a}(\theta) \hat{b}(\theta)} \varepsilon^c +
        \expect{\post}{\hat{a}(\theta) \hat{c}(\theta)} \varepsilon^b +
        \expect{\post}{\hat{b}(\theta) \hat{c}(\theta)} \varepsilon^a -
        2 \varepsilon^a \varepsilon^b \varepsilon^c.
    \end{align*}
    \begin{proof}
    Formally adding and subtracting,
    $\bar{a}(\theta) = \hat{a}(\theta) + \varepsilon_a$, so
    \begin{align*}
    \expect{\post}{\bar{a}(\theta) \bar{b}(\theta)}
    ={}&
    \expect{\post}{\hat{a}(\theta)} \varepsilon_b +
    \expect{\post}{\hat{a}(\theta) \hat{b}(\theta)} +
    \expect{\post}{\hat{b}(\theta)} \varepsilon_a +
    \varepsilon_a \varepsilon_b
    \\={}&
    \expect{\post}{\hat{a}(\theta) \hat{b}(\theta)} - \varepsilon_a \varepsilon_b.
    \end{align*}
    Similarly,
    \begin{align*}
    \MoveEqLeft
    \expect{\post}{\bar{a}(\theta) \bar{b}(\theta) \bar{c}(\theta)} =
    \\&
    \expect{\post}{\hat{a}(\theta) \hat{b}(\theta) \hat{c}(\theta)} +
        \expect{\post}{\hat{a}(\theta) \hat{b}(\theta)} \varepsilon^c +
        \expect{\post}{\hat{a}(\theta) \hat{c}(\theta)} \varepsilon^b +
        \expect{\post}{\hat{b}(\theta) \hat{c}(\theta)} \varepsilon^a +
    \\&
        \varepsilon^a \varepsilon^b \expect{\post}{\hat{c}(\theta) } +
        \varepsilon^a \varepsilon^c \expect{\post}{\hat{b}(\theta) } +
        \varepsilon^b \varepsilon^c \expect{\post}{\hat{a}(\theta) } +
        \varepsilon^a \varepsilon^b \varepsilon^c =
    \\\MoveEqLeft
    \expect{\post}{\hat{a}(\theta) \hat{b}(\theta) \hat{c}(\theta)} +
    \\&
        \expect{\post}{\hat{a}(\theta) \hat{b}(\theta)} \varepsilon^c +
        \expect{\post}{\hat{a}(\theta) \hat{c}(\theta)} \varepsilon^b +
        \expect{\post}{\hat{b}(\theta) \hat{c}(\theta)} \varepsilon^a -
        2 \varepsilon^a \varepsilon^b \varepsilon^c.
    \end{align*}
    \end{proof}
    \end{lem}


\subsubsection{Applicability of \thmref{bayes_clt_main} to $\ftdist$}\applabel{ftdist}

For a fixed $\t$, we can define the analogues of \defref{loglik_def} as
\begin{align*}
\likhat(\theta \vert \t) :={}& \t \likhat(\theta) + (1 - \t)\lik(\theta)
\quad\textrm{and}\quad
\thetahat(\t) :={} \argmax_{\theta \in \thetadom} \likhat(\theta \vert \t),
\end{align*}
and so on.  Note that $\expect{\fdist(\xn)}{\ell(\xn \vert \theta, \t)} =
\lik(\theta)$, so $\lik(\theta)$ and its derived quantities need no $\t$
dependence.  
An advantage of our chosen parameterization is that,
for a fixed $\t$, we can apply \thmref{bayes_clt_main} to $\postft$ for a
suitably chosen likelihood. Specifically, for a fixed $\t$,
\begin{align*}
\int \ell(\xn \vert \theta) \ftdist(d \xn) ={}&
\t \meann \ell(\x_n \vert \theta) + (1 - \t) \lik(\theta)
=: \meann \ell(\xn \vert \theta, \t)
\\\textrm{where }
\ell(\xn \vert \theta, \t) :={}& \t \ell(\xn \vert \theta) + (\t - 1)
\lik(\theta).
\end{align*}

Using \lemref{bclt_okay_t}, we can apply \thmref{bayes_clt_main} to $\p(\theta
\vert \ftdist, N)$ with no modification other than to replace the finite-sample
quantities of \defref{loglik_def} with their analogues for $\ell(\xn \vert
\theta, \t)$, which will depend on $\t$ but will converge to their population
quantities uniformly in $\t \in [0, 1]$.

We then have the following lemma.


\begin{lem}\lemlabel{bclt_okay_t}
Assume that both $\ell(\xn \vert \theta)$ and $\lik(\theta)$ satisfy the
conditions of \assuref{bayes_clt} and are $K$-th order BCLT-okay with index set
$\indexset = [N]$.  Then, for any $\t \in [0,1]$, $\ell(\xn \vert \theta, \t)$
also satisfies \assuref{bayes_clt} and is $K$-th order BCLT-okay for $\indexset =
[N]$.
Further, $\sup_{\t \in [0,1]} \thetahat(\t) \plim{} \thetatrue$.
\begin{proof}
\Assuref{bayes_clt} \itemref{finite_dim,prior_smooth,prior_proper} do not depend
on the particular form $\ell(\xn \vert \theta, \t)$, and \assuref{bayes_clt}
\itemref{loglik_smooth,loglik_ulln,strict_opt} hold since $\ell(\xn \vert
\theta, \t)$ is a linear combination of  $\ell(\xn \vert \theta)$ and
$\lik(\theta)$.  Therefore \assuref{bayes_clt} holds for $\ell(\xn \vert \theta,
\t)$.

Since $\thetadim$ is finite and remains fixed by \assuref{bayes_clt},
\begin{align*}
\expect{\prior(\theta)}{\norm{\ellgrad{k}(\xn \vert \theta, \t)}_2} \le{}&
\thetadim \expect{\prior(\theta)}{\norm{\ellgrad{k}(\xn \vert \theta, \t)}_\infty}
\\\le{}&
\sqrt{\thetadim}\left(
\expect{\prior(\theta)}{\norm{\ellgrad{k}(\xn \vert \theta)}_\infty} +
\expect{\prior(\theta)}{\norm{\likk{k}(\theta)}_\infty}
\right)< \infty,
\end{align*}
so \defref{no_data_bclt_okay} \itemref{eprior_finite} is satisfied.
\Defref{no_data_bclt_okay} \itemref{cont_diffable} is satisfied by linearity,
and so \defref{bclt_okay} \itemref{bclt_okay_as} is also satisfied.

As above, we can write
\begin{align*}
\MoveEqLeft
\norm{
    \t \ellgrad{k}(\xn \vert \theta) + (1- \t)\likk{k}(\theta)
    }^2_2
\le{}\\&
\thetadim\left(
\means \norm{
    \ellgrad{k}(\xn \vert \theta) }_\infty +
\means \norm{
    \likk{k}(\theta)}^\infty
\right),
\end{align*}
from which it follows that \defref{bclt_okay} \itemref{prior_op1,grad_op1}
are satisfied.

For the convergence of $\sup_{\t \in [0,1]} \thetahat(\t)$, reason as in
\lemref{thetahat_consistent} to see that
\begin{align*}
\norm{\thetahat(\t) - \thetatrue}_2
    \le{}& 2 \infoev^{-1} \norm{\likhatk{1}(\thetatrue | \t)}_2
 ={}
2 \infoev^{-1} \t
\norm{\likhatk{1}(\thetatrue)}_2 \plim 0.
\end{align*}
\end{proof}
\end{lem}

%

\subsection{Supporting results for high--dimensional exponential families}


We begin by defining some additional quantities.  

\begin{defn}\deflabel{high_dim_exp_terms_more}
First, define the stacked vectors
\begin{align*}
\eta(\gamma, \lambda) = 
    \begin{pmatrix}
        \eta_1(\gamma, \lambda) \\
        \vdots \\
        \eta_G(\gamma, \lambda)
    \end{pmatrix} \in \rdom{\ydim G}
\quad\textrm{and}\quad
\z_n =
\begin{pmatrix}
    a_{n1} y_n \\
    \vdots\\
    a_{nG} y_n
\end{pmatrix} =
\a_n \otimes \y_n
\in \rdom{\ydim G},
\end{align*}
so that the log likelihood of a single observation can be written as
\begin{align*}
\ell(\x_n \vert \gamma, \lambda) ={}& 
    \sumg a_{ng} \y_n^\trans \eta_g(\gamma, \lambda) ={}
    \z_n^\trans \eta(\gamma, \lambda).
\end{align*}
Then, define
    \begin{align*}
        \m :={}& (m_1^\trans, \ldots, m_G^\trans )^\trans \in \rdom{D G}\\
        \S :={}& \diag{\S_1, \ldots, \S_G} \in \rdom{D G} \times \rdom{D G} \\
        \mu(\gamma) :={}& (\mu_1(\gamma)^\trans, \ldots, \mu_G^\trans(\gamma) )^\trans \in \rdom{D G} \\
        \M  :={}&
            N^2 \expect{\postgf}{\gbar(\gamma) \mubar(\gamma) \mubar(\gamma)^\trans} 
            \in \rdom{D G} \times \rdom{D G} 
            \\
        \L :={}& 
            N \expect{\postgf}{\gbar(\gamma) \cov{\postlt}{\eta(\gamma, \lambda)} } 
            \in \rdom{D G} \times \rdom{D G} 
    \end{align*}.
\end{defn}


In terms of \defref{high_dim_exp_terms_more}, we have
\begin{align*}
    \expect{\fdist(\x_n)}{\a_{ng} \y_n} ={}& 
        \expect{\fdist(\a_n)}{\expect{\y_n \vert \a_{ng} = 1}{\y_n}} ={} 
        \frac{1}{G} m_g \Rightarrow \\
    \expect{\fdist(\x_n)}{\z_n} ={}& \frac{1}{G} m
    \quad\textrm{and}\quad
    \zbar_n = \z_n - \frac{1}{G}m  \Rightarrow \\
    \expect{\fdist(\x_n)}{\zbar_n \zbar_n^\trans } ={}&
    \expect{\fdist(\x_n)}{\left( \z_n - \frac{1}{G} m \right) \left( \z_n - \frac{1}{G} m \right)^\trans } 
    \\={}&
    \expect{\fdist(\x_n)}{\z_n \z_n^\trans} - \frac{1}{G^2}  m m^\trans 
    \\={}&
    \frac{1}{G} \S - \frac{1}{G^2}  m m^\trans 
\end{align*}
Also,
\begin{align*}
\expect{\postglt}{\eta(\gl)} 
    ={}& \expect{\postgt}{\mu(\gamma)} \Rightarrow \\
\etabar(\gl) ={}&
    \eta(\gl) - \expect{\postgt}{\mu(\gamma)}  
\\={}&
    \eta(\gl) - \mu(\gamma) + \mu(\gamma) - \expect{\postgt}{\mu(\gamma)}  
\\={}&
    \eta(\gl) - \mu(\gamma) + \mubar(\gamma).
\end{align*}
So we can write
$\ellbarbar(\x_n \vert \gamma, \lambda) = 
\zbar_n^\trans \etabar(\gamma, \lambda)$.

Plugging this and \defref{high_dim_exp_terms_more} into \eqref{highdim-resid} gives
\begin{align*}
\htil(\x_n, \x_m) ={}&
N^2 \expect{\postglt}{
    \gbar(\gamma)
    \ellbarbar(\x_n \vert \gl)
    \ellbarbar(\x_m \vert \gl)    
    }
\\={}&
N^2 \expect{\postglt}{
    \gbar(\gamma)
    \left( \zbar_n^\trans \etabar(\gamma, \lambda) \right)
    \left( \zbar_m^\trans \etabar(\gamma, \lambda) \right)
}
\\={}&
N^2 \zbar_m^\trans
\expect{\postglt}{
    \gbar(\gamma)
    \etabar(\gamma, \lambda)
    \etabar(\gamma, \lambda)^\trans 
} \zbar_n 
\\={}&
N^2 \zbar_m^\trans
\expect{\postgt}{
    \gbar(\gamma)
    \expect{\postlt}{
        \left( \eta(\gl) - \mu(\gamma) + \mubar(\gamma) \right)
        \left( \eta(\gl) - \mu(\gamma) + \mubar(\gamma) \right)^\trans
    }
} \zbar_n 
\\={}&
N^2 \zbar_m^\trans
\expect{\postgt}{
    \gbar(\gamma)
    \left( 
        \cov{\postlt}{\eta(\gl)} + \mubar(\gamma)\mubar(\gamma)^\trans
    \right)
} \zbar_n 
\\={}&
\zbar_m^\trans \left(N \L + \M \right) \zbar_n.
\end{align*}

We thus have
\begin{align}\eqlabel{resid_nm}
    \resid{T}(1) ={}& 
        \frac{1}{N^2} \sumnm \zbar_m^\trans \left(N \L + \M \right) \zbar_n
    \nonumber \\={}&
    \frac{1}{N^2} \sumn N \zbar_n^\trans \L \zbar_n +
    \frac{1}{N^2} \sum_{n \ne m} N \zbar_n^\trans \L \zbar_m +
    \frac{1}{N^2} \sumn \zbar_n^\trans \M \zbar_n +
    \frac{1}{N^2} \sum_{n \ne m} \zbar_n^\trans \M \zbar_m.
\end{align}





\subsection{Proof of \thmref{resid_divergence}}\applabel{resid_divergence}

The result of \thmref{resid_divergence} follows from the following lemmas.


%
\begin{lem}\lemlabel{high_dim_resid_order}
Under the conditions of \thmref{resid_divergence},
\begin{align*}
    \resid{T}(1) = 
    \frac{1}{N^2} \sumn N \zbar_n^\trans \L \zbar_n +
    \frac{1}{N^2} \sum_{n \ne m} N \zbar_n^\trans \L \zbar_m +
    \ordlog{N^{-1}} 
    =
    \frac{1}{N^2} \sumn N \zbar_n^\trans \L \zbar_n +
    \ordlog{G^{-1/2}} +
    \ordlog{N^{-1}}.
\end{align*}
\end{lem}

\begin{proof}
We bound the expected absolute value of the terms in \eqref{resid_nm}
one--by--one; then the result follows by the triangle inequality.
\begin{align*}
\expect{\xfdist}{\abs{ 
    \frac{1}{N^2} \sumn \zbar_n^\trans \M \zbar_n
    }} \le {}&
\expect{\xfdist}{
    \frac{1}{N^2} \sumn \abs{ \zbar_n^\trans \M \zbar_n
    }} 
\\= {}&
\frac{1}{N} \expect{\fdist(\x_n)}{
    \abs{ \zbar_n^\trans \M \zbar_n }} 
\\= {}&
\frac{1}{N} \sumg \expect{\fdist(\y_n | \a_{ng} = 1)}{
    \abs{ \zbar_n^\trans \M \zbar_n }} \p(\a_{ng} = 1)
\\= {}&
\frac{1}{N} \frac{1}{G} \sumg 
    \expect{\fdist(\y_n \vert \a_{ng} = 1)}{
    \abs{ \left(y_n - m_g / G \right)^\trans \M_g \left( y_n - m_g / G \right) }} 
\\\le {}&
\frac{1}{N} \frac{1}{G} \sumg 
    \expect{\fdist(\y_n \vert \a_{ng} = 1)}{
    \left(y_n - m_g / G \right)^\trans \left( y_n - m_g / G \right) 
    } \sqrt{\norm{\M_g}_F^2} 
\\={}&
\frac{1}{N} \frac{1}{G} \sumg 
\left(\trace{\S_g} - 2 m_g^2 / G + m_g^2 / G^2 \right)
     \sqrt{\norm{\M_g}_F^2} 
\\={}&
\frac{1}{N} \frac{1}{G} \sumg 
\left(\trace{\S_g} - m_g^2 \left( 2 / G - 1 / G^2\right) \right)
    \sqrt{\norm{\M_g}_F^2} 
\\\le{}&
\frac{1}{N} \frac{1}{G} \sumg  \trace{\S_g} \sqrt{\norm{\M_g}_F^2} 
\\={}& \ordlog{N^{-1}} 
\quad\textrm{by \assuref{high_dim_resid_assu} \itemref{m_matrix_avg_op}}.
\end{align*}

Next, we control the terms with $\sum_{n \ne m}$, beginning again 
with the $\M$ term.  By Jensen's inequality,
\begin{align*}
\expect{\xfdist}{\abs{
    \frac{1}{N^2} \sum_{n \ne m} \zbar_n^\trans \M \zbar_m
    }}
\le{}&
    \sqrt{
    \expect{\xfdist}{\left( 
    \frac{1}{N^2} \sum_{n \ne m} \zbar_n^\trans \M \zbar_m
    \right)^2}}
\\={}&
\sqrt{
\expect{\xfdist}{ 
    \frac{1}{N^4} 
        \sum_{n \ne m} \sum_{n \ne m}
        \zbar_n^\trans \M \zbar_m         \zbar_{n'}^\trans \M \zbar_{m'}
}}.
\end{align*}
Since $\z_n \perp \z_m$ for $n \ne m$, and $\expect{\fdist(\x_n)}{\zbar_n} = 0$,
terms in the preceding sum are zero unless $n = n'$ and $m = m'$ or $n = m'$ and
$m = n'$.  Since $n \ne m$ and $n' \ne m'$, there are $N$ choices for $n$, and
then $N - 1$ choices for $m$, from which there are two nonzero choices for
$n'$ and $m'$.  From this and the fact that $\z_n$ are IID, we get that
\begin{align*}
\expect{\xfdist}{\left( 
    \frac{1}{N^2} \sum_{n \ne m} \zbar_n^\trans \M \zbar_m
    \right)^2} 
={}&
\frac{2 N (N - 1)}{N^4} 
    \expect{\fdistnm}{\zbar_n^\trans \M \zbar_m \zbar_m^\trans \M \zbar_n}
\\={}&
\ord{N^{-2}} \trace{
    \expect{\fdistnm}{\zbar_n \zbar_n^\trans}
    \M 
    \expect{\fdistnm}{\zbar_n \zbar_n^\trans}
    \M }
\\={}&
\ord{N^{-2}} \trace{
    \left(\frac{1}{G} \S - \frac{1}{G^2} \m \m^\trans \right)
    \M 
    \left(\frac{1}{G} \S - \frac{1}{G^2} \m \m^\trans \right)
    \M }
\\={}&
\ord{N^{-2}} \left( 
\frac{1}{G^2} \trace{\S \M \S \M} -
\frac{1}{G^3} \trace{\S \M \m \m^\trans \M} -
\frac{1}{G^3} \trace{\m \m^\trans \M \S \M} +
\frac{1}{G^4} \trace{\m \m^\trans \M \m \m^\trans \M}
\right)
\\={}&
\ord{N^{-2}} \left( 
\frac{1}{G^2} \norm{\S^{1/2} \M \S^{1/2} }_F^2 -
2 \frac{1}{G^3} \m^\trans \M \S \M \m +
\left(\frac{1}{G^2}  \m^\trans \M \m \right)^2
\right) 
\\={}& \ordlog{N^{-2}},
\end{align*}
since
\begin{align*}
\frac{1}{G^2} \norm{\S^{1/2} \M \S^{1/2} }_F^2 ={}& \ordlog{1}
    & \textrm{ by \assuref{high_dim_resid_assu} \itemref{m_matrix_frob}} \\
\frac{1}{G^3} \m^\trans \M \S \M \m \le{}
    \frac{1}{G^2} \sqrt{\norm{\S^{1/2} \M \S \M \S^{1/2}}_F^2}
    \le{}
    \frac{1}{G^2} \norm{\S^{1/2} \M \S^{1/2}}_F^2 ={}& \ordlog{1}
    & \textrm{ by \assuref{high_dim_resid_assu} \itemref{m_matrix_frob}, \lemref{m_rescaling}} \\
\frac{1}{G^2}  \m^\trans \M \m \le{}
    \frac{1}{\sqrt{G}} \sqrt{ \frac{1}{G}\norm{\S^{1/2} \M \S^{1/2}}_F^2 } ={}& \ordlog{G^{-1/2}}.
    & \textrm{ by \assuref{high_dim_resid_assu} \itemref{m_matrix_frob}, \lemref{m_rescaling}} \\
\end{align*}
Therefore
\begin{align*}
    \expect{\xfdist}{\abs{
        \frac{1}{N^2} \sum_{n \ne m} \zbar_n^\trans \M \zbar_m
        }} = \ordlog{N^{-1}}.
\end{align*}

Analogously, we have
\begin{align*}
\expect{\xfdist}{\abs{ 
    \frac{1}{N^2} \sum_{n \ne m} N \zbar_n^\trans \L \zbar_m
    }}^2
\le{}&
\expect{\xfdist}{\left( 
    \frac{1}{N^2} \sum_{n \ne m} N \zbar_n^\trans \L \zbar_m
    \right)^2} 
\\={}&
\ord{N^{-2}} \trace{
    \left(\frac{1}{G} \S - \frac{1}{G^2} \m \m^\trans \right)
    N \L 
    \left(\frac{1}{G} \S - \frac{1}{G^2} \m \m^\trans \right)
    N \L }
\\={}&
\ord{1} \left( 
    \frac{1}{G^2} \trace{\S \L \S \L} -
    \frac{1}{G^3} \trace{\S \L \m \m^\trans \L} -
    \frac{1}{G^3} \trace{\m \m^\trans \L \S \L} +
    \frac{1}{G^4} \trace{\m \m^\trans \L \m \m^\trans \L}
\right)
\\={}&
    \ord{1} \left( 
    \frac{1}{G^2} \norm{\S^{1/2} \L \S^{1/2} }_F^2 -
    2 \frac{1}{G^3} \m^\trans \L \S \L \m +
    \left(\frac{1}{G^2} \m^\trans \L \m \right)^2
    \right) 
\\={}& 
\ordlog{G^{-1}},
\end{align*}
since
\begin{align*}
\frac{1}{G^2} \norm{\S^{1/2} \L \S^{1/2} }_F^2 ={}& \ordlog{G^{-1}}
    & \textrm{ by \assuref{high_dim_resid_assu} \itemref{l_matrix_frob}} \\
\frac{1}{G^3} \m^\trans \L \S \L \m \le{}
    \frac{1}{G^2} \sqrt{\norm{\S^{1/2} \L \S \L \S^{1/2}}_F^2}
    \le{}
    \frac{1}{G^2} \norm{\S^{1/2} \L \S^{1/2}}_F^2 ={}& \ordlog{G^{-1}}
    & \textrm{ by \assuref{high_dim_resid_assu} \itemref{l_matrix_frob}, \lemref{m_rescaling}} \\
\frac{1}{G^2}  \m^\trans \L \m \le{}
\frac{1}{\sqrt{G}}\sqrt{ \frac{1}{G}\norm{\S^{1/2} \L \S^{1/2}}_F^2 } ={}& \ordlog{G^{-1/2}}.
    & \textrm{ by \assuref{high_dim_resid_assu} \itemref{l_matrix_frob}, \lemref{m_rescaling}} \\
\end{align*}
Therefore
\begin{align*}
    \expect{\xfdist}{\abs{
        \frac{1}{N^2} \sum_{n \ne m} \zbar_n^\trans N \L \zbar_m
        }} = \ordlog{G^{-1/2}}.
\end{align*}
\end{proof}


\begin{lem}\lemlabel{high_dim_rho_moments}
Under the conditions of \thmref{resid_divergence},
\begin{align*}
    \expect{\fdist(\x_n)}{\rho_{nn}} =
    \expect{\fdist(\x_n)}{\zbar_n^\trans \L \zbar_n} =
    \kappa + \ordlog{G^{-1/2}}
    \quad\textrm{and}\quad
    \expect{\fdist(\x_n)}{\rho_{nn}^2} =
    \expect{\fdist(\x_n)}{\left( \zbar_n^\trans \L \zbar_n \right)^2}
    = \ordlog{1}.
\end{align*}
\end{lem}

\begin{proof}
We have
\begin{align*}
\expect{\fdist(\x_n)}{\zbar_n^\trans \L \zbar_n} ={}&
    \trace{ \L  \expect{\fdist(\x_n)}{\zbar_n \zbar_n^\trans}}
\\={}&
    \trace{ \L  \left(\frac{1}{G} \S - \frac{1}{G^2} \m \m^\trans \right)}
\\={}&
    \frac{1}{G} \sumg \trace{\L_g \S_g} -
    \frac{1}{G^2} \m^\trans \L \m
\\={}&
    \frac{1}{G} \sumg \trace{\L_g \S_g} + \ordlog{G^{-1/2}}
\end{align*}
because
\begin{align*}
\abs{\frac{1}{G^2} \m^\trans \L \m}
\frac{1}{G} \sqrt{\norm{\S^{1/2}\L \S^{1/2}}_F^2} = \ordlog{1}
\frac{1}{\sqrt{G}} \sqrt{\frac{1}{G} \norm{\S^{1/2}\L \S^{1/2}}_F^2} = \ordlog{1}
\end{align*}
by \assuref{high_dim_resid_assu} \itemref{l_matrix_frob} and \lemref{m_rescaling}.


Next, 
\begin{align*}
\expect{\fdist(\x_n)}{\zbar_n^\trans \L \zbar_n \zbar_n^\trans \L \zbar_n}
    ={}&
\expect{\fdist(\x_n)}{
    \trace{
    \left(\z_n - \m / G\right)^\trans \L 
    \left(\z_n - \m / G\right) \left(\z_n - \m / G\right)^\trans \L 
    \left(\z_n - \m / G\right)
    }}.
\end{align*}
Expanding the previous expression will give integer multiples of 
the following terms, which are named according to the order
of $\z_n$ included in the expectation:
\begin{align*}
T_4 :=& \expect{\fdist(\x_n)}{\trace{
    \L \z_n \z_n^\trans \L \z_n \z_n^\trans
}} \\
T_3 :=& \frac{1}{G} \expect{\fdist(\x_n)}{\trace{
    \L \z_n \z_n^\trans \L \z_n \m^\trans
}} \\
T_{2} :=& \frac{1}{G^2} \expect{\fdist(\x_n)}{\trace{
    \L \z_n \z_n^\trans \L \m \m^\trans
}} \\
T_{1} :=& \frac{1}{G^3} \expect{\fdist(\x_n)}{\trace{
    \L \z_n \m^\trans \L \m \m^\trans
}} \\
T_{0} :=& \frac{1}{G^4} \trace{
    \L \m \m^\trans \L \m \m^\trans
}.
\end{align*}
(Note that $\z_n^\trans \L \m = \m^\trans \L \z_n$, so the $\z_n$ terms can
always be grouped together in $T_2$.) 

We analyze the terms one by one using
\assuref{high_dim_resid_assu} and \lemref{m_rescaling}:
\begin{align*}
T_0 ={}& \frac{1}{G} \left(\frac{1}{G \sqrt{G}} \m^\trans \L \m \right)^2
    ={} \ordlog{G^{-1}} & \textrm{\itemref{l_matrix_frob}, \lemref{m_rescaling}} \\
\end{align*}
\begin{align*}
T_{1} =& \frac{1}{G^3} \trace{
    \L \expect{\fdist(\x_n)}{\z_n} \m^\trans \L \m \m^\trans}
={}
\frac{1}{G^4} \trace{
    \L \m \m^\trans \L \m \m^\trans} = T_{0} = \ordlog{G^{-1}}.
\end{align*}

\begin{align*}
\abs{T_{2}} ={}& 
\abs{\frac{1}{G^2} \trace{
        \L \expect{\fdist(\x_n)}{\z_n \z_n^\trans} \L \m \m^\trans}}
\\={}&
\abs{\frac{1}{G^3} \trace{\L \S \L \m \m^\trans}}
\\\le{}&
\frac{1}{G^2} \sqrt{\norm{\S^{1/2} \L \S \L \S^{1/2}}_F^2}
\\\le{}&
\frac{1}{G^2} \norm{\S^{1/2} \L \S^{1/2}}_F^2
\\={}& \ordlog{G^{-1}}
\textrm{by \itemref{l_matrix_frob}}.
\end{align*}
\begin{align*}
T_{3} ={}& \frac{1}{G} \expect{\fdist(\x_n)}{\trace{
    \L \z_n \z_n^\trans \L \z_n \m^\trans
}} 
\\={}&
\frac{1}{G} \expect{\fdist(\x_n)}{
    \left(\sumg \a_{ng}^2 \y_n^\trans \L_{gg} \y_n \right)
    \left(\sumg \a_{ng} \y_n^\trans \L_{gg} \m_g \right)
}
\\={}&
\frac{1}{G} \expect{\fdist(\x_n)}{
    \sumg \a_{ng}^3 \left(\y_n^\trans \L_{gg} \y_n \right)
    \left(\y_n^\trans \L_{gg} \m_g \right)
}
&\textrm{(since }\a_{ng}\a_{nh} = 0\textrm{ if }g \ne h\textrm{)}
\\={}&
\frac{1}{G} 
    \sumg \expect{\fdist(\y_n \vert \a_{ng} = 1)}{
        \left(\y_n^\trans \L_{gg} \y_n \right)
        \left(\y_n^\trans \L_{gg} \m_g \right)
}
\end{align*}
We can then further simplify to
\begin{align*}
\abs{T_3} \le{}&
\frac{1}{G} 
    \sumg \expect{\fdist(\y_n \vert \a_{ng} = 1)}{
        \norm{\y_n}_2^3 }
    \norm{\m_g}_2 \norm{\L_{gg}}_F^2
    = \ordlogp{1}
    & \textrm{ by \assuref{high_dim_resid_moments}}.
\end{align*}
\begin{align*}
T_{4} ={}&
    \expect{\fdist(\x_n)}{\trace{
    \L \z_n \z_n^\trans \L \z_n \z_n^\trans
}} 
\\={}&
\expect{\fdist(\x_n)}{
    \left(\sumg \a_{ng}^2 \y_n^\trans \L_{gg} \y_n \right)^2
}
\\={}&
\expect{\fdist(\x_n)}{
    \sumg \a_{ng}^4 \left(\y_n^\trans \L_{gg} \y_n \right)^2
}
&\textrm{(since }\a_{ng}\a_{nh} = 0\textrm{ if }g \ne h\textrm{)}
\\={}&
\frac{1}{G} \sumg \expect{\fdist(\y_n \vert \a_{ng} = 1)}{
    \left(\y_n^\trans \L_{gg} \y_n \right)^2
}.
\end{align*}
We can then further simplify to
\begin{align*}
\abs{T_4} \le{}&
\frac{1}{G} 
    \sumg \expect{\fdist(\y_n \vert \a_{ng} = 1)}{
        \norm{\y_n}_2^4 } \norm{\L_{gg}}_F^2
        = \ordlogp{1}
        & \textrm{ by \assuref{high_dim_resid_moments}}.
\end{align*}
It follows that $\expect{\fdist(\x_n)}{\rho_{nn}^2} = \ordlog{1}$.

\end{proof}


\begin{lem}
Under the conditions of \thmref{resid_divergence},
$\resid{T}(1)  \plim \kappa$ and
$\sqrt{N}\abs{\resid{T}(1)} \rightarrow \infty$ if $\kappa \ne 0$.
\end{lem}
\begin{proof}
By \lemref{high_dim_resid_order},
\begin{align*}
    \expect{\xfdist}{\abs{\resid{T}(1) - \kappa}} \le
    \abs{\frac{1}{N} \sumn\left(\rho_{nn} - \kappa \right)} +
    \ordlog{G^{-1/2}} + \ordlog{N^{-1}}.
\end{align*}
Then
\begin{align*}
    \expect{\xfdist}{\abs{\frac{1}{N} \sumn\left(\rho_{nn} - \kappa \right)}}^2
    \le{}&
    \expect{\xfdist}{\left(\frac{1}{N} \sumn\left(\rho_{nn} - \kappa \right)\right)^2}
    \\={}&
    \frac{1}{N} \left(
        \expect{\fdist(\x_n)}{\rho_{nn}^2} - 
        \kappa \expect{\fdist(\x_n)}{\rho_{nn}} + \kappa^2
    \right)
    \\={}&
    \frac{1}{N} \left(
        \expect{\fdist(\x_n)}{\rho_{nn}^2} - 
        \kappa \left( \expect{\fdist(\x_n)}{\rho_{nn}} - \kappa \right)
    \right)
    \\={}& 
    \frac{1}{N} \left(\ordlog{1} + \kappa \ordlog{G^{-1/2}}\right)
    &\textrm{by \lemref{high_dim_rho_moments}}.
    \\={}& 
    \ord{N^{-1}}.
\end{align*}
Since convergence in expectation implies convergence in probability,
$\resid{T}(1) \plim \kappa$.

Finally, assume that $\kappa \ne 0$.  For any $M > 0$,
take $N \ge (2 M / \kappa)^2$, so
\begin{align*}
    \p\left( \sqrt{N} \abs{\resid{T}(1)} \le M\right) ={}&
    \p\left( -\frac{M}{\sqrt{N}} - \kappa \le 
        \resid{T}(1) - \kappa \le 
        \frac{M}{\sqrt{N}} - \kappa\right)
    \\\le{}&
    \p\left( -\abs{\kappa / 2} - \kappa \le 
        \resid{T}(1) - \kappa \le 
        \abs{\kappa / 2} - \kappa\right) 
        \rightarrow 0,
\end{align*}
where the final line follows from $\resid{T}(1) - \kappa \plim 0$, and the
fact that $\kappa > 0 \rightarrow \abs{\kappa / 2} - \kappa \le -\abs{\kappa / 2}$,
and $\kappa < 0 \rightarrow -\abs{\kappa / 2} - \kappa \ge \abs{\kappa / 2}$.
\end{proof}


\begin{lem}\lemlabel{m_rescaling}
For a symmetric matrix $A$, 
$\frac{1}{G} \abs{m^\trans A m} \le  \sqrt{\norm{\S^{1/2} A \S^{1/2}}_F^2}.$
\begin{proof}
\def\U{U}
Write the eigenvalue decomposition of $\S$ as
\begin{align*}
    \S = \U \diag{\sigma(\S)_1, \ldots, \sigma(\S)_G} \U^\trans
\end{align*}
where $\U$ is orthonormal and $\sigma(\S)_1$ are the eigenvalues.
Let $\U_k$ denote the $k$--th eigenvector, i.e. the $k$--th column 
of $\U$.  We can write the square root and pseudo--inverse square root as
\begin{align*}
    \S^{1/2} ={}& \U \diag{\sigma(\S)_1^{1/2}, \ldots, \sigma(\S)_G^{1/2}} \U^\trans \\
    \S^{\dagger/2} ={}& \U \diag{\sigma(\S)_1, \ldots, \sigma(\S)_G}^{\dagger / 2} \U^\trans
\end{align*}
where the $k$--diagonal of the pseudo-inverse 
$\diag{\sigma(\S)_1, \ldots, \sigma(\S)_G}^{\dagger / 2}$ is $0$ if $\sigma(\S)_k= 0$
and $\sigma(\S)_k^{-1/2}$ otherwise.

We first show that $\S^{1/2} \S^{\dagger/2} m = m$. Note that if $\sigma(\S)_k =
0$, then $\U_k^\trans m = 0$, because
\begin{align*}
    0 = \sigma(\S)_k = \U_k^\trans \S \U_k = 
    \U_k^\trans \left( \cov{\fdist(\x_n \vert \a_{ng})}{\y_n} + m m^\trans \right) \U_k \ge
    \U_k^\trans m m^\trans \U_k.
\end{align*}
Therefore,
\begin{align*}
\norm{\S^{\dagger/2} m}_2^2 ={}&
    m^\trans \U \diag{\sigma(\S)_1, \ldots, \sigma(\S)_G}^{\dagger / 2} 
                \diag{\sigma(\S)_1, \ldots, \sigma(\S)_G}^{\dagger / 2} \U^\trans m 
\\={}&
\sum_{k: \sigma(\S)_k > 0} \frac{m^\trans \U_k \U_k^\trans m}{\sigma(\S)_k}
\\={}&
\sum_{k: \sigma(\S)_k > 0} 
    \frac{m^\trans \U_k \U_k^\trans m}
            {\U_k^\trans(\cov{\fdist(\x_n \vert \a_{ng})}{\y_n} + m m^\trans)\U_k}
\\\le{}& G.
\end{align*}
Thus we have that $m = S^{1/2} \S^{\dagger/2} m$ and
$\frac{1}{G} \norm{\S^{\dagger/2} m}  \le 1$.  
The final result then follows from
\begin{align*}
\frac{1}{G} m^\trans A m 
={}& 
    m^\trans \S^{\dagger/2} S^{1/2} A S^{1/2} \S^{\dagger/2} m 
\\\le{}&
    \norm{\S^{\dagger/2} m}_2^2 \max_k \abs{\sigma_k(S^{1/2} A S^{1/2})}
\\\le{}&
    \max_k \sqrt{\sigma_k(S^{1/2} A S^{1/2})^2}
\\\le{}&
        \sqrt{\sum_k \sigma_k(S^{1/2} A S^{1/2})^2}
\\={}&
        \sqrt{\norm{S^{1/2} A S^{1/2}}_F^2}.
\end{align*}
Here, we have used the fact that the absolute value of the maximum eigenvalue of
a matrix is upper bounded by the square root of its Frobenius norm.
\end{proof}

\end{lem}

\subsection{Intuition from V--statistics}\applabel{v_stats}


To provide some additional intuition for \thmref{finite_dim_resid}, it is
helpful to heuristically imagine even stronger assumptions which allow us to
study $\resid{T}(\t)$ through the lens of V--statistics.  
Referring to \eqref{v_def}, observe that the term $\htil(\x_n, \x_m)$ depends on
$\x_n$ and $\x_m$ both through $\ellbarbar(\x_n \vert \theta)\ellbarbar(\x_m
\vert \theta)$ and through the posterior $\postft$.  If we neglect the
dependence through $\postft$, then $\sqrt{N} \resid{T}(\t)$ takes the form of a
V--statistic, with
\begin{align*}
\expect{\fdistnm}{\htil(\x_n, \x_m)} =
\expect{\fdist(\x_n \vert \x_m)}{\htil(\x_n, \x_m)} =
\expect{\fdist(\x_m \vert \x_n)}{\htil(\x_n, \x_m)} = 0.
\end{align*}
If we further assume that all variances are finite, and that
\assuref{postg_conc} applies, then
\begin{align*}
\htil(\x_n, \x_m) =
N^2
\expect{\postft}{
    \gbar(\theta)
    \ellbarbar(\x_n \vert \theta)
    \ellbarbar(\x_m \vert \theta)} = \ordlogp{1},
\end{align*}
with error square summable in $n$ and $m$.
We can then heuristically apply standard results for V--statistics.
Specifically, in the notation of Section 5.2.1 of
\citet{serfling:1980:approximation}, $\resid{T}(\t)$ is approximately V-statistic with
(Serfling's) $m=2$, $\zeta_1 = 0$, and $\zeta_2 ={} \expect{\fdistnm}{\htil[1](\x_n, \x_m)^2}$.
By Lemma A of the same section, the variance of the corresponding U-statistic is
\begin{align*}
    U :={}& \frac{1}{N(N-1)} \sumn \sum_{m\ne n}^N \htil[1](\x_n, \x_m)
    &
    \var{\fdist(\x_n) \fdist(\x_m)}{U} = 
    \frac{2 \expect{\fdistnm}{\htil[1](\x_n, \x_m)^2}}{N (N-1)}.
\end{align*}
Then the difference between the V--statistic and U--statistic is
order $\frac{1}{N^2} \sumn \htil(\x_n, \x_n) = \ordlogp{N^{-1}}$:
\begin{align*}
    \frac{1}{N^2} \sumn \htil(\x_n, \x_n) = 
        \ord{
            \sqrt{
                \frac{\expect{\fdist(\x_n)}{\htil(\x_n, \x_n)^2}}
                        {N^2}
                }
        } =
        \ord{
            \sqrt{
                \frac{\expect{\fdist(\x_n)}{\htil(\x_n, \x_n)^2}}
                        {N^2}
                }
        } = \ord{N^{-1}}.
\end{align*}
Similarly, the variance of the U--statistic is $\ordlogp{N^{-1}}$:
\begin{align*}
\var{\fdistnm}{\sqrt{N} \resid{T}(\t)} \approx 
&
N \ordlogp{ \frac{\expect{\fdistnm}{\htil(\x_n, \x_m)^2}}{N^2}} =
\ordlogp{N^{-1}}.
\end{align*}
Though the assumptions needed to formally justify the preceding argument would be
stronger than what is needed for \thmref{finite_dim_resid}, the V--statistic
perspective shows clearly the role of the $N^{-2}$ rate of convergence of
posterior expectations of three centered quantities given by
\assuref{postg_conc}.  In the next section,
we will see that, in high--dimensional models, we lose this $N^{-2}$
rate, and the residual will fail to vanish in general as a consequence.

\subsection{Details for \exref{poisson_re_resid}}\applabel{poisson_re_resid}

In this section, we provide more details for \exref{poisson_re_resid}.


Recall that the Gamma density is given by
\begin{align*}
\log \Gamma(\lambda \vert \alpha, \beta) ={}&
    (\alpha - 1) \log \lambda - \beta \lambda + \alpha \log \beta - \log \Gamma(\alpha) \\
A(\alpha, \beta) ={}& -\alpha \log \beta + \log \Gamma(\alpha) \\
\frac{\partial A}{\partial \alpha} ={}& 
    \expect{}{\log \lambda} = \mathrm{Digamma}(\alpha) - \log \beta \\
\frac{\partial A}{\partial (-\beta)} ={}& 
    \expect{}{\lambda} = \alpha / \beta \\
\frac{\partial^2 A}{\partial \alpha^2} ={}& 
    \cov{}{\log \lambda} = \mathrm{Trigamma}(\alpha) \\
\frac{\partial A}{\partial (-\beta)^2} ={}& 
    \cov{}{\lambda} = \alpha / \beta^2 \\
\frac{\partial^2 A}{\partial \alpha \partial (-\beta)} ={}& 
    \cov{}{\log \lambda, \lambda} = 1 / \beta.
\end{align*}

The posterior $\postlf$ is given by
\begin{align*}
\expect{\fdist(\x_n)}{\ell(\x_n \vert \gamma, \lambda)} ={}& 
    \frac{1}{G}  \sumg \left( \rho_g (\log(\gamma) + \log(\lambda_g)) - \gamma \lambda_g\right) \\
\expect{\fdist(\xvec)}{\ell(\xvec \vert \gamma, \lambda)} ={}& 
    \frac{N}{G}  \sumg \left( \rho_g (\log(\gamma) + \log(\lambda_g)) - \gamma \lambda_g\right) \\
\p\left( \lambda_g \vert \fdist, N, \gamma\right) ={}& 
    \textrm{Gamma}\left(
        \alpha + \frac{N}{G} \rho_g, 
        \beta + \frac{N}{G} \gamma \right) \\
\cov{\postlf}{\eta(\gamma, \lambda_g), \eta(\gamma, \lambda_g)^\trans} ={}&
    \begin{pmatrix}
        \mathrm{Trigamma}(\alpha + \frac{N}{G} \rho_g) & 
        -\gamma / \left(\beta + \frac{N}{G} \gamma\right)\\ 
        -\gamma / \left(\beta + \frac{N}{G} \gamma\right) & 
        \gamma^2 \frac{\alpha + \frac{N}{G} \rho_g}{\left( \beta + \frac{N}{G} \gamma \right)^2}\\ 
    \end{pmatrix} \\
\expect{\postlf}{\eta(\gamma, \lambda_g)} ={}&
    \begin{pmatrix}
        \log(\gamma) + 
            \mathrm{Digamma}(\alpha + \frac{N}{G} \rho_g) - 
            \log \left( \beta + \frac{N}{G} \gamma\right) \\ 
        - \gamma \frac{\alpha + \frac{N}{G} \rho_g}{\beta + \frac{N}{G} \gamma}
    \end{pmatrix}\\
\end{align*}
Let
\begin{align*}
    \rho_g :={}& \expect{\fdist(\y_n \vert \a_{ng}=1)}{\r_n}
    &
    \v_g :={}& \var{\fdist(\y_n \vert \a_{ng}=1)}{\r_n}.
\end{align*}
We then have
\begin{align*}
\m_g ={}& \expect{\fdist(\y_n \vert \a_{ng}=1)}{\y_n} = 
    \begin{pmatrix}
        \rho_g \\ 
        1
    \end{pmatrix}
&
\S_g ={}&
    \expect{\fdist(\y_n \vert \a_{ng}=1)}{\y_n \y_n^\trans} = 
    \begin{pmatrix}
        \v_g + \rho_g^2 & \rho_g \\ 
        \rho_g & 1
    \end{pmatrix}.
\end{align*}
The fact that $\fdist$ may be misspecified means the values for $\rho_g$ and
$\v_g$ are arbitrary.  A correctly specified model would take them to be an
infinite sequence with $\rho_g = \v_g$, drawn from a $\mathrm{Gamma}(\alpha,
\beta)$ distribution and then conditioned on as we take $G \rightarrow \infty$.

\def\digamma{\Psi}
\def\trigamma{\Psi_1}
\def\noverg{N_G}

Let $\noverg := N / G$ denote the average number of observations per group.
By standard properties of the Gamma distribution (see \appref{poisson_re_resid}),
we can evaluate $\postlf$ and derive
\begin{align*}
\mu_g(\gamma) ={}&
    \begin{pmatrix}
        \log(\gamma) + 
            \digamma(\alpha + \noverg \rho_g) - 
            \log \left( \beta + \noverg \gamma\right) \\ 
        - \gamma \frac{\alpha + \noverg \rho_g}{\beta + \noverg \gamma}
    \end{pmatrix}
    &
\J_{gg}(\gamma) ={}&
    \begin{pmatrix}
        \trigamma(\alpha + \noverg \rho_g) & 
        -\frac{\gamma}{\beta + \noverg \gamma} \\
        -\frac{\gamma}{\beta + \noverg \gamma} &
        \gamma^2 \frac{\alpha + \noverg \rho_g}{\left( \beta + \noverg \gamma \right)^2}\\ 
    \end{pmatrix},
\end{align*}
where $\digamma$ and $\trigamma$ are the digamma and trigamma functions, respectively,
and $\J_{gh}(\gamma) ={} 0$ for $g \ne h$ since $\lambda_g$ are a posteriori independent
given $\gamma$.

Let us see how $\M_{gh}$ and $\L_{gh}$ behave in this case.  
As argued in \exref{re_gamma}, $\expect{\postgf}{\gamma} \rightarrow \gamma_0$
for some $\gamma_0$ determined by the sequence of $\rho_g$ and $\v_g$, as long
as the sequence is sufficiently regular.  Furthermore, by
\assuref{postg_conc}, we can evaluate
\begin{align*}
\mubar_{g\,(1)}(\gamma_0) ={}&
    \begin{pmatrix}
        \gamma_0^{-1} - \frac{\noverg}{\beta + \noverg \gamma_0} \\ 
        - \frac{\alpha + \noverg \rho_g}{\beta + \noverg \gamma_0}
        + \noverg \gamma_0 
            \frac{\alpha + \noverg \rho_g}{\left(\beta + \noverg \gamma_0\right)^2}
    \end{pmatrix} \\
\J_{gg\,(1)}(\gamma_0) ={}&
\begin{pmatrix}
    0 & 
    \frac{\noverg \gamma_0}{\left( \beta + \noverg \gamma\right)^2} 
    -\frac{1}{\beta + \noverg \gamma_0} 
    \\
    \frac{\noverg \gamma_0}{\left( \beta + \noverg \gamma\right)^2} 
    -\frac{1}{\beta + \noverg \gamma_0} 
    &
    2 \gamma_0 \frac{\alpha + \noverg \rho_g}{\left( \beta + \noverg \gamma_0 \right)^2} -
    2 \noverg \gamma_0^2 \frac{\alpha + \noverg \rho_g}{\left( \beta + \noverg \gamma_0 \right)^3}
        \\ 
\end{pmatrix}.
\end{align*}
From these expressions we can see that the assumptions in
\assuref{high_dim_resid_assu} are reasonable.
For example,
\begin{itemize}
\item Irrespective of $\rho_g$ and $\v_g$, there will typically be $\ydim^2 = 4$
    nonzero entries in each $\M_{gh}$, and so $\sum_{g=1}^G \sum_{h=1}^G \norm{
    \S_g^{1/2} \M_{gh} \S_h^{1/2}}_F^2$ has $4 G^2$ nonzero
    entries, as assumed in \assuref{high_dim_resid_assu} \itemref{m_matrix_frob}.
\item Similarly, $\J_{gh} = 0$ for $g \ne h$, and $\J_{gg}$ has $3$ nonzero
    entries so $\sum_{g=1}^G \sum_{h=1}^G \norm{\S_g^{1/2} \L_{gh}
    \S_h^{1/2}}_F^2$ has $3 G$ nonzero entries, which matches the rate assumed
    in \assuref{high_dim_resid_assu} \itemref{l_matrix_frob}.
\item Given these, whether \assuref{high_dim_resid_assu} is in fact satisfied as
    $G \rightarrow \infty$ depends on the tail behavior of the sequence of
    $\rho_g$ and $\v_g$.  For example, it is easy to see that, if $\rho_g$ and
    $\v_g$ are bounded, then \assuref{high_dim_resid_assu} is satisfied.
\end{itemize}
%






\ifbool{jrssb}{
    \end{appendices}
}

\end{document}